\documentclass[a4paper,USenglish,cleveref, autoref, thm-restate, numberwithinsect]{lipics-v2021}

\hideLIPIcs  %

\nolinenumbers

\title{Independence and Domination \\ on Bounded-Treewidth Graphs:\\
	Integer, Rational, and Irrational Distances}

\titlerunning{Independence and Domination on Bounded-Treewidth Graphs}

\author{Tim A. {Hartmann}}{CISPA Helmholtz Center for Information Security, 
Germany}{tim.hartmann@cispa.de}{https://orcid.org/0000-0002-1028-6351}{}

\author{D\'{a}niel Marx}{CISPA Helmholtz Center for Information Security, 
Germany}{marx@cispa.de}{https://orcid.org/0000-0002-5686-8314}{}

\authorrunning{Tim A.\ {Hartmann} and D\'{a}niel Marx}

\Copyright{Tim A.\ {Hartmann} and D\'{a}niel Marx}



\category{} %

\relatedversion{} %

\EventEditors{John Q. Open and Joan R. Access}
\EventNoEds{2}
\EventLongTitle{42nd Conference on Very Important Topics (CVIT 2016)}
\EventShortTitle{CVIT 2016}
\EventAcronym{CVIT}
\EventYear{2016}
\EventDate{December 24--27, 2016}
\EventLocation{Little Whinging, United Kingdom}
\EventLogo{}
\SeriesVolume{42}
\ArticleNo{23}

\usepackage[utf8]{inputenc}
\usepackage[table]{xcolor}
\usepackage[textsize=tiny,textwidth=2cm]{todonotes}

\usepackage{latexsym}
\usepackage{amssymb}

\usepackage{bibentry}

\usepackage{tikz}
\usetikzlibrary{shapes}
\usetikzlibrary{calc}
\usetikzlibrary{patterns}
\usetikzlibrary{fit}

\tikzstyle{invertex} = [trapezium, scale=1.5,rotate=-90,draw]
\tikzstyle{outvertex} = [trapezium, scale=1.5,rotate=-90,draw]
\tikzstyle{inoutvertex} = [trapezium, scale=1.5,rotate=-90,draw]

\usepackage{amsmath}
\usepackage{amsthm}
\usepackage{thmtools}
\usepackage{mathtools} %
\usepackage{xspace} %
\usepackage{dsfont} %
\usepackage{tabularx} %
\usepackage{diagbox} %

\usepackage{hyperref}
\usepackage[all]{hypcap}
\hypersetup{hidelinks}

\usepackage{cleveref}
\renewcommand{\autoref}[1]{\cref{#1}}

\newcommand{\name}[1]{\textsc{#1}\xspace}

\newcommand{\problem}[1]{\textsc{#1}}

\newcommand{\dispersion}{\textsc{Dispersion}\xspace}
\newcommand{\disp}[1]{#1\text{-}\mbox{disp}}

\newcommand{\covering}{\problem{Covering}\xspace}
\newcommand{\cover}[1]{#1\text{-}\mbox{cover}}

\crefname{equation}{Equation}{Equations}

\newcommand{\Oh}{{O}}

\newcommand{\problemm}[1]{\textsc{#1}\xspace}

\newcommand{\fpt}{\name{FPT}}
\newcommand{\p}{\name{P}}
\newcommand{\np}{\name{NP}}

\newcommand{\wone}{\name{W[1]}}
\newcommand{\wtwo}{\name{W[2]}}

\newcommand{\xp}{\name{XP}}

\newcommand{\ETH}{\name{ETH}}

\newcommand{\define}[1]{\textit{#1}}

\newcommand{\backward}{(\( \Leftarrow \))\xspace}
\newcommand{\forward}{(\( \Rightarrow \))\xspace}

\newcommand{\problemdefSimple}[3]{
\medskip
\begin{tabularx}{\textwidth-0.5cm}{ r X p{0.5cm} }
\multicolumn{2}{l}{#1} \\
Input: & #2 \\	
Question: & #3
\end{tabularx}
\medskip
}

\makeatletter
\newcommand{\labeltext}[3][]{%
    \@bsphack%
    \csname phantomsection\endcsname%
    \def\tst{#1}%
    \def\labelmarkup{}%
    \def\refmarkup{}%
    \ifx\tst\empty\def\@currentlabel{\refmarkup{#2}}{\label{#3}}%
    \else\def\@currentlabel{\refmarkup{#1}}{\label{#3}}\fi%
    \@esphack%
    \labelmarkup{#2}%
}
\makeatother

	\DeclarePairedDelimiter{\ceil}{\lceil}{\rceil}
	\DeclarePairedDelimiter{\floor}{\lfloor}{\rfloor}

\newcommand{\devcomment}[1]{COMMENT}

\usepackage{caption}
\usepackage{subcaption}

\newcommand{\commentt}[1]{}

\newcommand{\refstar}{${(\star)}$}

\newcommand{\pw}{\operatorname{{\mathsf{pw}}}}
\newcommand{\tw}{\operatorname{{\mathsf{tw}}}}

\newcommand{\vc}{{\tau}}

\usepackage{pgfplots}
\pgfplotsset{%
    ,compat=1.12
    ,every axis x label/.style={at={(current axis.right of origin)},anchor=north west}
    ,every axis y label/.style={at={(current axis.above origin)},anchor=north east}
    }

\newcommand{\NNN}{\mathbb{N}}

\newcommand{\half}{\tfrac{1}{2}}

\usepackage{fixmath} %
\usepackage{graphics} %

\newcommand{\indset}{\problemm{Independent Set}}
\newcommand{\domset}{\problemm{Dominating Set}}

\makeatletter
\newcommand*{\textoverline}[1]{$\overline{\hbox{#1}}\m@th$}
\makeatother

\newcommand{\AAA}{\mathcal{A}}

\newcommand{\FF}{\mathcal{F}}
\newcommand{\GG}{\mathcal{G}}

\newcommand{\LL}{\mathcal{L}}

\newcommand{\SSS}{\mathcal{S}}
\newcommand{\TT}{\mathcal{T}}
\newcommand{\UU}{\mathcal{U}}

\newcommand{\dispersionstar}{\textsc{Dispersion}^{\star}\xspace}

\newcommand{\dombarr}{walk dominating\xspace}  %
\newcommand{\dombar}[1]{${#1}$-\dombarr}
\newcommand{\dombarsetvoid}{\dombarr set\xspace} 
\newcommand{\dombarset}[1]{${#1}$-walk dominating set} %
\newcommand{\dombarsets}[1]{${#1}$-\dombarr sets}
\newcommand{\dombares}[1]{${#1}$-walk dominates\xspace}
\newcommand{\dombared}[1]{${#1}$-walk dominated\xspace}

\newcommand{\DS}{\DomSet}
\newcommand{\DomSet}{\problem{Walk Dominating Set}\xspace}
\newcommand{\IndSet}{{\problem{Independent Set}}\xspace} %

\newcommand{\ColorfulClique}{\problem{Colorful Clique}\xspace}
\newcommand{\GridTiling}{$\leq$-\problem{Grid Tiling}\xspace}

\newcommand{\vv}{\text{v}}

\newcommand{\ogamma}{\overline{\gamma}}

\newcommand{\dr}{\mathsf{dr}}

\newcommand{\SETH}{\textsc{SETH}\xspace}
\newcommand{\ppseth}{\textsc{PWSETH}\xspace}

\newcommand{\ahalf}{\tfrac{a}{2}}

\tikzstyle{vertex}=[circle,draw,scale=0.5,fill=white]

\tikzstyle{superedge}=[ultra thick]

\newcommand{\ineditnote}[1]{}
\newcommand{\dm}[1]{\textcolor{red}{DM: #1}}

\begin{document}

\maketitle 

\begin{abstract}
  The distance-$d$ variants of \textsc{Independent Set} and \textsc{Dominating Set} problems have been extensively 
  studied from different algorithmic viewpoints. In particular, the complexity of these problems are well understood 
  on bounded-treewidth graphs [Katsikarelis, Lampis, and Paschos, \textit{Discret.\ Appl.\ Math} 2022][Borradaile and 
  Le, IPEC 2016]: given a tree decomposition of width $t$, the two problems can be solved in time $d^t\cdot 
  n^{O(1)}$ and $(2d+1)^t\cdot n^{O(1)}$, respectively. Furthermore, assuming the Strong Exponential-Time Hypothesis 
  (SETH), the base constants are best possible in these running times: they cannot be improved to $d-\varepsilon$ 
  and $2d+1-\varepsilon$, respectively,  for any $\varepsilon>0$. We investigate continuous versions of these 
  problems in a setting introduced by Megiddo and Tamir [SICOMP 1983], where every edge is modeled by a unit-length 
  interval of points. In the $\delta$-\dispersion problem, the task is to find a maximum number of points (possibly 
  inside edges) that are pairwise at distance at least $\delta$ from each other. Similarly, in the 
  $\delta$-\covering problem, the task is to find a minimum number of points (possibly inside edges) such that every 
  point of the graph (including those inside edges) is at distance at most $\delta$ from the selected point set. We 
  provide a comprehensive understanding of these two problems on bounded-treewidth graphs.
  \begin{enumerate}
  \item Let $\delta=a/b$ with $a$ and $b$ being coprime.
  If $a\le 2$, then  $\delta$-\dispersion is polynomial-time solvable. For $a\ge 3$, given a tree decomposition of 
  width $t$, the problem can be solved in time $(2a)^t\cdot n^{O(1)}$, and, assuming SETH, there is no 
  $(2a-\varepsilon)^t\cdot n^{O(1)}$ time algorithm for any $\varepsilon>0$.
  \item Let $\delta=a/b$ with $a$ and $b$ being coprime.
If $a = 1$,
	then $\delta$-\covering is polynomial-time solvable.
For $a\ge 2$, given a tree decomposition of width $t$,
	the problem can be solved in time $((2+2(b\bmod 2)) a)^t\cdot n^{O(1)}$,
	and, assuming SETH, there is no $((2+2(b\bmod 2))a -\varepsilon)^t\cdot n^{O(1)}$ time algorithm for any 
	$\varepsilon>0$.
 \item For every fixed irrational number $\delta>0$ satisfying some mild computability condition, both 
 $\delta$-\dispersion and $\delta$-\covering can be solved in time $n^{O(t)}$ on graphs of treewidth $t$.
We show a very explicitly defined irrational number
$\delta = (4\sum_{j =1}^{\infty} 2^{-2^j} )^{-1} \allowbreak \approx 0.790085$
 such that  $\delta$-\dispersion and $\delta/2$-\covering  are W[1]-hard parameterized by the treewidth $t$ of the 
 input graph, and, assuming ETH, cannot be solved in time $f(t) \cdot n^{o(t)}$.
 \end{enumerate}

 As a key step in obtaining these results, we extend earlier results on distance-$d$ versions of 
 \textsc{Independent Set} and \textsc{Dominating Set}: We determine the exact complexity of these problems in the 
 special case when the input graph arises from some graph $G'$ by subdividing every edge exactly $b$ times.
\end{abstract}

\newpage

\tableofcontents

\newpage

\section{Introduction}
\label{section:intro}

An independent set of a graph $G$ is a subset of vertices
	that have pairwise distance at least $2$.
A well-known generalization to higher distance is the notion of \emph{$a$-independent set} 
	for some integer $a$, which
	is a subset of vertices
	that have pairwise distance at least $a$.
Receiving extensive attention in the literature,
	e.g.~\cite{EtoGM14,MontealegreT16,EtoILM17,bacso2019,PilipczukSiebertz2021,KatsikarelisLP23},
	the problem seems reasonably well understood. The dual notion of \emph{distance-$d$ dominating set}, which is a 
	set $D$ of vertices such that every vertex of the graph is at distance at most $d$ from $S$, was also similarly 
	well studied. In this paper, we present an extensive study of both problems, focusing on their complexity on 
	subdivided and  bounded-treewidth graphs. Furthermore, we explore the generalization of these problems to 
	noninteger (even irrational!) distances in an appropriate continuous model \cite{Dearing1974,Shier1977} that 
	received renewed attention lately \cite{FreiGHHM2024,GrigorievHLW21,HartmannJanssen2024,HartmannL22,HartmannLW22}.

        \subsection{Independent Set and Dispersion}

        \subparagraph{Integer distances} Finding a maximum $a$-independent set is NP-hard for $a=2$ (as it is the 
        same as the classic \IndSet problem) and it is not difficult to show that it remains NP-hard for any fixed 
        $a\ge 2$.
In contrast, there are polynomial time algorithms
	for $a \in \{2,4\}$ when the input is a \emph{$2$-subdivided graph},
	that is a graph $G$ that resulted from a graph $G'$ by replacing every edge by a path of length two.

\begin{theorem}[Grigoriev et al.~\cite{GrigorievHLW21}]
\label{lemma:disp:p}
For $a\in\{2,4\}$, a maximum $a$-independent set on a $2$-subdivided graph
	can be found in linear time.\footnote{
	The original statement is about a continuous dispersion problem,
	but can be put as above using a connection of these two problem which we mention later.}
\end{theorem}

Due to an important connection to the dispersion problem (see later in this section), we are particularly interested in the complexity of finding a maximum $a$-independent set on subdivided graphs, where every edge is replaced by a path of length $b$. Formally, let $\alpha_a(G)$ be the maximum cardinality of an $a$-independent set of a graph $G$.
For a graph class $\GG$,
	let $a$-$\IndSet(\GG)$ be the corresponding decision problem,
	which is, given a graph $G \in \GG$ and integer $k$,
	deciding whether $\alpha_a(G) \geq k$.
Let $\GG_b$ be class of graphs $G$ that are \emph{$b$-subdivisions},
	meaning $G$ results from a graph $G'$ by replacing every edge by a path of length~$b$.

       As a first contribution,  for every fixed integer $a,b$, 
	we settle the NP-hardness and the parameterized complexity of finding a maximum $a$-independent set
	when parameterized by the solution size (as color-coded in \cref{figure:ind:cells}).
If the ratio $\frac{a}{b}$ is smaller than $2$,
	then the problem is \fpt, and otherwise
	it is \wone-hard unless it is a polynomial time solvable case.

\commentt{
\begin{theorem}[\Cref{section:ind:parameterized}]
\label{i:lemma:ind:param}
Let $a,b \in \NNN_+$.
$a$-$\IndSet(\GG_b)$ is polynomial time solvable
	if $b$ is a multiple of $a$,
	or $\frac{a}{2}$ is even and $b$ is an odd multiple of $\frac{a}{2}$;
	else is \np-hard and, if $\frac{a}{b}<2$ or $\frac{a}{b}=2$, $b$ even,
	is \fpt for the solution size as parameter,
	else is \wone-hard.\dm{If we have the space, would be nice the rephrase to make it easier to parse.}
\end{theorem}
}

\newcommand{\lemmaTextIndParam}{$a$-$\IndSet(\GG_b)$ is
\begin{itemize}
\item
polynomial time solvable if $b=ca$ or if $b=c\frac{a}{2}$ and $\frac{a}{2}$ is even for some integer $c$;
	and \np-hard for all other integers $a,b$; and is
\item
fixed-parameter tractable for the solution size as parameter
	if $\frac{a}{b}<2$ or if $\frac{a}{b}=2$, $b$ even;
	and \wone-hard for all other integers $a,b$.
\end{itemize}}

\begin{theorem}[\Cref{section:ind:parameterized}]
\label{i:lemma:ind:param}
\lemmaTextIndParam
\ineditnote{\cref{lemma:ds:param}}
\end{theorem}

\sloppy
Next, we consider the problem parameterized by treewidth. 
Intuitively, the $a$-$\IndSet$ problem is harder for larger $a$.
Indeed, for all $a\geq 2$, there is a matching upper and lower bound
	with $a$ in the base of an exponential run time
	for graphs of bounded treewidth,
	assuming the Strong Exponential Time Hypothesis (\SETH).

\begin{theorem}[Katsikarelis et al.~\cite{KatsikarelisLP2022}]
\label{lemma:a:is:tw}
For $a\geq 2$,
	given a tree decomposition of width $t$ of an $n$ vertex input graph,
	a maximum $a$-independent set can be found in time $a^t \cdot {n}^{O(1)}$.
Assuming \SETH,
	there is no $(a-\varepsilon)^{\tw(G)} \cdot n^{O(1)}$ time algorithm for any $\varepsilon> 0$,
	even for graphs without a cycle of length $<a$.\footnote{The restriction to graphs without short cycles is not 
	explicitly given, but easily observed.
	We will rely on this restriction later.}
\end{theorem}

We refine Theorem~\ref{lemma:a:is:tw} by restricting the $a$-\IndSet problem to $b$-subdivided graphs and 
determining the optimal base of the exponent for 
all integers $a$ and $b$. We expect that larger $a$ makes the problem harder (as in Theorem~\ref{lemma:a:is:tw}) and larger $b$ makes the problem easier (as the graphs become more restricted), but it turns out that the optimal base depends on $a$ and $b$ in a very subtle way. 
Let $\gcd(a,b)$ denote the greatest common divisor of integers $a$ and $b$.

\newcommand{\aaa}{a'}
\newcommand{\bbb}{b'}
\newcommand{\aac}{a}
\newcommand{\bbc}{b}
\newcommand{\lemmaTextIsTwSummaryCoreItems}{
	\item
	If $\gcd(\aaa,\bbb)$ is odd:
If $\aac=1$,
		$\aaa$-$\IndSet(\GG_{\bbb})$ is in $\p$,
	else
	$\aaa$-$\IndSet(\GG_{\bbb})$ can be solved in time $\aac^{t} \cdot n^{O(1)}$
	but not in $(\aac-\varepsilon)^{\pw(G)} \cdot n^{O(1)}$.
	\item
	If $\gcd(\aaa,\bbb)$ is even:
	If $\aac \in \{1,2\}$,
		$\aaa$-$\IndSet(\GG_{\bbb})$ is in \p,
else
	$\aaa$-$\IndSet(\GG_{\bbb})$ can be solved in time $(2\aac)^{t} \cdot n^{O(1)}$
	but not in $(2\aac-\varepsilon)^{\pw(G)} \cdot n^{O(1)}$.}
\newcommand{\lemmaTextIsTwSummary}{
Let $\aaa,\bbb$ integers with $\gcd(\aaa,\bbb)=c$, $c\aac=\aaa$ and $c\bbc=\bbb$.
Assume \SETH, an $\varepsilon>0$ and that a tree decomposition of width $t$ is part of the input.
\begin{itemize}
	\lemmaTextIsTwSummaryCoreItems
\end{itemize}
}
\newcommand{\lemmaTextIsTwSummaryExtended}{
Let integers $\aaa,\bbb$ define $\gcd(\aaa,\bbb)=c$ and $c\aac=\aaa$ and $c\bbc=\bbb$.
Let $n$ be the number of vertices of the input graph.
Assume \SETH, an $\varepsilon>0$ and that a tree decomposition of width $t$ is part of the input.
\begin{itemize}
	\lemmaTextIsTwSummaryCoreItems
	\item
	If $\aaa \in \{1,2\}$,
	$\frac{\aaa}{\bbb}$-\dispersion is in \p;
	while if $\aaa \geq 3$
	can be solved in $(2\aac)^{t} \cdot n^{O(1)}$ time
	but not in time $(2\aac-\varepsilon)^{\pw(G)} \cdot n^{O(1)}$.
\end{itemize}
}
\begin{theorem}[\Cref{section:lower:bound}]
\label{i:lemma:is:tw:summary}
\lemmaTextIsTwSummary
\end{theorem}

\newcommand{\inp}{\cellcolor{green!30}}
\newcommand{\infpt}{\cellcolor{orange!60}}
\newcommand{\inwt}{\cellcolor{red!60}}
\begin{figure}
\begin{center}
\begin{tabular}{r|c|c|c|c|c|c|}
\backslashbox{$a$}{$b$} &$1$&$2$&$3$&$4$&$5$&$6$ \\\hline
$1$ & \inp $|V|$ & \inp $|V|+|E|$ & \inp $|V|+2|E|$ & \inp $|V|+3|E|$ & \inp $|V|+4|E|$ & \inp $|V|+5|E|$ \\\hline
$2$ & \inwt IS & \inp $1$-\textsc{Disp} & \infpt IS$+|E|$ & \inp $1/2$-\textsc{Disp} & \infpt IS$+2|E|$ & \inp $1/3$-\textsc{Disp} \\\hline
$3$ & \inwt $3$-IS & \infpt & \inp $|V|$ & \infpt $3$-IS$+|E|$ & \infpt & \inp $|V|+|E|$ \\\hline
$4$ & \inwt $4$-IS & \inp $2$-\textsc{Disp} & \infpt & \inp $1$-\textsc{Disp} & \infpt $4$-IS$+|E|$ & \inp $2/3$-\textsc{Disp} \\\hline
$5$ & \inwt $5$-IS & \inwt & \infpt & \infpt & \inp $|V|$ & \infpt $5$-IS$+|E|$ \\\hline
$6$ & \inwt $6$-IS & \inwt $3$-\textsc{Disp} & \inwt IS & \infpt $3/2$-\textsc{Disp} & \infpt & \inp $1$-\textsc{Disp} \\\hline
\end{tabular}
\end{center}
\caption{
Some problems,
	such as
	\textsc{\textbf{I}ndependent\textbf{S}et} and \textsc{\textbf{Disp}ersion},
	(or solution sizes) corresponding to $a$-\IndSet on a graph $G_b$ for small values of $a,b$
	and $V=V(G)$ and $E=E(G)$.
A light green cell indicates a polynomial time solvable case,
	an orange \np-hardness \& \fpt,
	and a dark red \np-hardness \& \wone-hardness.
See also \cref{figure:ind:cells:big} for a bigger table.
		}
\label{figure:ind:cells}
\end{figure}

The proof heavily uses hidden symmetries of $a$-independent sets on $b$-subdivided graphs
	for different values of $a$ and $b$.
Such symmetries were explored first for a continuous version of $a$-independent set,
	in a series of work~\cite{GrigorievHLW21,HartmannL22,thesis}.
We show that these results hold in similar form for $a$-independent sets as well.

\subparagraph{Rational distances} 
As the distance between any two vertices in a graph is an integer, it makes no sense to consider $a$-independent 
sets for noninteger $a$. However, noninteger distances can be highly relevant if we consider the complexity of the 
said continuous version of $a$-independent. The continuous version, introduced by Dearing and Francis~\cite{Dearing1974}, 
is known as  \emph{$\delta$-dispersion} for a positive real distance $\delta$.
In this setting, instead of requiring a selection of vertices of a graph $G$,
	we allow the selection of \emph{points} that may be on a vertex or somewhere on the continuum of an edge.
We fix the length of the edges to $1$, which defines a distance relation of the points in the graph $G$.
A \emph{$\delta$-dispersed} set then is a subset of \emph{points} $S$
	where every distinct points $p,q\in S$ have distance at least $\delta$;
	as studied for example in~\cite{Tamir1991,GrigorievHLW21,HartmannL22}.
The problem $\delta$-\dispersion is the decision version asking for a $\delta$-dispersed set of size at least $k$,
	for some budget $k$ given in the input.
It turns out that the notion of $\frac{a}{b}$-dispersed sets is similar to $a$-independent sets on $b$-subdivided graphs. Indeed, a crucial connection between the two types of sets is that $\frac{a}{b}$-dispersed sets are in one-to-one correspondence
	to $2a$-independent sets on the $2b$-subdivided graph,
	as follows from a discretization argument by Grigoriev et al.~\cite{GrigorievHLW21}.
Particularly, the polynomial time solvable case of $a$-independent set,
	as stated in \cref{lemma:disp:p},
	follow from this discretization argument and
	a characterization of the polynomial time solvable cases of $\delta$-dispersion.
Finding a maximum $\delta$-dispersed set is polynomial time solvable
	if $\delta$ is a twice a unit-fraction (including $1$ and $2$),
	and all other cases are \np-hard~\cite{GrigorievHLW21}.
Further, $\delta$-dispersion when parameterized by the solution size is \fpt when $\delta \leq 2$
	and otherwise \wone-hard, as shown by Hartmann et al.~\cite{HartmannL22}.

With such connections and Theorem~\ref{i:lemma:is:tw:summary} at our hands, we can turn the results on $a$-independent set on $b$-subdivided graphs into tight results for $\delta$-\dispersion
	on bounded treewidth graphs
	for every fixed \emph{rational} $\delta$.

\begin{theorem}[\Cref{section:lower:bound}]
\label{i:lemma:disp:tw}
Let coprime $a,b$ define $\delta=\frac{a}{b}$.
If $a\le 2$, then  $\delta$-\dispersion is in \p.
For $a\ge 3$, given a tree decomposition of width $t$ of an $n$ vertex input graph,
	the problem can be solved in time $(2a)^t\cdot n^{O(1)}$,
	and, assuming \SETH, there is no $(2a-\varepsilon)^t\cdot n^{O(1)}$ time algorithm for any $\varepsilon>0$.
\end{theorem}

\newcommand{\fracxy}{\frac{x}{y}}
\newcommand{\fracxyinv}{\frac{y}{x}}
\newcommand{\fracxystar}{\frac{x^\star}{y^\star}}
\newcommand{\fracxystarinv}{\frac{y^\star}{x^\star}}

\subparagraph{Irrational distances}
By Theorem~\ref{i:lemma:disp:tw}, for a fixed \emph{rational} $\delta=\frac{a}{b}$,
	finding a maximum size $\frac{a}{b}$-dispersed set is fixed-parameter tractable in the treewidth of the input 
	graph.
This is not necessarily the case for \emph{irrational} $\delta$.
Deciding $\delta$-\dispersion can be as hard as outputting the digits of $\delta$,
	which for some $\delta$ is not even computable.
Consider, for example, a path of length $\ell$.
Then there is a dispersed set of size $k+1$
	if and only if $\frac{\ell}{k} \leq \delta$.
Hence it is reasonable to consider the question of efficient algorithms only if $\delta$ is \emph{efficiently 
comparable} to rationals, meaning that there is an algorithm that, given $\fracxyinv$,
	decides whether $\fracxy \leq \delta$ in time polynomial in $\log x + \log y$.

For every fixed {efficiently comparable} $\delta$, it is possible to find a maximum $\delta$-dispersed set in 
an $n$-vertex graph in time $n^{O(\tw(G))}$,
	i.e., there is an \xp algorithm parameterized by treewidth.
This follows from a rounding procedure by Hartmann et al.~\cite{HartmannL22},
	by which for an $n$-vertex graph
	the dispersion number of $\delta$ equals to the dispersion number of the smallest rational $\frac{x}{y}$
	where $\delta \leq \fracxy$ with $x\leq 2n$.
Using that $\delta$ is efficiently comparable,
	we can find this rational in polynomial time since $x \leq 2n$ and
	$y$ is in the order of $n$ for a fixed $\delta$.
Then it remains to apply the algorithm of \cref{i:lemma:disp:tw} to find a maximum $\fracxy$-dispersed set.

In contrast, the above algorithm cannot be improved to a
	fixed-parameter tractable under standard complexity assumptions. 
As we show, there is a very explicitly defined and {efficiently comparable}
	irrational $\delta = (4\sum_{j =1}^{\infty} 2^{-2^j} )^{-1}\approx 0.790085$,
	for which computing the $\delta$-dispersion number is \wone-hard
	parameterized by the treewidth (in fact even for pathwidth),
	and an according lower bound holds under the Exponential Time Hypothesis (\ETH).

\newcommand{\lemmaTextDispIrrational}{There is an efficiently comparable irrational $\delta$
	for which $\delta$-\dispersion
	is \wone-hard in the pathwidth $\pw(G)$ of the $n$-vertex input graph $G$
	and, assuming \ETH, cannot be solved in time $f(\pw(G)) \cdot {n}^{o(\pw(G))}$
	for any computable function~$f$.
}
\begin{theorem}[\Cref{section:irrational}]
\label{i:lemma:disp:wone:pw}
\lemmaTextDispIrrational
\end{theorem}

\newcommand{\ocovering}{overlaying\xspace} %
\newcommand{\ocover}{overlay\xspace}
\newcommand{\ocov}[1]{#1\text{-}\mbox{overlay}}

\subsection{Domination Problems}

In addition to distance $a$-independent set, we perform a similar study of the dual domination problems.
As we show, the results for $a$-independence hold quite similarly for according domination problems.
We use a definition that unifies several concepts such as that of a dominating set and a vertex cover.

\smallskip

A \emph{distance-$d$ dominating set} $D$ is a subset of vertices
	such that every other vertex is at distance at most $d$ to a vertex in $D$.
The literature contains several more distance domination-like problems,
	which are often quite well understood on bounded treewidth graphs.
A well-studied example is \emph{mixed dominating set},
	for example~\cite{ZhaoKS2011,JainJPS2017,MadathilPSS2019,XiaoS2020,DubloisLP21} and under the name 
	\emph{total 
	covering}~\cite{AlaviBLN77,ErdosM1977,Meir1978,AlaviLWZ1992,PeledSun1994},
	which is
	(even though not directly phrased as such)
	equivalent to a distance-$2$ dominating set of the $2$-subdivision of a graph $G$.
Similarly, a \emph{vertex-edge dominating set} is a subset of vertices $D$
	such that every edge
	has one of its end vertices in distance at most $1$ to a vertex in $D$,
	as studied in~\cite{lewis2010vertex, zylinski2019, ZiemannZylinski20}.
More generally, a \emph{distance-$d$ vertex cover} (not to be confused with a \emph{$d$-path vertex cover})
	is a subset of vertices $D$
	such that every edge has one of its end vertices in distance at most $d$ to a vertex in $D$,
	as studied in~\cite{AlvaradoDR2015,DallardKM2021}.

We unify all above concepts
	by the notion of an \emph{\dombarset{a}} for an integer $a$,
	which is a subset of vertices $D$ such that
	for every edge $e\in E(G)$,
	there are (possibly identical) vertices $w_1,w_2 \in D$
	and a $w_1,w_2$-walk of length at most $a$ containing edge~$e$.

\begin{observation}
\label{lemma:ds:corresponding:problems}
For a graph $G$ without isolated vertices, the following notions coincide:
\begin{itemize}
\item a vertex cover and a \dombarset{2},
\item a dominating set and a \dombarset{3},
\item a vertex-edge dominating set and a \dombarset{4},
\item a distance-$d$ dominating set and a \dombarset{(2d+1)}, for every $d\geq 1$, 
\item a mixed dominating set in $G$ and a \dombarset{5} in the $2$-subdivision $G_2$, and
\item a distance-$d$ vertex cover and a \dombarset{(2d+2)}, for every $d\geq 0$.
\end{itemize}
\end{observation}

\subparagraph*{Integer distances}
Finding a minimum \dombarset{a} is \np-hard for $a=2$ (i.e., finding minimum vertex cover)
	and it is not difficult to show that it remains \np-hard for any fixed $a\geq 2$.
In some cases, the hardness also extends to when we restrict the input to $2$-subdivided graphs:
Finding a minimum \emph{mixed dominating set}, i.e., a \dombarset{5} of $2$-subdivided graphs,
	is \np-hard, as shown by Majumdar~\cite{Majumdar1992}.
In contrast, there are polynomial time algorithms for $a\in\{2,4\}$
	when the input is a $2$-subdivided graph.

\begin{theorem}[Hartmann et al.~\cite{HartmannLW22}]
\label{lemma:ds:p}
For $a \in \{2,4\}$, a minimum \dombarset{a}
	on a $2$-subdivided graph can be found in linear time.\footnote{
	The original statement is about a continuous covering problem,
	but can be phrased as here by using a discretization argument given in the same work.}
\end{theorem}

These examples give a glimpse into the complexity
	of finding an \dombarset{a} of a $b$-subdivided graphs for integers $a,b$.
This work settles, for every fixed integer $a,b$, whether
	finding a minimum \dombarset{a} of a $b$-subdivided graph is polynomial time solvable or \np-hard,
	and additionally settles the parameterized complexity for the solution size as parameter
	(as color-coded in \cref{figure:ds:cells}).
Formally, let $\ogamma_a(G)$ be the minimum cardinality of an \dombarset{a} of a graph $G$.
For a graph class $\GG$, let $a$-$\DS(\GG)$ be the according decision problem,
	which is, given a graph $G \in \GG$ and an integer $k$, deciding whether $\ogamma_a(G)\geq k$.

\begin{figure}
\begin{center}
\begin{tabular}{r|c|c|c|c|c|c|}
\backslashbox{$a$}{$b$} &$1$&$2$&$3$&$4$&$5$&$6$ \\\hline
$1$ & \inp $|V|$ & \inp $|V|+|E|$ & \inp $|V|+2|E|$ & \inp $|V|+3|E|$ & \inp 
$|V|+4|E|$ & \inp $|V|+5|E|$ \\\hline
$2$ & \infpt \textsc{VC} & \inp ${1/2}$-\textsc{Cover} & \infpt \textsc{VC}$+|E|$ & \inp $1/4$-\textsc{Cover} & \infpt \textsc{VC}$+2|E|$ & \inp $1/6$-\textsc{Cover} \\\hline
$3$ & \inwt \textsc{DS} & \infpt \textsc{DS}$(G_2)$ & \inp $|V|$ & \infpt \textsc{DS}+$|E|$ & \infpt {DS}$(G_2)+|E|$ & \inp $|V|+|E|$ \\\hline
$4$ & \inwt \text{VED} & \inp ${1}$-\textsc{Cover} & \infpt & \inp ${1/2}$-\textsc{Cover} & \infpt VED$+|E|$ & \inp $1/3$-\textsc{Cover} \\\hline
$5$ & \inwt $2$-\text{DS} & \infpt \textsc{MDS} & \infpt & \infpt & \inp $|V|$ & \infpt $2$-DS$+|E|$ \\\hline
$6$ & \inwt & \inwt $3/2$-\textsc{Cover} & \infpt \textsc{VC} & \infpt \textsc{DS}$(\GG_2)$ & \infpt & \inp ${1/2}$-\textsc{Cover} \\\hline
$7$ & \inwt $3$-\text{DS} & \inwt  & \infpt & \infpt & \infpt & \infpt \\\hline
$8$ & \inwt & \inwt $2$-\textsc{Cover} & \infpt & \inp ${1}$-\textsc{Cover} & \infpt & \infpt $2/3$-\textsc{Cover} \\\hline
$9$ & \inwt $4$-\text{DS} & \inwt  & \inwt \textsc{DS} & \infpt & \infpt & \infpt \textsc{DS}$(\GG_2)$ \\\hline
\end{tabular}
\end{center}
\caption{Some problems, such as
	\textsc{\textbf{V}ertex\textbf{C}over},
	\textsc{(\textbf{M}ixed)\textbf{D}ominating\textbf{S}et},
	\textsc{\textbf{V}ertex\allowbreak\textbf{E}dge\allowbreak\textbf{D}omination},
	(or solution sizes) corresponding to $a$-\DomSet of a graph 
$G_b$ for small values of $a,b$,
	where $V=V(G)$ and $E=E(G)$.
A light green cell indicates a polynomial time solvable case,
	an orange \np-hardness and \fpt,
	and a dark red \np-hardness \& \wtwo-hardness.
See \cref{figure:ds:cells:big} for a bigger table.
		}
\label{figure:ds:cells}
\end{figure}

\commentt{
\begin{lemma}[\cref{section:dom:parameterized}]
\label{lemma:ds:param}
Let $a,b \in \NNN_+$.
$a$-$\DS(\GG_b)$ is polynomial time solvable
	if $b$ is a multiple of $a$,
	or $\ahalf$ is even and $b$ is an odd multiple of $\ahalf$;
	else is \np-hard and, if $\frac{a}{b} < 3$, is \fpt for the solution size as parameter;
	else is \wone-hard
\end{lemma}
}

\newcommand{\lemmaTextDsParam}{$a$-$\DS(\GG_b)$ is
\begin{itemize}
\item polynomial time solvable if $b=ca$ or if $b=c\ahalf$ and $\ahalf$ is even for some integer $c$;
	and is \np-hard for all other integers $a,b$; and is
\item fixed-parameter tractable for the parameter solution size
	if $\frac{a}{b} < 3$; and \wtwo-hard for all other integers $a,b$.
\end{itemize}}

\begin{theorem}[\cref{section:dom:parameterized}]
\label{lemma:ds:param}
\lemmaTextDsParam
\end{theorem}

We note that $a$-$\DS(\GG_b)$
	is polynomial time solvable for the same set of integers $a,b$
	where $a$-$\IndSet(\GG_b)$ is polynomial time solvable.
In contrast,
	the threshold which separates the fixed-parameter tractable cases from the \wone-hard/\wtwo-hard cases is 
	shifted,
	which should be expected
	as \textsc{Vertex Cover} and \textsc{Dominating Set} are to be separated by this threshold.

\smallskip

Further, we study the problem parameterized by treewidth.
Again, intuitively, the \dombarset{a} problem is harder for larger $a$.
Indeed, for many cases the notion of an \dombarset{a} corresponds to a known problem
	(as in \cref{lemma:ds:corresponding:problems})
	where the literature knows a matching upper and lower bound
	with $a$ in the base of an exponential run time for graphs of bounded treewidth, assuming \SETH.
This is also the case for $a=5$ on \emph{$2$-subdivided graphs}, as this corresponds to a mixed dominating set.

\begin{theorem}[\cite{Niedermeier06,RooijBR09,BorradaileL16,LokshtanovMS18,DubloisLP21}]
\label{lemma:ds:lit:seth}
For $a \in \{2\}\cup\{3,5,7,\dots\}$,
	given a tree decomposition of width $t$ of an $n$ vertex input graph,
	a minimum \dombarset{a} can be found in time $a^t \cdot n^{O(1)}$,
	and, assuming \SETH,
	there is no $(a-\varepsilon)^{t} \cdot n^{O(1)}$ time algorithm for any $\varepsilon>0$.
Moreover, for $a=5$ this even holds when the input graph is restricted to $2$-subdivided graphs.
\end{theorem}

We refine \cref{lemma:ds:lit:seth} by including also even distances $a\geq 4$
	and by considering the restriction of $a$-\DomSet to $b$-subdivided graphs,
	beyond the case $a=5$ and $b=2$.
We determine the optimal base of the exponent for all integers $a$ and $b$.
As it turns out, the optimal base depends on $a$ and $b$ in a very subtle way.

\begin{theorem}[\Cref{section:ds:treewdith}]
\label{lemma:i:ds:tw:domination}
Let $\aaa,\bbb$ integers with $\gcd(\aaa,\bbb)=c$ and $c \aac=\aaa$ and $c \bbc=\bbb$.
Let $n$ be the number of vertices of the input graph.
Assume \SETH, $\varepsilon>0$, and that a tree decomposition of width $t$ is part of the input.
\begin{itemize}
\item
If $\gcd(\aaa,\bbb)$ is odd:
If $\aac=1$, $\aaa$-$\DS(\GG_{\bbb})$ is in \p,
else
	$\aaa$-$\DS(\GG_{\bbb})$ can be solved in time $a^{t} \cdot n^{O(1)}$
	but not in $(a-\varepsilon)^{\pw(G)} \cdot n^{O(1)}$.
\item
If $\gcd(\aaa,\bbb)$ is even:
If $\aac\in\{1,2\}$, $\aaa$-$\DS(\GG_{\bbb})$ is in \p,
else
	$\aaa$-$\DS(\GG_{\bbb})$ can be solved in time $(2a)^{t} \cdot n^{O(1)}$
	but not $(2a-\varepsilon)^{\pw(G)} \cdot n^{O(1)}$.
\end{itemize}
\end{theorem}

The proof of \cref{lemma:i:ds:tw:domination}
	heavily uses hidden symmetries of \dombarsets{a} on $b$-subdivided graphs,
	which are of similar nature as for independent sets.
Such symmetries were explored first for a continuous version of \dombarset{a}~\cite{HartmannLW22}.
We show that these results hold in similar form for \dombarset{a} as well.

\smallskip

Regarding distance-$d$ domination, our results so far imply the following.

\begin{corollary}
Let $n$ be the number of vertices of the input graph.
Finding a minimum distance-$d$ dominating set in $b$-subdivided graphs
\begin{itemize}
\item is polynomial time solvable if $b$ is a multiple of $2d+1$, otherwise is \np-hard;
\item if $b$ is not a multiple of $2d+1$, with $\gcd(2d+1,b)=c$
	can be solved in time $((2d+1)/c)^t \cdot n^{O(1)}$ if a tree decomposition of width $t$ is part of the input,
	and, assuming \SETH, cannot be solved in time $((2d+1)/c-\varepsilon)^{\pw(G)} \cdot n^{O(1)}$; and
\item fixed-parameter tractable for the parameter solution size
	if $\frac{2d+1}{b} < 3$; and \wtwo-hard for all other values of $d,b$.
\end{itemize}
\end{corollary}

\subparagraph{Rational distances}

The continuous version of \dombarset{a},
	as introduced by Shier~\cite{Shier1977},
	is known as \emph{$\delta$-covering} for a positive real distance~$\delta$.
Similarly to $\delta$-dispersion,
	we allow the selection of \emph{points} that may be on a vertex or somewhere on the continuum of an edge.
We fix the length of the edges to $1$, which defines a distance relation of the points in the graph $G$.
A \emph{$\delta$-cover} is a set of points $S$
	that covers every point $p$ in the graph,
	that is there is a point $q \in S$ such that $p,q$ have distance at most $\delta$;
	as studied for example in~\cite{MegiddoTamir1983,Tamir1987}
	and receiving renewed attention lately \cite{HartmannJanssen2024,HartmannLW22}.
The problem $\delta$-\covering is the decision version
	asking for a $\delta$-cover of size at most $k$, for some budget $k$ given in the input.
The notion of a $\frac{a}{b}$-cover is quite similar to an \dombarset{a} of a $b$-subdivided graph;
	though the connection is more subtle compared to the independence problems.
Based on a discretization argument by Hartmann et al.~\cite{HartmannLW22} we show that
	$\frac{a}{b}$-covers are in one-to-one correspondence to \dombarsets{2a} on $b$-subdivided graphs,
	if $b$ is even;
	while if $b$ is odd,
	$\frac{a}{b}$-covers are in one-to-one correspondence to \dombarsets{4a} on $2b$-subdivided graphs.
Particularly, we obtain the polynomial time solvable cases of $\delta$-covering
	based on this connection.
Finding a minimum $\delta$-cover is polynomial time solvable if $\delta$ is a unit-fraction
	and otherwise $\np$-hard~\cite{HartmannLW22}.
By the same work, $\delta$-\covering parameterized by the solution size is \fpt in case $\delta < \frac{3}{2}$
	an otherwise \wtwo-hard.
(We observe a similar dichotomy for $a$-independent sets on $b$-subdivided graphs,
	as stated in \cref{lemma:ds:param}.)

With such connections and \cref{lemma:i:ds:tw:domination} at our hands,
	we can turn the results on $a$-independent set on $b$-subdivided graphs
	into tight results for $\delta$-\covering on bounded treewidth graphs
	for every fixed \emph{rational} $\delta$.

\begin{theorem}[\Cref{section:ds:treewdith}]
\label{lemma:i:ds:tw:covering}
Let $\aaa,\bbb$ integers with $\gcd(\aaa,\bbb)=c$ and $c \aac=\aaa$ and $c \bbc=\bbb$.
Let $n$ be the number of vertices of the input graph.
Assume \SETH, $\varepsilon>0$, and that a tree decomposition of width $t$ is part of the input.
\begin{itemize}
\item
$\frac{\aaa}{\bbb}$-\covering for $\aac =1$ is in $\p$;
	if $\aac \geq 2$ and $\bbc$ is odd,
	can be solved in time $(4\aac)^{t} \cdot n^{O(1)}$
	but not in $(4\aac-\varepsilon)^{\pw(G)} \cdot n^{O(1)}$;
	if $\aaa \geq 2$ and $\bbc$ is even,
	can be solved in time $(2\aac)^{t} \cdot n^{O(1)}$
	but not in $(2\aac-\varepsilon)^{\pw(G)} \cdot n^{O(1)}$.
\end{itemize}
\end{theorem}

\subparagraph{Irrational distances}

By \cref{lemma:i:ds:tw:covering}, for every fixed \emph{rational} $\delta = \frac{a}{b}$,
	finding a minimum $\frac{a}{b}$-cover is fixed-parameter tractable parameterized by the treewidth of the input graph.
As is the case for $\delta$-covering, this is not necessarily true for \emph{irrational} $\delta$.
Deciding $\delta$-\covering can be as hard as outputting the digits of $\delta$,
	which for some $\delta$ is not even computable.
For a path of length $\ell$, there is a covering set of size $k$
	if and only if $\delta \geq \frac{\ell}{2k}$.
Hence it is reasonable to consider only $\delta$ which are efficiently comparable.

For every fixed \emph{efficiently comparable} $\delta$,
	it is possible to find a minimum $\delta$-cover in time $n^{O(\tw(G))}$,
	i.e., there is an \xp algorithm parameterized by treewidth.
This follows from a rounding procedure by Tamir~\cite{Tamir1987},
	by which for an $n$-vertex graph the covering number of $\delta$
	equals to the covering number of the largest rational $\fracxy \leq \delta$ with $x \leq 2n$.
Using that $\delta$ is efficiently comparable, we can find this rational in polynomial time
	since $x < 2n$ and $y$ is in the order of $n$ for a fixed $\delta$.
Then it remains to apply the algorithm of \cref{lemma:i:ds:tw:covering} to
	find a minimum $\fracxy$-covering set.

In contrast, the above algorithm cannot be improved to a fixed-parameter tractable
	under standard complexity assumptions.
As we show, there is a very explicitly defined and efficiently comparable irrational
	$\delta' = (2\sum_{j =1}^{\infty} 2^{-2^j} )^{-1}\approx 0.395043$,
	for which computing the $\delta'$-covering number is \wone-hard parameterized by the treewidth
	(in fact even for pathwidth),
	and an according lower bound under \ETH.

\newcommand{\lemmaTextCoverIrrational}{
There is an efficiently comparable irrational $\delta'$
	such that $\delta'$-\covering is \wone-hard in the pathwidth $\pw(G)$
	of the $n$ vertex input graph $G$
	and, assuming \ETH, cannot be solved in time $f(\pw(G)) \cdot {n}^{o(\pw(G))}$
	for any computable function $f$.}

\begin{theorem}[\cref{section:cov:irrational}]
\label{lemma:cover:irrational}
\lemmaTextCoverIrrational
\end{theorem}

\subsection{Organization of this Work}

\cref{section:intro} to \cref{section:overview:domination} have the form of an extended abstract.
Eventually, we provide all details and full proofs in \cref{section:detaill:independence}
	and \cref{section:domination:and:covering} regarding independence and domination, respectively.

In more detail, the extended abstract part continues with preliminaries in \cref{section:preliminaries}.
Then, \cref{section:ind:and:disp} to \cref{section:irrational} concern independent set and covering.
That is \cref{section:ind:and:disp} explores the relationship between independent sets and dispersed sets;
\cref{section:ind:parameterized} studies the parameterized complexity of independent set with the solution size as parameter;
\cref{section:lower:bound} derives upper and lower bounds under SETH for independence;
and \cref{section:irrational} derives the hardness result of $\delta$-\dispersion for an irrational $\delta$.
Finally, \cref{section:overview:domination} overviews dominating set and covering.

\section{Preliminaries}
\label{section:preliminaries}
All our graphs are simple and undirected.
Usually we assume that our graphs as input do not contain isolated vertices,
	as they can easily be preprocessed for the studied problems.
	
\subparagraph{$b$-Subdivision}
For a graph $G$ and an integer $b$,
	the \emph{$b$-subdivision} of $G$, denoted as $G_b$,
	results from $G$ by replacing every edge $\{u,v\} \in E(G)$ by a $u,v$-path
	of length $b$.
For example, $G=G_1$.
For an edge $\{u,v\} \in E(G)$ and $\beta \in \{0,\dots,b\}$,
	let $\vv(u,v,\beta)$ be the unique vertex on the unique shortest $u,v$-path in $G_b$
	with distance $\beta$ to $u$ and distance $b-\beta$ to $v$.
Let $\GG_b$ be the class of every graph $H$ that is the $b$-subdivision $G_b$ of a graph $G$.

\subparagraph{Point space}
For a graph $G$, we assume that its edges have unit length.
Let $p(u,v,\lambda)$, for an edge $\{u,v\}\in E(G)$ and a real $\lambda \in [0,1]$,
	denote the \emph{point} on the edge $\{u,v\}$ with distance $\lambda$ to $u$
	and distance $1-\lambda$ to $v$.
Hence $p(u,v,\lambda)$ coincides with $p(u,v,1-\lambda)$,
	and the point $p(u,v,0)$ coincides with the vertex $u$.
By $P(G)$ we denote the set of points of a graph $G$.
Let $d(p,q)$, for two points $p,q \in P(G)$, denote the distance of $p,q$
	of the underlying metric space on $P(G)$.

\subparagraph{Graph parameters}

We use the following well known relation of graph parameters.
By $\nu(G)$ we denote the maximum size of a matching in a graph $G$.
A tree decomposition $(T,\beta)$ of $G$
	consists of a tree $T$
	and a mapping $\beta$ from the vertices of $T$ (referred to as \emph{nodes})
	to subsets of $V(G)$ (referred to as \emph{bags}),
	such that (1) $\bigcup_{n \in V(T)}=V(G)$, (2) $\{u,v\}\in E(G)$ implies a node $n \in V(T)$ with $\{u,v\}\subseteq \beta(n)$, and (3) for nodes $n_1,n_2,n_3\in V(T)$ where $n_2$ lies on the path from $n_1$ to $n_2$,
	we have $\beta(n_1) \cap \beta(n_3) \subseteq \beta(n_2)$.
The width of a tree decomposition is the maximum of $|\beta(n)|-1$ for all $n \in V(T)$.
A path decomposition of $G$ is a tree decomposition $(P,\beta)$ where $P$ is a path.
We let $(\beta(n_1),\dots,\beta(n_t))$ denote the path decomposition using a path $(n_1,\dots,n_t)$.
The \emph{treewidth} $\tw(G)$ of a graph is the minimum width of a tree decomposition,
	and likewise the \emph{pathwidth} $\pw(G)$ is the minimum width of a path decomposition.
	
It is well known that $2 \nu(G) \geq \pw(G) \geq \tw(G) $. %
Hence a parameterized algorithm with the treewidth as parameter is more general result
	than using the pathwidth (assuming a respective decomposition is given in the input).
On the other hand, a lower bound for the parameter pathwidth also holds for the parameter treewidth.

\subparagraph{Efficiently comparable real}
A real $\delta$ is \emph{efficiently comparable} if
	there is an algorithm that, given $\fracxyinv$,
	decides whether $\fracxy \leq \delta$ in time polynomial in $\log x + \log y$.

\section{Independent and Dispersed Sets}
\label{section:ind:and:disp}
\ineditnote{\cref{section:ds:and:covering}}

This section explores the close relationship between
	$a$-independent sets and $\delta$-dispersed sets.
Our goal is to establish the following two transformations of $\delta$-dispersed sets
	also for independent sets.
The first relates the dispersion number to the subdividing the graph.

\begin{lemma}[\cite{GrigorievHLW21}]
\label{lemma:disp:subdivide}
For every real $\delta > 0$ and integer $c \geq 1$,
$\disp{\delta}(G) = \disp{c\delta}(G_c)$.
\end{lemma}

The second relates the dispersion number of the same graph but of different distances.
For certain distances the solution size differs by exactly one point for every edge.

\begin{lemma}[\cite{HartmannL22,thesis}]
\label{lemma:disp:translate}
$\disp{\delta}(G) + |E(G)| = \disp{\tfrac{\delta}{1+\delta}}(G)$
	for every real $\delta>0$ and graph $G$ without cycles of length $<\delta$.
\end{lemma}

In a key result later, we will apply \cref{lemma:disp:translate} multiple time, stated as follows.
(The proof thereof and of other statements marked with \refstar\ are deferred to \cref{section:detaill:independence} and \cref{section:domination:and:covering}.)

\newcommand{\lemmaTextDispTranslateB}{
$\disp{\delta}(G) = \disp{\tfrac{\delta}{1+b \cdot \delta}}(G) - b|E(G)| $
	for an integer $b\geq1$, real $\delta>0$ and graph $G$ without cycles of length $<\delta$.}
\begin{corollary}\hyperref[a:lemma:disp:translate:b]{\refstar}
\label{lemma:disp:translate:b}
\lemmaTextDispTranslateB
\end{corollary}

To obtain \cref{lemma:disp:subdivide} and \cref{lemma:disp:translate} in terms of $a$-independent sets,
	we consider certain normalized dispersed sets.
To that end, recall that a point $p(u,v,\lambda) \in P(G)$ is \emph{$c$-simple}
	if $\lambda$ is a multiple of $c$.
Further, let $S \subseteq P(G)$ be $c$-simple if all its points are $c$-simple.
The good news is that there is a direct correspondence of $b$-simple $\tfrac{a}{b}$-dispersed sets in a graph $G$
	and $a$-independent sets in the $b$-subdivision $G_b$.

\begin{observation}
\label{lemma:is:iff:disp}
Let $k\in\NNN$.
There is an $a$-independent set of size $k$ in $G_b$,
	if and only if there is a $b$-simple $\tfrac{a}{b}$-dispersed set of size $k$ in $G$.
\end{observation}
\begin{proof}
The vertex $\vv(u,v,\beta)$ for an edge $\{u,v\}\in E(G)$ and $\beta \in \{0,\dots,b\}$
	corresponds to the $b$-simple point $p(u,v,\frac{i}{b})$ and vice versa.
The distance between vertices corresponds to the same distance of corresponding points multiplied by the factor $\frac{1}{b}$.
\end{proof}

Unfortunately, there may be no minimum $\frac{a}{b}$-dispersed set that is $b$-simple.
On the positive side, a minimum $\frac{a}{b}$-dispersed set $S$ can be modified
	to a \emph{$2b$-simple} dispersed set $S^\star$ of same size.
Actually, one can observe that $S^\star$ is not $b$-simple only because of certain points in $S$.

\begin{lemma}[\cite{GrigorievHLW21}]
\label{lemma:covering:simplify}
For an $\tfrac{a}{b}$-dispersed set $S$,
	the following set $S^\star = \{p^\star \mid p \in S\}$ is $2b$-simple and $\tfrac{a}{b}$-dispersed.
For each point $p(u,v,\lambda)$
	with $\lambda \in (\tfrac{i}{b}-\tfrac{1}{2b}, \tfrac{i}{b}+\tfrac{1}{2b})$
	for some integer $i\geq 0$,
	let $p^\star \coloneqq p(u,v,\tfrac{i}{b})$;
for all other $p\in S$, let $p^\star = p$.
\end{lemma}

\begin{corollary}
\label{lemma:is:to:disp}
$\disp{\tfrac{a}{b}}(G) = \alpha_{2a}(G_{2b})$ %
	for every $a,b \in \NNN_+$ and graph $G$.
\end{corollary}

This already implies a subdivision argument almost as for $a$-independent sets (\cref{lemma:disp:subdivide}).

\begin{lemma}
\label{lemma:alpha:sub}
For every graph $G$,
$\alpha_a(G_b) = \alpha_{ca}(G_{cb})$ if $c$ is odd, or $a,b$ are even.
\end{lemma}
\begin{proof}
First consider that $a$ and $b$ are even.
Then $\alpha_a(G_b) = \alpha_{2a'}(G_{2b'})$ for some integers $a',b'$.
Then $\alpha_{2a'}(G_{2b'}) = \disp{\frac{a'}{b'}} = \disp{\frac{c a'}{c b'}} = \alpha_{2ca'}(G_{2cb'}) = \alpha_{ca}(G_{cb})$, because of \cref{lemma:is:to:disp}.

Now consider that $c$ is odd.
Clearly, an $a$-independent set in $G_b$ corresponds to a $ca$-independent set in $G_{cb}$.
For the reverse direction,
	let $I \subseteq V(G_{cb})$ be a $ca$-independent set of $G_{cb}$.
Consider the corresponding $a$-dispersed set $S \subseteq V(G_b)$ in $G_b$,
	which is $bc$-simple.
We apply the construction of \autoref{lemma:covering:simplify}.
As $c$ is odd,
	there is no point $p \in S$ with edge position $\tfrac{1}{2}$.
Hence the construction only produces points with edge position $0$ or $1$.
That is, the constructed set $S^\star$ is $1$-simple,
	and hence $S^\star$ corresponds to an $a$-independent set in $G_b$.
\end{proof}

Next, we obtain the second connection (\cref{lemma:disp:translate})
	quite similarly for dispersed sets.
That is, we relate $a$-independent sets in the subdivided graphs $G_b$ and $G_{a+b}$.
The basic idea is to translate the independent set to a dispersed set and apply \cref{lemma:disp:translate}.
However, the construction of \cref{lemma:disp:translate}
	does not preserve simplicity.
Hence we adapt the construction slightly
	such that it always maps $b$-simple inputs to $(a+b)$-simple outputs.
As we show in \cref{section:detaill:independence},
	the correctness follows with almost the same proof as for \cref{lemma:disp:translate}.

\newcommand{\lemmaTextDispTranslateBSimple}{
Let $G$ be a graph without a cycle of length $<\tfrac{a}{b}$.
A $b$-simple $\tfrac{a}{b}$-dispersed set $S$
	implies an $(a+b)$-simple $\tfrac{a}{a+b}$-dispersed set of size $|S|+|E(G)|$.
Further, an $(a+b)$-simple $\tfrac{a}{a+b}$-dispersed set $S'$
	implies a $b$-simple $\tfrac{a}{b}$-dispersed set of size $|S'|-|E(G)|$.}

\begin{lemma}
(\cref{a:section:disp:translation})
\label{lemma:disp:translate:b:simple}
\ineditnote{\cref{lemma:cover:translate}}
\lemmaTextDispTranslateBSimple
\end{lemma}

Equipped with \cref{lemma:disp:translate:b:simple},
	we obtain an result analogous to \cref{lemma:disp:translate}. %

\begin{theorem}
\label{lemma:alpha:translate}
$\alpha_a(G_b) + |E(G)| = \alpha_a(G_{a+b})$ if graph $G_b$ contains no cycle of length $< a$.
\end{theorem}
\begin{proof}
Let $I$ be an $a$-independent set $I$ in $G_b$.
Then $I$ corresponds to a $b$-simple $\tfrac{a}{b}$-dispersed set $S$ in $G$, by \cref{lemma:is:iff:disp}.
\cref{lemma:disp:translate:b:simple} maps $S$ to an $(a+b)$-simple $\tfrac{a}{a+b}$-dispersed set $S'$ in $G$
	of size $|I|+|E(G)|$.
This construction is applicable since $G_b$ contains no cycle of length $<a$,
	and hence $G$ contains no cycle of length $<\tfrac{a}{b}$.
Then the $(a+b)$-simple set $S'$ corresponds to an $a$-independent set in $G_{a+b}$
	of size $|I|+|E(G)|$, by \cref{lemma:is:iff:disp}.
That means $\alpha_a(G_b) + |E(G)| \leq \alpha_a(G_{a+b})$.

Vice versa, let $I$ be an $a$-independent set in $G_{a+b}$.
Then $I'$ corresponds to an $(a+b)$-simple $\tfrac{a}{a+b}$-dispersed set $S'$ in $G$ of size $|I|$,
	by \cref{lemma:is:iff:disp}.
\cref{lemma:disp:translate:b:simple} maps $S$ to an $a$-simple $\tfrac{a}{b}$-dispersed set $S$ in $G$
	of size $|I'|-|E(G)|$.
Then the $b$-simple set $S$ corresponds to an $a$-independent set in $G_b$ of size $|I'|-|E(G)|$,
	by \cref{lemma:is:iff:disp}.
That means $\alpha_a(G_b) + |E(G)| \geq \alpha_a(G_{a+b})$.
\end{proof}

Now with these two relation of independent sets established
	(\cref{lemma:is:to:disp} and \cref{lemma:alpha:translate}),
	we easily obtain the integers $a,b$
	for which finding a maximum $a$-independent set on $b$-subdivided graphs
	is polynomial time solvable.
A simple case is, that $a$ is odd and $b$ is a multiple of $a$.
By \cref{lemma:alpha:sub}, this is equivalent to finding a maximum
	$1$-independent set on a $\frac{a}{b}$-subdivided graph,
	which trivially consist of all vertices.
In the case that $a$ is even and $b$ is a multiple of $a$,
	and that $\ahalf$ is even and $b$ is a multiple of $\ahalf$,
	\cref{lemma:alpha:sub} and \cref{lemma:alpha:translate}
	allow to reduce the problem to a polynomial time solvable case from \cref{i:lemma:ind:param}.
\cref{lemma:alpha:translate} is applicable as in above cases $\frac{a}{b} \leq 2$.

\begin{theorem}
\label{lemma:ind:poly}
$a$-\IndSet on $b$-subdivided graph is polynomial time solvable
	if $b$ is a multiple of $a$,
	or $\ahalf$ is even and $b$ is an odd multiple of $\ahalf$.
\end{theorem}

\section{Independent Set with Parameter Solution Size}
\label{section:ind:parameterized}

This section settles the parameterized complexity
	of finding a maximum $a$-independent set on $b$-subdivided graphs
	with the solution size as parameter,
	for every integer $a,b$.
These results are also color-coded in \cref{figure:ind:cells}
	and summarized as follows.

\begin{theorem}[\Cref{i:lemma:ind:param} restated]
\label{lemma:ind:param}
\ineditnote{\cref{lemma:ds:param}}
\lemmaTextIndParam
\end{theorem}

The polynomial time solvable cases are already settled by \cref{lemma:ind:poly}.
As is well-known, $2$-\IndSet (on general graphs, hence $1$-subdivided graphs)
	is \wone-hard~\cite{DowneyFellows1995}.
This also puts all cases where $\frac{a}{b}=2$ and $b$ odd to \wone-hard,
	by applying \cref{lemma:alpha:sub};
	while all over cases with $\frac{a}{b}=2$ are polynomial time solvable.
The fixed-parameter tractable cases rely on bounding the solution size below
	by the size of a maximum matching in a graph $G$, denoted as $\nu(G)$;
	similarly as for the covering problem~\cite{HartmannL22}.

\begin{lemma}
\label{lemma:ind:leq:matching}
\ineditnote{\cref{lemma:ds:leq:matching}}
$\nu(G) \leq \alpha_a(G_b)$
	for every graph $G$ and integers $a,b$ with $\frac{a}{b} < 2$.
\end{lemma}
\begin{proof}
If $\frac{a}{b}\leq 1$, then $V(G)$ is an $a$-independent set in $G_b$.
Since $|V(G)| \geq \nu(G)$, the statement follows for $\frac{a}{b}\leq 1$.
Otherwise, consider a maximum matching $M \subseteq E(G)$.
Then the two vertices $\vv(u,v,\floor{\frac{b}{2}})$ and $\vv(u',v',\floor{\frac{b}{2}})$
	for distinct matching edges $\{u,v\},\{u',v'\}\in M$
	have distance at least $b+2 \floor{\frac{b}{2}} \geq b + ({b-1}) = 2b-1\geq a$.
The last inequality holds since otherwise $2b \leq a$ in contradiction to $\frac{a}{b} < 2$.
In conclusion $\{ \vv(u,v,\floor{\frac{b}{2}}) \mid \{u,v\}\in M, u \prec v\}$ is an $a$-independent set of $G_b$,
	where $\prec$ is an arbitrary ordering of $V(G)$.
\end{proof}

As the maximum size of a matching in a graph upper bounds the treewidth,
	an \fpt-algorithm results from a win-win situation.
Either the input asks for an independent set that is relatively large
	compared to $\nu(G)$ and hence also compared to the treewidth,
	in which case we can use \cref{lemma:a:is:tw},
	or the answer is trivially `yes'.

\newcommand{\lemmaTextIndFpt}{For every $a,b$ with $\frac{a}{b} < 2$,
	$a$-$\IndSet(\GG_b)$ is \fpt for the parameter solution size.}
\begin{lemma}\hyperref[a:lemma:ind:fpt]{\refstar}
\label{lemma:ind:fpt}
\lemmaTextIndFpt
\end{lemma}

It remains to show \wone-hardness if $\frac{a}{b} > 2$,
	which follows from two parameter preserving reductions from \indset
	showing \wone-hardness,
	the first for $\frac{a}{b} \in (2,3)$, the second for $\frac{a}{b} \geq 3$;
	similarly as for the covering problem~\cite{HartmannL22}.

\newcommand{\lemmaTextIndTwoHardness}{For integers $a,b$ where $\frac{a}{b} \in (2,3)$, $a$-$\IndSet(\GG_b)$ is 
\wone-hard.}

\newcommand{\lemmaTextIndThreeHardness}{
For integers $a,b$ where $\frac{a}{b}\geq 3$,
	$a$-$\IndSet(\GG_b)$ is \wone-hard.}

\newcommand{\lemmaTextIndHardness}{For integers $a,b$ where $\frac{a}{b} > 2$, $a$-$\IndSet(\GG_b)$ is \wone-hard.}

\begin{lemma}\hyperref[a:lemma:ind:hardness]{\refstar}
\label{lemma:ind:hardness}
\lemmaTextIndHardness
\end{lemma}

\section{Dispersion for Rational Distances}
\label{section:lower:bound}

This section derives the upper and lower bounds under \SETH for
	finding a minimum $a$-independent set on $b$-subdivided graphs
	for the parameter treewidth, for all integers $a,b$.
All lower bounds follow from the mere lower bound for even distance $a\geq 6$.

\newcommand{\lemmaTextAlphaSubLB}{
Let $a\geq6$ be even.
Let $n$ be the number of vertices of the input graph.
Assuming \SETH,
	$a$-$\IndSet$ has no $(a-\varepsilon)^{\pw(G)} \cdot n^{O(1)}$ time algorithm
	for any $\varepsilon > 0$,
	even when the input is restricted to $2$-subdivided graphs
	without a cycle of length $<a$.}
\begin{theorem}\hyperref[a:lemma:is:ppset:lb]{\refstar}
\label{lemma:alpha:sub:lower:bound}
\lemmaTextAlphaSubLB
\end{theorem}

In fact, we show this lower bound assuming the Primal Pathwidth SETH,
	recently introduced by Lampis~\cite{Lampis2024pwseth}.
We provide the details in \cref{a:section:lower:bound}.

\begin{theorem}[\cref{i:lemma:is:tw:summary}, \cref{i:lemma:disp:tw} combined]
\label{p:lemma:is:tw:summary}
\ineditnote{\cref{lemma:ds:tw:summary}}
\lemmaTextIsTwSummaryExtended
\end{theorem}

\begin{proof}
First, we consider that $c = \gcd(\aaa,\bbb)$ is odd.
Then $ca$-$\IndSet(\GG_{cb})$ is equivalent to $\aac$-$\IndSet(\GG_{\bbc})$
	by \cref{lemma:alpha:sub}.
In case $\aac=1$, then $1$-$\IndSet(\GG_{\bbc})$ has the trivial $1$-independent set $|V(G)|$.
Otherwise, $\aac$-$\IndSet(\GG_{\bbc})$ can be solved in time $\aac^{t} \cdot n^{O(1)}$ using \cref{lemma:a:is:tw}.

For the lower bound we use that $y \bbc = 1+x \aac$ for some integers $x,y$.
Assume \SETH.
Then we know from \cref{lemma:a:is:tw} that
	$\aac$-$\IndSet(\GG_1)$ has no $(\aac-\varepsilon)^{\pw(G)} \cdot n^{O(1)}$ time algorithm for any $\varepsilon 
	> 0$.
Particularly, this lower bound relies on graphs without a cycle of length $<\aac$.
Then \cref{lemma:alpha:translate} applied $x$ times
	yields that
	$\aac$-$\IndSet(\GG_{1+xa})$,
	and equivalently $\aac$-$\IndSet(\GG_{y \bbc})$,
	also has no $(\aac-\varepsilon)^{\pw(G)} \cdot n^{O(1)}$ time algorithm for any $\varepsilon >0$.
Thus especially $\aac$-$\IndSet(\GG_{\bbc})$
	has no $(\aac-\varepsilon)^{\pw(G)} \cdot n^{O(1)}$ time algorithm.
This settles the cases for independent set with odd $\gcd(\aaa,\bbb)$.

Next, we consider that $c = \gcd(\aaa,\bbb)$ is even, hence that $c=2\hat c$ for some integer $\hat c$.
Then $(2\hat ca)$-$\IndSet(\GG_{2\hat cb})$ is equivalent to $2a$-$\IndSet(\GG_{2b})$,
	by \cref{lemma:alpha:sub}. %
Again, by \cref{lemma:a:is:tw},
	$2a$-$\IndSet(\GG_{2b})$ can be solved in time $(2a)^{t} \cdot n^{O(1)}$.

If $\aac \in \{1,2\}$,
	then $\aac$-\dispersion is polynomial time solvable, by \cref{lemma:disp:p}.
Further, as $\frac{a}{b}\leq 2$, applying \cref{lemma:disp:translate} yields that
	also $\frac{\aaa}{\bbb}$-\dispersion
	is polynomial time solvable (as also observed in~\cite{GrigorievHLW21}).
By \cref{lemma:is:to:disp}, $2a$-$\IndSet(\GG_{2b})$ is equivalent to $\frac{\aaa}{\bbb}$-\dispersion,
	hence polynomial time solvable.

In case $\aac \ge 3$, and assuming \SETH,
	\cref{lemma:alpha:sub:lower:bound} provides that
	$2a$-$\IndSet(\GG_2)$ has no $(2a-\varepsilon)^{\pw(G)} \cdot n^{O(1)}$ time algorithm
	for any $\varepsilon>0$.
Particularly, this lower bound does not rely on graphs with a cycle of length $<\aac$.
Since $\aac,\bbc$ are co-prime,
	again \cref{lemma:alpha:translate} applies,
	and we obtain that
	$2a$-$\IndSet(\GG_{2b})$ has no $(2a-\varepsilon)^{\pw(G)} \cdot n^{O(1)}$ for any $\varepsilon>0$.
This settles the cases for independent set with even $\gcd(\aaa,\bbb)$.

Finally, $\frac{\aaa}{\bbb}$-\dispersion is equivalent to $\frac{\aac}{\bbc}$-\dispersion by definition.
Then $\frac{\aac}{\bbc}$-\dispersion is equivalent to $2a$-$\IndSet(\GG_{2b})$
	by \cref{lemma:is:to:disp}.
Since $\aac,\bbc$ are co-prime, $2a,2b$ have greatest common devisor $2$.
By the discussion for an greatest common devisor which is even,
	we follow that $\frac{\aac}{\bbc}$-\dispersion has an $(2a)^{t} \cdot n^{O(1)}$ time algorithm,
	and, assuming \SETH, has no $(2a-\varepsilon)^{t} \cdot n^{O(1)}$ time algorithm
	for any $\varepsilon > 0$.
\end{proof}

\section{Dispersion for Irrational Distance}
\label{section:irrational}
\newcommand{\SubgraphIso}{\textsc{SubgraphIso}\xspace}
\newcommand{\PartSubgraphIso}{\textsc{PartSubgraphIso}\xspace}
\newcommand{\CSP}{\textsc{CSP}\xspace}

This section derives the hardness result for computing a maximum $\delta$-dispersed set
	for the \emph{efficiently comparable} irrational $\delta = (\sum_{j \in [i+1]} 
2^{2-2^j} )^{-1} \approx 0.790085$.

\begin{theorem}[\cref{i:lemma:disp:wone:pw} restated]
\label{lemma:disp:wone:pw} %
\lemmaTextDispIrrational
\end{theorem}

\begin{figure}
\begin{center}
\begin{tikzpicture}[every text node part/.style={align=center}]

\node (b) at (4,0) {$\ColorfulClique$\\ $k$ color classes \\ each of size $n$};
\node[rotate=35] (ab1) at (6,0.75) {$\leq_p$};
\node[rotate=-35] (ab2) at (6,-0.75) {$\leq_p$};
\node (c1) at (8,1.25) {$c_i$-$\dispersion$\\ pathwidth $\Oh(k)$};

\node (c12) at (7.5,0) {same mapping, \\ \cref{lemma:clique:leq:disp}};
\draw[->] (c12) -- (ab1) {};
\draw[->] (c12) -- (ab2) {};

\node (c2) at (8,-1.25) {$(c_i-1)$-$\dispersion$\\ pathwidth $\Oh(k)$};
\node (cd1) at (10,1.25) {$\leq_p$};
\node (cd12) at (10,0) {$\quad$same mapping,\\ \cref{lemma:disp:reduction:delta:gamma}};

\node (cd122) at (12.5,0) {$\gamma_i \leq \delta \leq \delta_i$};

\node (cd2) at (10,-1.25) {$\leq_p$};
\draw[->] (cd12) -- (cd1) {};
\draw[->] (cd12) -- (cd2) {};

\node (d1) at (12,1.25) {$\delta_i$-$\dispersion$\\ pathwidth $\Oh(k)$};
\node (d2) at (12,-1.25) {$\gamma_i$-$\dispersion$\\ pathwidth $\Oh(k)$};
\node[rotate=-35] (ed1) at (14,0.75) {$=$};
\node[rotate=35] (ed2) at (14,-0.75) {$=$};
\draw[->] (cd122) -- (ed1) {};
\draw[->] (cd122) -- (ed2) {};

\node (e) at (15,0) {$\delta$-$\dispersion$\\ pathwidth $\Oh(k)$};
\end{tikzpicture}
\end{center}
\caption{Reductions used for the proof of \cref{lemma:disp:wone:pw}.
\cref{lemma:clique:leq:disp} is simultaneously a reduction to $c_i$-\dispersion and to $(c_i-1)$-\dispersion.
'Agnostic' to whether the distance is $c_i$ or $c_i-1$,
	\cref{lemma:disp:reduction:delta:gamma} reduces to \dispersion with distance $\gamma_i$ respectively $\delta_i$.
These rationals $\gamma_i$ and $\delta_i$ form an approximation of $\delta$ from below and above.
Thus the combined reduction reduces to $\delta$-\dispersion.
}
\end{figure}
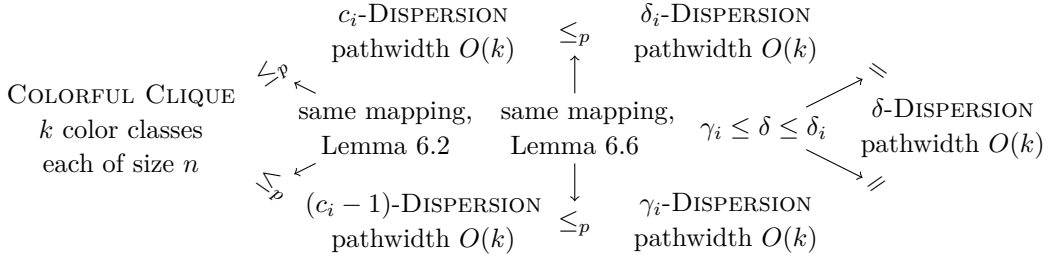

Our proof is based on two main reductions.
The first one is fairly standard:
It is a reduction from a colorful clique problem to $c_i$-\dispersion, where $c_i=2^{2^i}$ is a sufficiently large 
integer (polynomially bounded in the number of vertices of the colorful clique problem).
The reduction is robust in the sense that it is simultaneously a reduction to $(c_i-1)$-\dispersion as well, i.e., 
the yes/no answer does not change if we reduce the radius by 1.
The second reduction is the main nontrivial part of the proof:
We reduce $c_i$-\dispersion to 
$\delta_i$-\dispersion for some rational $\delta_i$ in a robust way. That is, the reduction can be interpreted also 
as a reduction from $(c_i-1)$-\dispersion to $\gamma_i$-\dispersion for some $\gamma_i<\delta_i$. Thus if the 
source instance has the same yes/no answer for radius $c_i$ and $c_i-1$, then the target problem has the same 
answer for any radius $\delta\in [\gamma_i,\delta_i]$.
We manage 
to define $\gamma_i,\delta_i$ in such a way that there is an irrational $\delta$ that is in $[\gamma_i,\delta_i]$ 
for every $i$. Thus for every $i$, the problem can be reduced to $\delta$-\dispersion.

\smallskip

Our main tools so far are subdividing (using \cref{lemma:disp:subdivide}),
	and translating (using \cref{lemma:disp:translate}).
Translating only applies to instances $\langle \delta,G,k\rangle$
	where the graph $G$ contains no cycle of length $<\delta$.
Hence, for convenience, let $\dispersionstar$ be the $\dispersion$ problem
	restricted to instances where the graph $G$ contains no cycle of length $<\delta$.

The starting point is \ColorfulClique where, given a graph $G$ and an integer $k$
	and a proper $k$-coloring of $G$,
	the task is to decide whether $G$ contains a $k$-clique
	that contains exactly one vertex of each color.	
It is known~\cite{bookParameterized} that \ColorfulClique is \wone-hard parameterized by the solution size $k$ and,
	assuming \ETH, has no $f(k) \cdot n^{o(k)}$ time algorithm for any computable function $f$.

The first reduction is based on a reduction
	from \ColorfulClique to the task of finding a maximum $a$-independent set with $a$ as part of the input,
	as given by Katsikarelis et al.~\cite{KatsikarelisLP2022}.
We output a graph with enough leeway
	such that the $\delta$-dispersion number does not change for $\delta$ in the interval $[c,c-1]$ for some integer 
	$c$.
Doing so, our construction constitutes a reduction from \ColorfulClique to $c$-\dispersion
	and, at the same time, a reduction from \ColorfulClique to $(c-1)$-\dispersion.
Further, we make sure that the construction does not introduce any short cycles.

\begin{lemma}\hyperref[a:lemma:clique:leq:disp]{\refstar}
\label{lemma:clique:leq:disp}
\ineditnote{\cref{lemma:gt:leq:cover}.}
There is a polynomial time reduction that,
	given a \ColorfulClique-instance $\langle G,k\rangle$ with $k$ color classes of size $n$,
	outputs a graph $G'$ of pathwidth $\Oh(k)$ and integer $k'$, such that: %
$\langle G,k\rangle$ is a yes-instance of \ColorfulClique
	if and only if $\langle G', k'\rangle$ is a yes-instance of $32n$-$\dispersionstar$
	if and only if $\langle G', k'\rangle$ is a yes-instance of $(32n-1)$-$\dispersionstar$.
\end{lemma}

Next, let us define $\delta$ in an abstract sense.
Distance $\delta$ is approximated by values $\delta_i$ and $\gamma_i$ from below and above with increasing precision.
The idea is to define distance $\delta_i$ by a fraction $\frac{a_i}{b_i}$
	that results from applying translation (\cref{lemma:disp:translate}) and subdivision 
	(\cref{lemma:disp:subdivide})
	to a (quite large) distance $c_i$.
Later, our second reduction then reduces to $\delta_i$-\dispersion and to $\gamma_i$-\dispersion
	by applying the according translation and subdivision.

\begin{definition}
\label{i:def:delta}
Let $c_1,c_2,\dots \in \NNN_+$ be an increasing integer sequence.
Then $a_0 = b_0 = 1$ and, for $i\geq 1$,
	$$a_i \coloneqq a_{i-1} c_i, \quad\quad b_i \coloneqq b_{i-1} c_i + 1, \quad\quad
	\delta_i \coloneqq \frac{a_i}{b_i}, \quad\quad \gamma_i \coloneqq \frac{a_i - a_{i-1}}{b_i - b_{i-1}}
	.$$
This defines $\delta \coloneqq \lim_{i\to\infty}\delta_i$.
\end{definition}

The sequence is decreasing and bounded from below, hence the limit $\delta\coloneqq \lim_{i\to\infty}\delta_i$ is 
well defined.

\begin{lemma}
\hyperref[a:lemma:disp:gamma:delta]{\refstar}
\label{lemma:disp:gamma:delta}
For $i\geq2$, and $\gamma_i$, $\delta$, $\delta_i$ as defined in \Cref{i:def:delta}, we have
$
\gamma_{i-1} < \gamma_i < \delta < \delta_i < \delta_{i-1}
$
\end{lemma}

We obtain nice computational properties if we use the double-exponential sequence for $c_i$.

\begin{lemma}
\label{lemma:disp:computing:a:b}
Using sequence $c_i \coloneqq 2^{2^i}$ for \cref{i:def:delta},
	integers $a_i, b_i$ are polynomial-time computable given $c_i$,
		and $a_{i}$ is polynomial in $c_i$.
	Further, $\delta$ is efficiently comparable.
\end{lemma}
\begin{proof}
We observe that $a_i = \prod_{j=1}^{i} c_j = 2^{2^1} \cdot 2^{2^2} \cdots 2^{2^i} = 2^{2^{i+1}-2} = 2^{2^{i+1}}/4 = 
{c_i}^2/4$,
	hence that $a_i$ is polynomial-time computable given $c_i$ and is polynomial in $c_i$.
Further, $b_i = (\dots((c_1+1)c_2+1)\dots)c_i+1 = \sum_{j=1}^{i} c_{j}c_{j+1}\dots c_i
	\leq  i \cdot  a_i = \log\log c_i \cdot c_i^2/4$.
Hence $b_i$ is polynomial-time computable given $c_i$, by at most $2i$ multiplications and additions
	of integers that are polynomial in $c_i$.

Let us determine $\delta$.
We let $\eta_i = \sum_{j=1 }^i 2^{-2^j}$.
Then
$$ \textstyle
b_i \;=\; \sum_{j=1}^{i+1} \; \prod_{k=j}^{i} c_k %
	\;=\; \prod_{k=1}^{i} c_k \sum_{j=1}^{i+1} \; \prod_{k=1}^{j-1} c_k^{-1}
	\;=\; a_i \sum_{j=1}^{i+1 } 2^{-(2^j-2)}
	\;=\; 4 a_i \eta_{i+1}. %
$$
This yields
$
\delta = \lim_{i \to \infty} \frac{a_i}{b_i}
	= \left(4 \sum_{j=1}^{\infty} 2^{-2^j} \right)^{-1} \approx 0.790085.
$

We show that $\delta$ is efficiently comparable,
	that is there is an algorithm that, given a rational $\fracxy$,
	decides whether $\fracxy < \delta$ in time polynomial in $\log x + \log y$.
Our algorithm first checks whether $\half < \fracxy < 1$,
	and if not can conclude that $\fracxy < \delta$ or $\fracxy > \delta$.
Instead of comparing $\fracxy$ with $\delta$,
	we compare their inverses $\fracxyinv$ and $\delta^{-1}$, and output the negated answer.
In base $2$, we obtain that $\delta^{-1} = 1.010001 0000 0001 \dots$,
	which is that the $i$-th digit $1$ succeeds the $(i-1)$-st digit $1$
	in $2^i$ steps.
Hence the first $j$ digits (after the dot) of $\delta^{-1}$ can be output in time polynomial in 
$j$.
We may also output the first $j$ digits (after the dot) of $\fracxyinv$ in time polynomial in $j$.
If there is a position where the digits differ,
	we can conclude whichever is larger.
It remains to show that there will be a difference in the first $j =O(\log x +\log y)$ digits
	of $\delta^{-1}$ and $\fracxyinv$,
	hence that comparing the first $j$ digits suffices.
Indeed, the digits of $\fracxyinv$ as a string
	cannot contain the substring $0^{\ceil{ \log x}}1$ after the dot.
Otherwise $\fracxyinv+\fracxyinv$ contains the substring $0^{\ceil{ \log x}-1}1$ after the dot,
	and by induction $\fracxyinv \cdot x = y$
	contains the substring $1$ after the dot,
	in contradiction that $y$ is integer.
In contrast, the first $O(\log x)$ digits of $\delta^{-1}$ do include the substring $0^{\ceil{ \log x}}1$.
Thus it suffices to compare the fist $\log x$ digits of $\delta^{-1}$ and $\fracxyinv$,
	which concludes the proof.
\end{proof}

The following lemma lies at the heart of our result:
the definition of the sequences $a_i$, $b_i$, $c_i$ allows us to reduce $c_i$-\dispersion to $\delta_i$-\dispersion 
and, at the same time, $(c_i-1)$-\dispersion to $\gamma_i$-\dispersion with the same reduction.
\begin{lemma}
\label{lemma:disp:reduction:delta:gamma}
Let sequence $(c_i)_{i\geq 1}$ be as in \Cref{i:def:delta}.
There is a polynomial-time reduction that,
	given integers $c_i,k'$ and a graph $G'$,
	outputs a subdivision $G''$ of $G'$ and integer $k''$, such that:
$\langle G',k'\rangle$ is a yes-instance of $c_i$-$\dispersionstar$,
	if and only if the output $\langle G'',k''\rangle$ is a yes-instance of $\delta_i$-$\dispersionstar$.
Also, $\langle G',k'\rangle$ is a yes-instance of $(c_i-1)$-$\dispersionstar$,
	if and only if the output $\langle G'',k''\rangle$ is a yes-instance of $\gamma_i$-$\dispersionstar$.
\end{lemma}
\begin{proof}
Let $a_i, b_i$ and sequence $(c_i)_{i\geq 1}$ be defined as in \Cref{i:def:delta}.
Our algorithm begins by addressing some border cases.
If $k'=0$, we output a trivial simultaneous yes-instance of $\delta_i$-$\dispersionstar$ and 
$\gamma_i$-$\dispersionstar$.
Else, if $c_i$ exceeds $|V(G')|$ and $k'\geq 1$, we output a trivial simultaneous no-instance.
Else, we output an $a_{i-1}$-subdivision of the input graph $G'$ as $G''$
	and as budget $k'' = k' + b_{i-1} |E(G')|$.
Hence $G'$ and $G''$ have the same pathwidth up to subdividing the edges.
Any number of subdivisions of edges may increase the pathwidth only by a total of one.
To compute $g_i$ with $c_i$ as part of the input,
	we use \cref{lemma:disp:computing:a:b} to compute $a_{i-1}$ and $b_{i-1}$ in polynomial time.
In particular, $a_{i-1}$ is polynomial in $c_i$ and hence polynomial in $|V(G')|$,
	such that we may output $G''$, the $a_{i-1}$-subdivision of $G'$, in polynomial time.

We have $\disp{c_i}(G') = \disp{\frac{c_i}{1+c_i}}(G') - |E(G')|$ by \cref{lemma:disp:translate}
	and as $G'$ contains no cycle of length $<a$.
Applying this translation not only once but $b_{i-1}$ times, by \cref{lemma:disp:translate:b},
	we obtain $\disp{c_i}(G') = \disp{\frac{c_i}{ 1+b_{i-1}c_i }}(G') - b_{i-1}|E(G')|$.
Then by an $a_{i-1}$-subdivision of the input graph we have
	$\disp{\frac{c_i}{ 1+b_{i-1}c_i} }(G') =  \disp{\frac{a_{i-1} c_i}{ 1+b_{i-1}c_i }}(G'_{a_{i-1}})
	= \disp{\frac{a_i}{b_i}}(G'') = \disp{\delta_i}(G'')$ by \cref{lemma:disp:subdivide}.
Thus the input $\langle G',k'\rangle$ is a yes-instance of $c_i$-$\dispersionstar$,
	if and only if the output $\langle G'',k'' \rangle$ is a yes-instance of $\delta_i$-$\dispersionstar$.

The analogous transformations yields that
	$\disp{(c_i-1)}(G') = \disp{\frac{a_{i-1}(c_i-1)}{ 1 + b_{i-1}(c_i-1) }}(G'')  - b_{i-1}|E(G')|$.
We observe that the numerator of the latter is $a_{i-1}(c_i-1)=a_{i-1}c_i - a_{i-1} = a_i - a_{i-1}$,
	while the denominator is $1 + b_{i-1}(c_i-1) = 1 + b_{i-1}c_i - b_{i-1} = b_i - b_{i-1}$.
Hence this rational is equal to $\gamma_i$.
Thus $\langle G',k'\rangle$ is a yes-instance of $(c_i-1)$-$\dispersionstar$,
	if and only if the output $\langle G'',k'' \rangle$ is a yes-instance of $\gamma_i$-$\dispersionstar$.
\end{proof}

\begin{proof}[Proof of \cref{lemma:disp:wone:pw}]
Let $\delta$ be defined by integer sequence $c_i = 2^{2^i}$ for $i\geq 1$.
Then $\delta$ is efficient comparable by \cref{lemma:disp:computing:a:b}.
Consider a \ColorfulClique-instance $\langle G,k\rangle$ with color classes of size $\hat n$.
Let $i$ be such that $\hat n \leq c_i/32=2^{2^i-5} \eqqcolon n$, hence $32n=c_i$.
We note that $c_{j+1} = {c^2_j}$, for $j\geq 0$,
	and hence $n \leq {\hat n}^2$.
Thus $n$ and $c_i$ are polynomial-time computable, and $n,c_i$ are polynomial in $\hat n$.
We extend the color-classes of $\langle G,k\rangle$ with $n-\hat n$ isolated vertices each,
	resulting in color classes of size~$n$.
Next, we apply the reduction of \cref{lemma:clique:leq:disp} on $\langle G,k \rangle$,
	now with $c_i$ color classes,
	which outputs $\langle G',k'\rangle$.
In turn, we apply the reduction of \cref{lemma:disp:reduction:delta:gamma} on $\langle G',k'\rangle$
	which outputs $\langle G'',k''\rangle$,
	forming our final output.

We note that the reductions of \cref{lemma:clique:leq:disp} and \cref{lemma:disp:reduction:delta:gamma}
	are polynomial-time computable.
The former outputs a graph of pathwidth $\Oh(k)$,
	the latter does not change the pathwidth up to a constant.
Hence overall we output a graph $G''$ of pathwidth $\Oh(k)$.

By \cref{lemma:clique:leq:disp} and \cref{lemma:disp:reduction:delta:gamma},
	$\langle G,k\rangle$ is a yes-instance of \ColorfulClique,
	if and only if $\langle G'',k''\rangle$ is a yes-instance of $\delta_i$-\dispersion,
	if and only if $\langle G'',k''\rangle$ is a yes-instance of $\gamma_i$-\dispersion.
Since $\gamma_i < \delta < \delta_i$, by \cref{lemma:disp:gamma:delta},
	the dispersion numbers satisfy $\disp{\gamma_i}(G'')=\disp{\delta}(G'')=\disp{\delta}(G'')$.
Thus the output $\langle G'',k'' \rangle$ is a yes-instance of $\delta$-\dispersion,
	if and only if $\langle G,k \rangle$ is a yes-instance of \ColorfulClique.

Since \ColorfulClique is \wone-hard parameterized by $k$,
	also $\delta$-\dispersion is \wone-hard parameterized by the pathwidth of the input graph.
For the lower bound under \ETH, assume an $f(\pw(G)) \cdot n^{o(\pw(G))}$ time algorithm for $\delta$-\dispersion
	for a computable function~$f$.
Then using the above reduction on a \ColorfulClique-instance
	yields an $f(k) \cdot n^{o(k)}$ time algorithm for \ColorfulClique, in contradiction to \ETH.
\end{proof}

\section{Domination and Covering}
\label{section:overview:domination}

\newcommand{\definitionsWalkDomination}{
Let $G$ be a graph without isolated vertices.
For an integer $a$,
	a subset $D \subseteq V(G)$
	$a$-\emph{walk dominates} some subset of edges $E' \subseteq E(G)$,
	defining $V' \coloneqq V(G[E'])$,
	if:
\begin{itemize}
\item \labeltext{$($D1$)$}{def:ds:1}
For every edge $e\in E'$,
	there are (possibly identical) vertices $w_1,w_2 \in D$
	and a $w_1,w_2$-walk in $G$ of length at most $a$ that contains $e$; or
\item \labeltext{$($D2$)$}{def:ds:2}
Every vertex $u \in V(G[E'])$ has $d(u,D) \leq \ahalf$,
	and the set vertices $u \in V(G[E'])$ where $d(u,D)=\ahalf$ forms an independent set.
\item \labeltext{$($D3$)$}{def:ds:3}
$D \subseteq V(G)$ $a$-dominates $V'\cup E'$ in the $2$-subdivision $G_2$ of $G$
	(when identifiying an edge $\{u,v\}$ with the vertex with neighborhood $\{u,v\}$ in $G_2$).
That is, for every vertex $u \in V'\cup E'$ ,
	there is a vertex $w \in D$ with $d_{G_2}(u,w)\leq a$.
\end{itemize}
An $a$-\emph{walk dominating set} of $G$ is a subset $D \subseteq V(G)$
	that $a$-\emph{walk dominates} $E(G)$.
}

Finally, we turn to the domination problems and covering as their continuous counterpart.
This section outlines the connection of \dombarset{a} and $\delta$-covers.
We establish tools that relate \dombarset{a} on $b$-subdivided graphs for different values of $a,b$,
	similarly as we did for the independent set problem.
These tools then allow to derive the complexity results for $a$-\DS on $b$-subdivided graphs
	and $\delta$-\covering as stated in the introduction.
The details thereof are deferred to \cref{section:domination:and:covering}.

The notion of an \dombarset{a} can be defined in three different ways, (D1), (D2) and (D3),
	which are useful for different kind of proofs.
\definitionsWalkDomination

\newcommand{\lemmaTextDSequivalent}{
Conditions \ref{def:ds:1}, \ref{def:ds:2}, \ref{def:ds:3} are equivalent.} 

\begin{lemma}\hyperref[lemma:ds:equivalent]{\refstar}
\label{i:lemma:ds:equivalent}
\lemmaTextDSequivalent
\end{lemma}

We have the following two transformation
	of $\delta$-dispersed sets for different values of $\delta$,
	as shown by Hartmann et al.~\cite{HartmannLW22}.

\begin{lemma}[\cite{HartmannLW22}]
\label{lemma:cover:subdivide}
For every real $\delta > 0$ and integer $c\geq 1$,
	$\cover{\delta}(G)=\cover{c\delta}(G_c)$.
\end{lemma}

\begin{lemma}[\cite{HartmannLW22}]
\label{lemma:cover:translate}
$\cover{\delta}(G) + |E(G)| = \cover{\tfrac{\delta}{1+2\delta}}(G)$.
\end{lemma}

Aiming to translate these modifications to the realm of \dombarset{a} on $b$-subdivided graphs,
	we observe the following connection.

\begin{observation}\hyperref[lemma:dom:iff:cov]{\refstar}
\label{i:lemma:dom:iff:cov}
Let $k \in \NNN$.
There is a $b$-simple $\tfrac{a}{2b}$-covering set of size~$k$ of a graph $G$ without isolated vertices,
	if and only if there is an \dombarset{a} of $G_{b}$ of size $k$.
\end{observation}

A minimum $\tfrac{a}{b}$-cover $S$ can be assumed to be $b$-simple~\cite{HartmannLW22}.
Actually, if $b$ is even, we observation can be improved.
For example, a $\half$-covering set implies a $2$-simple $\half$-covering set of same size.

\begin{lemma}\hyperref[lemma:ocover:simple]{\refstar}
\label{i:lemma:ocover:simple}
Let $S$ be an $\tfrac{a}{b}$-cover of a graph $G$ for integers $a,b \in \NNN$.
Then there is an $\tfrac{a}{b}$-cover $S^\star$ of $G$ of size $|S^\star|=|S|$
	that is $2b$-simple and, if $b$ is a multiple of $2$, is $b$-simple.
	
Assuming that $S$ contains no point at a position $\tfrac{2i-1}{2b}$ for $i \in \NNN$,
	then $S^\star$ is $b$-simple,
	and, if additionally $b$ is a multiple of $2$,
	$S^\star$ is $\frac{b}{2}$-simple.
\end{lemma}

With this connection at hand, we can state a refined connection of minimum $\frac{a}{b}$-covers of a graph $G$
	and minimum \dombarset{a} on $b$-subdivided graphs.

\begin{corollary}
\label{i:lemma:ds:to:cov}
$\cover{\tfrac{a}{b}}(G) = \ogamma_{4a}(G_{2b}) = \ogamma_{4ca}(G_{2cb})$;
$\cover{\tfrac{a}{2b}}(G) = \ogamma_{2a}(G_{2b}) = \ogamma_{2ca}(G_{2cb})$,
	for any $a,b,c \in \NNN$ and graph $G$ without isolated vertices.
\end{corollary}

Now we can put the earlier stated transformation of $\delta$-dispersed sets
	in terms of \dombar{a} on $b$-subdivided graphs.

\newcommand{\lemmaTextDsSub}{
	$\ogamma_{a}(G_b) = \ogamma_{ca}(G_{cb})$ when $c$ is odd, or $a,b$ are even.}
\begin{theorem}\hyperref[a:lemma:ds:sub]{\refstar}
\label{lemma:ds:sub}
\lemmaTextDsSub
\end{theorem}

\newcommand{\lemmaTextDsTranslation}{
$\ogamma_a(G_b) + |E(G)| = \ogamma_{a}(G_{a+b})$.}

\begin{theorem}\hyperref[lemma:trans:ogamma]{\refstar}
\label{i:lemma:trans:ogamma}
\lemmaTextDsTranslation
\end{theorem}

These two transformation lay the groundwork
	for show \cref{lemma:i:ds:tw:domination}, \cref{lemma:i:ds:tw:covering}
	and \cref{lemma:cover:irrational}.
For details, we refer to \cref{section:domination:and:covering}.

\newpage

\newcommand{\tikzmark}[1]{}
\begin{figure}[h]
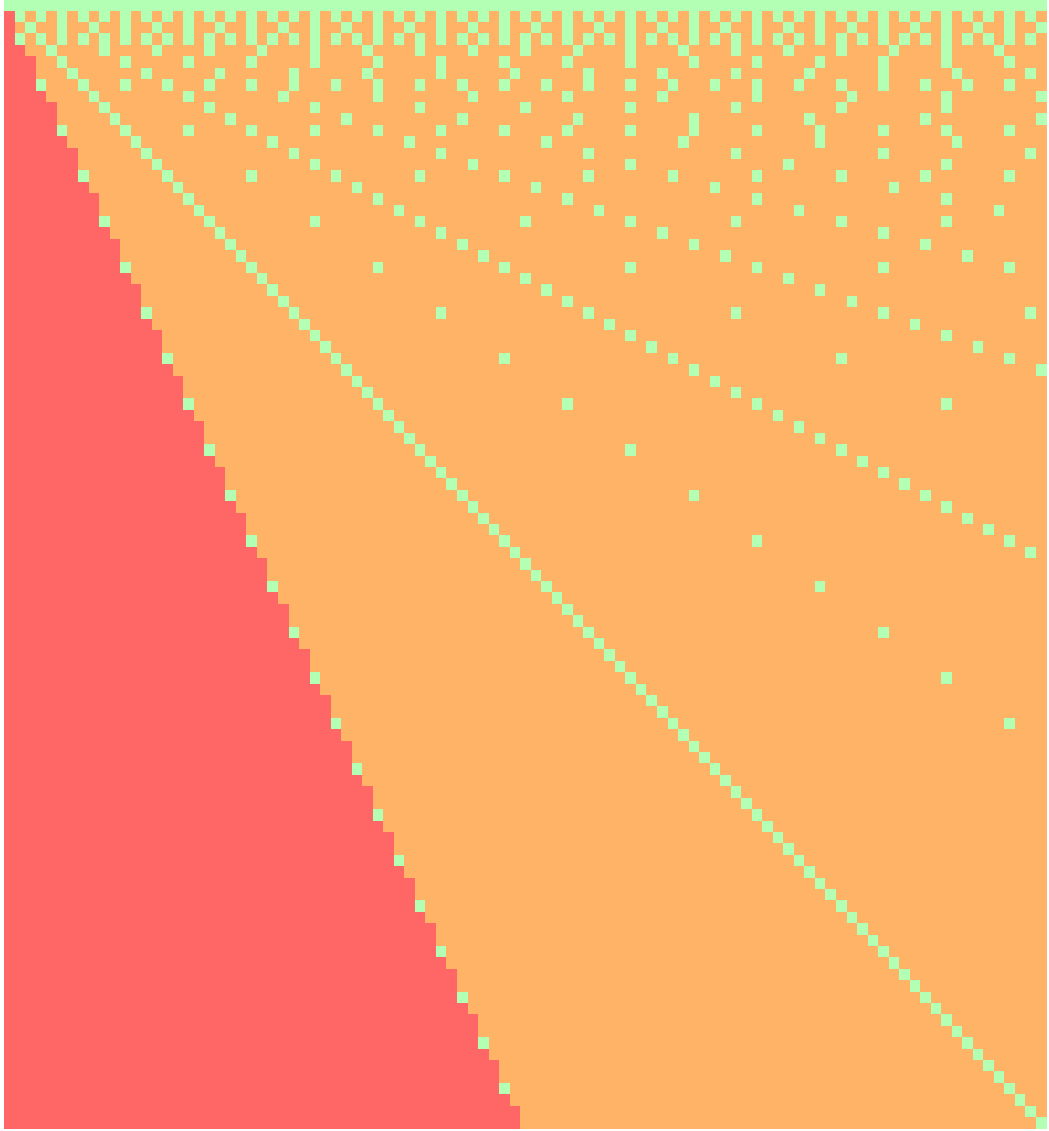

\centering
\resizebox{\columnwidth}{!}{%
\begin{tabular}{rccccccccccccccccccccccccccccccccccccccccccccccccccccccccccccccccccccccccccccccccccccccccccccccccccc}
 \tikzmark{test} &  	&  	&  	&  	&  	&  	&  	&  	&  	&  	&  	&  	&  	&  	&  	&  	&  	&  	&  	&  	&  	&  	&  	&  	&  	&  	&  	&  	&  	&  	&  	&  	&  	&  	&  	&  	&  	&  	&  	&  	&  	&  	&  	&  	&  	&  	&  	&  	&  	&  	&  	&  	&  	&  	&  	&  	&  	&  	&  	&  	&  	&  	&  	&  	&  	&  	&  	&  	&  	&  	&  	&  	&  	&  	&  	&  	&  	&  	&  	&  	&  	&  	&  	&  	&  	&  	&  	&  	&  	&  	&  	&  	&  	&  	&  	&  	&  	&  	&  	\\ 
$$ 	& \inp	& \inp	& \inp	& \inp	& \inp	& \inp	& \inp	& \inp	& \inp	& \inp	& \inp	& \inp	& \inp	& \inp	& \inp	& \inp	& \inp	& \inp	& \inp	& \inp	& \inp	& \inp	& \inp	& \inp	& \inp	& \inp	& \inp	& \inp	& \inp	& \inp	& \inp	& \inp	& \inp	& \inp	& \inp	& \inp	& \inp	& \inp	& \inp	& \inp	& \inp	& \inp	& \inp	& \inp	& \inp	& \inp	& \inp	& \inp	& \inp	& \inp	& \inp	& \inp	& \inp	& \inp	& \inp	& \inp	& \inp	& \inp	& \inp	& \inp	& \inp	& \inp	& \inp	& \inp	& \inp	& \inp	& \inp	& \inp	& \inp	& \inp	& \inp	& \inp	& \inp	& \inp	& \inp	& \inp	& \inp	& \inp	& \inp	& \inp	& \inp	& \inp	& \inp	& \inp	& \inp	& \inp	& \inp	& \inp	& \inp	& \inp	& \inp	& \inp	& \inp	& \inp	& \inp	& \inp	& \inp	& \inp	& \inp	\\
\end{tabular}
}
\caption{Complexity of $a$-\IndSet on a graph $G_b$.
The cell in the $a$-th row and $b$-th column represents $a$-\IndSet on a graph $G_b$.
A light green cell indicates a polynomial time solvable case,
	an orange \np-hardness \& \fpt,
	and a dark red \np-hardness \& \wone-hardness.
See \cref{figure:ind:cells} for a smaller excerpt of this table.}
\label{figure:ind:cells:big}
\end{figure}

\newpage

\begin{figure}[h]
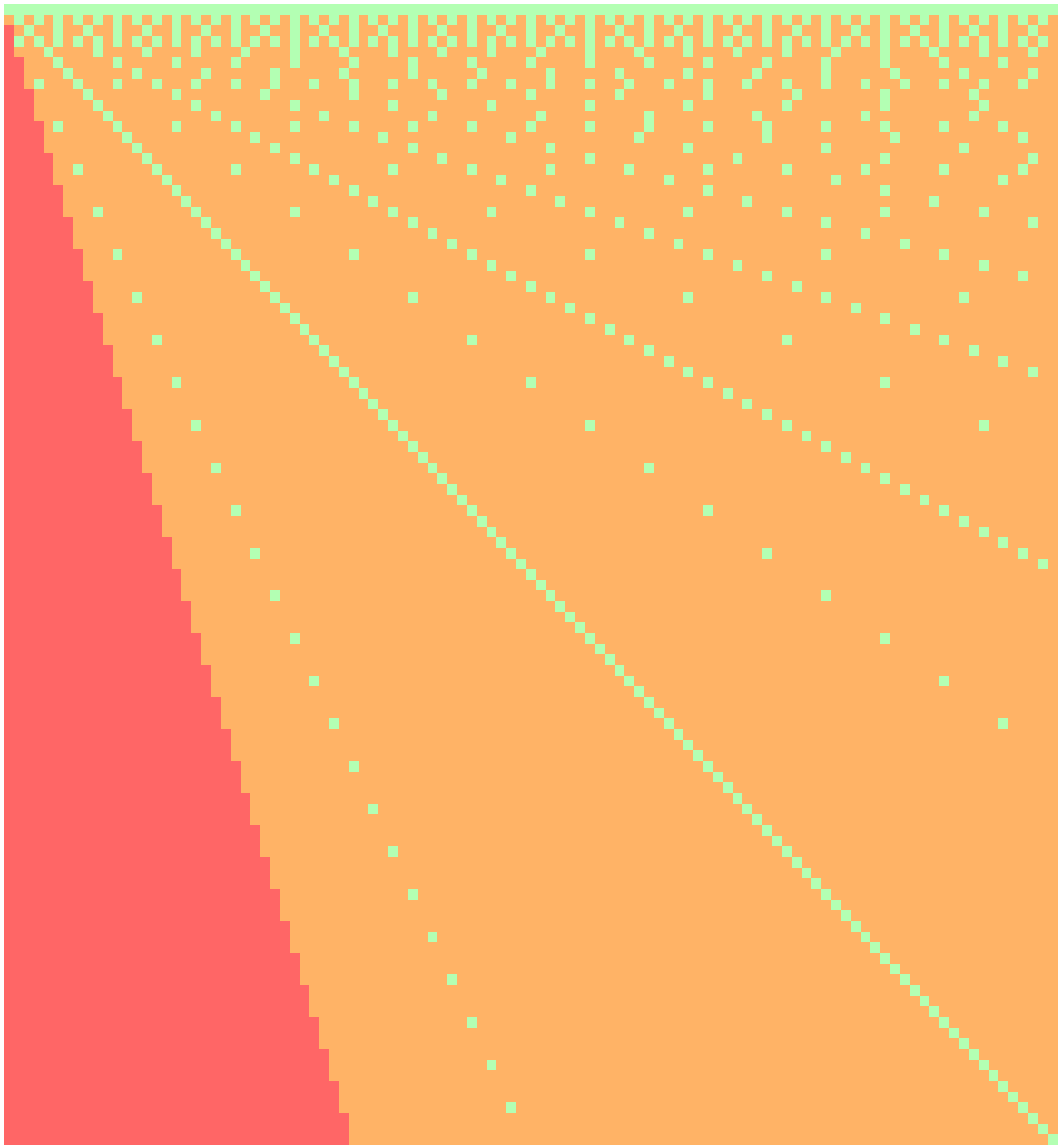

\centering
\resizebox{\columnwidth}{!}{%
\begin{tabular}{ccccccccccccccccccccccccccccccccccccccccccccccccccccccccccccccccccccccccccccccccccccccccccccccccccccccccccc}
\inp	& \inp	& \inp	& \inp	& \inp	& \inp	& \inp	& \inp	& \inp	& \inp	& \inp	& \inp	& \inp	& \inp	& \inp	& \inp	& \inp	& \inp	& \inp	& \inp	& \inp	& \inp	& \inp	& \inp	& \inp	& \inp	& \inp	& \inp	& \inp	& \inp	& \inp	& \inp	& \inp	& \inp	& \inp	& \inp	& \inp	& \inp	& \inp	& \inp	& \inp	& \inp	& \inp	& \inp	& \inp	& \inp	& \inp	& \inp	& \inp	& \inp	& \inp	& \inp	& \inp	& \inp	& \inp	& \inp	& \inp	& \inp	& \inp	& \inp	& \inp	& \inp	& \inp	& \inp	& \inp	& \inp	& \inp	& \inp	& \inp	& \inp	& \inp	& \inp	& \inp	& \inp	& \inp	& \inp	& \inp	& \inp	& \inp	& \inp	& \inp	& \inp	& \inp	& \inp	& \inp	& \inp	& \inp	& \inp	& \inp	& \inp	& \inp	& \inp	& \inp	& \inp	& \inp	& \inp	& \inp	& \inp	& \inp	& \inp	& \inp	& \inp	& \inp	& \inp	& \inp	& \inp	& \inp	\\ 
\infpt	& \inp	& \infpt	& \inp	& \infpt	& \inp	& \infpt	& \inp	& \infpt	& \inp	& \infpt	& \inp	& \infpt	& \inp	& \infpt	& \inp	& \infpt	& \inp	& \infpt	& \inp	& \infpt	& \inp	& \infpt	& \inp	& \infpt	& \inp	& \infpt	& \inp	& \infpt	& \inp	& \infpt	& \inp	& \infpt	& \inp	& \infpt	& \inp	& \infpt	& \inp	& \infpt	& \inp	& \infpt	& \inp	& \infpt	& \inp	& \infpt	& \inp	& \infpt	& \inp	& \infpt	& \inp	& \infpt	& \inp	& \infpt	& \inp	& \infpt	& \inp	& \infpt	& \inp	& \infpt	& \inp	& \infpt	& \inp	& \infpt	& \inp	& \infpt	& \inp	& \infpt	& \inp	& \infpt	& \inp	& \infpt	& \inp	& \infpt	& \inp	& \infpt	& \inp	& \infpt	& \inp	& \infpt	& \inp	& \infpt	& \inp	& \infpt	& \inp	& \infpt	& \inp	& \infpt	& \inp	& \infpt	& \inp	& \infpt	& \inp	& \infpt	& \inp	& \infpt	& \inp	& \infpt	& \inp	& \infpt	& \inp	& \infpt	& \inp	& \infpt	& \inp	& \infpt	& \inp	& \infpt	\\
\end{tabular}
}
\caption{Complexity of $a$-\DomSet on a graph $G_b$.
The cell in the $a$-th row and $b$-th column represents $a$-\DomSet on a graph $G_b$.
A light green cell indicates a polynomial time solvable case,
	an orange \np-hardness \& \fpt,
	and a dark red \np-hardness \& \wtwo-hardness.
See \cref{figure:ds:cells} for a smaller excerpt of this table.}
\label{figure:ds:cells:big}
\end{figure}

\newpage

\section{Details for Independence and Dispersion}
\label{section:detaill:independence}

This section complements the main part regarding independent set and dispersion.
First, \cref{a:section:lower:bound} gives the SETH lower bound on bounded treewidth graphs.
Then \cref{a:section:ind:parameterized} concerns the parameterized complexity in the solution size.
Further, \cref{a:section:irrational} completes the computational hardness proof
	for the fixed irrational distance $\delta$.
Finally, \cref{a:section:disp:translation} proves the translation lemma
	as stated in \cref{lemma:disp:translate:b:simple}.

\smallskip

Before we begin, let us make a general observation about $a$-independent sets.
There is a maximum independent set that contains every leaf of $G$
	that does not neighbor another leaf.
A similar statement holds for maximum $a$-independent sets.

\begin{observation}
\label{lemma:ind:leaf}
Let $G$ be a graph and $a\geq2$.
For every maximum $a$-independent $I$ set of $G$
	the set $I'$ resulting from $I$ as follows is still a maximum $a$-independent set:
For every leaf $u^\star \in U$
		where $B^{\ceil{a/2}}(u^\star)$ is a path of length $\ceil{\ahalf}$,
	replace any vertex in $B^{a-1}(u^\star)$ by $u^\star$.
\end{observation}

Also, we include here the omitted proof of \cref{section:ind:and:disp}.

\begin{corollary}[\cref{lemma:disp:translate:b} restated]
\label{a:lemma:disp:translate:b}
\lemmaTextDispTranslateB
\end{corollary}
\begin{proof}
For $b=1$, the statement coincides with \cref{lemma:disp:translate}.
For $b\geq2$, by induction, we have $\disp{\delta}(G) = \disp{\frac{\delta}{\delta + (b-1)}} - (b-1)|E(G)|$.
Since then $\frac{\delta}{\delta + (b-1)}<1$, \cref{lemma:disp:translate} applies again,
	such that $\disp{\delta}(G) = \disp{\gamma}(G) - b|E(G)|$ with
$$
\gamma
\,=\; \frac{ \frac{\delta}{1+(b-1)\cdot\delta} }{ 1+\frac{\delta}{1+(b-1)\delta} }
\;=\; \frac{ \frac{\delta}{1+(b-1)\cdot\delta} }{\frac{1+b\cdot\delta}{1+(b-1)\delta} }
\;=\; \frac{\delta}{1+ b \cdot \delta}
$$
\end{proof}

\subsection{Treewidth Lower Bound for Independent Set}
\label{a:section:lower:bound}
\label{section:is:ppseth:lb}

Here, we complement \cref{section:lower:bound}.
That is, for every even $a\geq 6$,
	we provide the lower bound under \SETH of computing a maximum $a$-independent set
	in graphs of bounded pathwidth,
	which even holds if the input is restricted to $2$-subdivided graphs without a cycle of length $<a$.
Actually, we show the statement
	assuming the Primal-Pathwidth \SETH (\ppseth),
	which was recently introduced by Lampis~\cite{Lampis2024pwseth}.
It is a weaker prerequisite than \SETH.
At the same time, using the \ppseth allows a nicer presentation
	by deferring many technicalities to a special intermediate
	Constraint Satisfaction Problem (CSP) problem.

\begin{theorem}
\label{a:lemma:is:ppset:lb}
Assume \ppseth and an $\varepsilon>0$.
Then, for every even $a\geq 6$,
	 $a$-\IndSet
	has no $(a-\varepsilon)^{\pw(G)} \cdot n^{O(1)}$ time algorithm,
	even when restricted to $2$-subdivided graphs without a cycle of length $<a$.
\end{theorem}

A first lower bound under \SETH was given by Katsikarelis et al.~\cite{KatsikarelisLP2022},
	which, however, does not apply to $2$-subdivided graphs.
Later, Lampis showed the same lower bound assuming only the \ppseth instead of the \SETH~\cite{Lampis2024pwseth}.
We adapt this construction of Lampis in two ways.
We make sure that we output a $2$-subdivided graph $G$,
	and that $G$ does not contain a cycle of length $<a$.

The basic idea is to deploy for every variable $i$ a long path
	where an $a$-independent set can contain at most every $a$-th vertex of that path
	and thereby encoding an assignment of variable $i$ to $[a]$.
For every constraint, we attach a gadget to these paths that enforces an encoding that satisfies this constraint. 
The crucial difference to the proof of Lampis \cite{Lampis2024pwseth}
	is that as we need to construct a $2$-subdivided graph,
	we may only attach such a gadget to the even positions of a long path.
We use two kinds of attachments,
	that together can enforce an encoding at even and odd positions of such a long path.

\smallskip

To introduce the mentioned intermediate problem (from~\cite{Lampis2024pwseth},
	let $a\geq 2$ and
	consider a $4$-CSP instance $\varphi$ on variables $[n]$, with clauses $[m]$ and alphabet $[a]$;
	that is for $\mu \in [m]$
	a $4$-variable set $X_\mu \subseteq [n]$ and
	a set $C_\mu$ of assignments from $X_\mu$ to $[a]$.
The \emph{primal graph} of $\varphi$ is the graph on the set of variables $[n]$
	with an edge $\{i,j\}$ if the variables $i,j$ occur together in some constraint of $\varphi$.
Let $(B_1,\dots,B_t)$ be a path decomposition of the primal graph of given~$\varphi$.
Let $f$ be an injective mapping $f: [m] \to [t]$
	such that $X_\mu \subseteq B_{f(m)}$ for every constraint $\mu \in [m]$.
A \emph{multi-assignment} is a mapping $\sigma: [n] \times [t] \to [a]$.
\begin{itemize}
\item $\sigma$ \emph{satisfies} $\varphi$ (relative to $(B_1,\dots,B_t)$ and $f$)
	if for every constraint $\mu \in [m]$,
	the mapping $\sigma$ with the second argument fixed to $f(\mu)$, that is
	the mapping $\sigma_\mu: [n] \to [a], i \mapsto \sigma(i, f(\mu))$, satisfies the constraint $C_\mu$.
\item $\sigma$ is \emph{monotone} decreasing (respectively, increasing)
	if for every variable $i \in [n]$ and $\tau_1 < \tau_2$ (respectively, $\tau_1 > \tau_2$)
	with $i \in B_{\tau_1} \cap B_{\tau_2}$
	we have $\sigma(i,\tau_1) \preceq \sigma(i,\tau_2)$.
\item $\sigma$ is \emph{consistent} for a variable $i \subseteq [n]$
	if for every $\tau_1,\tau_2 \in [t]$ where $i \in B_{\tau_1} \cap B_{\tau_2}$
	we have $\sigma(i,\tau_1)=\sigma(i,\tau_2)$.
\end{itemize}

\newcommand{\technicalProblem}{\textsc{Pathwidth $a$-Ary CSP}\xspace}
\newcommand{\technicalProblemLong}{\textsc{Pathwidth $a$-Ary CSP}, $a \geq 2$}
\problemdefSimple{\technicalProblemLong}{ %
A $4$-CSP instance $\varphi$ with alphabet $[a]$, a partition of the variables into $V_1,V_2$,
a width $p$ path decomposition $(B_1,\dots,B_t)$ of the primal graph of $\varphi$
	where each bag contains at most $O(a \log p)$ variables of $V_2$,
and an injective mapping $f: [m] \to [t]$
	such that $X_\mu \subseteq B_{f(m)}$ for every constraint $\mu \in [m]$.
}{
Is there a monotone decreasing and satisfying multi-assignment $\sigma$
	that is consistent for every variable in $V_2$?
}

\begin{lemma}[Lampis~\cite{Lampis2024pwseth}] %
\label{lemma:lampis:technical}
Let $a \geq 2$ and $\varepsilon > 0$.
If \technicalProblem\ can be solved in time
	$O( (a-\varepsilon)^{p} \cdot |\varphi|^{O(1)} )$, then \ppseth is false.
\end{lemma}

We give a reduction from
	\technicalProblem to $a$-\IndSet that outputs a $2$-subdivided graph without a cycle of length $<a$
	and of pathwidth at most $p+O(a^2 \log p + a^4)$.
Then, assuming \ppseth and an $(a-\varepsilon)^{\pw(G)} \cdot n^{O(1)}$ time algorithm for an $\varepsilon>0$,
	$a$-\DomSet yields a contradiction to \cref{lemma:lampis:technical}.
This then completes the proof of \cref{a:lemma:is:ppset:lb}.

\begin{lemma}
\label{lemma:is:ppseth:reduction}
\ineditnote{\cref{lemma:ds:ppseth:reduction}}
Let $a\geq 6$ be even.
There is a polynomial time reduction from
	\technicalProblem to $a$-\IndSet
	that outputs a $2$-subdivided graph $G$
	without a cycle of length $a$ and of pathwidth at most $p + O(a^2 \log p + a^4)$.
\end{lemma}

Before we begin with the actual construction,
	let us introduce two useful graph gadgets:
The \emph{super edge} forbids that both of its end vertices are selected
	and helps to avoid short cycles.
The \emph{blocker} attached to a vertex $u$ lets us assume that $u$ is \emph{not} part of the solution. The blocker 
is useful if we intend to attach an odd length path to a vertex,
	but such a path is not allowed in a $2$-subdivided graph.
Instead, we may add an even length path that extends to one more vertex $u$, and add a blocker on $u$.
Hence a blocker helps to construct a $2$-subdivided graph.
	
A \emph{super edge} $S_{u,v}$,
	for some vertices $u,v$,
	is a path $(u_0,\dots,u_{a-2})$ of length $a-2$ between $u_0=u$ and $u_{a-2}=v$
	where, for $i \in \{2,4,\dots,a-4\}$, we attach a path of length $a-2$ to $u_i$ and ending in some vertex 
	$u_i'$.
(We use here that $a \geq 6$.) The vertices of $S_{u,v}$ other than $u$ and $v$ are the \emph{inner} vertices of 
$S_{u,v}$.
We will use super edges $S_{u,v}$ in the construction of a graph $G$ in such a way that the inner vertices are not 
adjacent to any vertex outside $S_{u,v}$.
Observe that every inner vertex $u'_i$ is at distance at least $a$ to $u$ and $v$,
	hence they do not constrain the selection of vertices outside of $S_{u,v}$.
Thus, for a maximum $a$-independent set $I$, we may assume that
	$I$ contains the inner vertices $u'_i$ for $i \in \{2,4,\dots,a-4\}$
	and no other inner vertex.
Further, we may assume that $I$ contains at most one vertex from the closed neighborhood of $u$ and $v$
	that is not an internal vertex of the super edge.
        
A \emph{blocker} on a vertex $u$ is a cycle $C$ of even length $2a-2$ where one vertex coincides with $u$.
Then, similarly to \cref{lemma:ind:leaf},
	we may assume that a maximum independent set $I$ contains the unique vertices of $C$
	with distance $a-1$ to $u$.
As a result, we have that a maximum independent set $I$ of $G$
	does not contain $u$.
Now we are ready for the actual construction, see also \cref{figure:ind:seth:lb}.
\smallskip

\begin{figure}
\begin{center}
\begin{tikzpicture}[scale=0.4]

\newcommand{\drawEvenEncorcer}[1]{	\node[vertex] (q#1) at (\horshift+\qshift+#1,\vershift+4) {};
	\node[vertex] (qh#1) at (\horshift+\qshift+#1,\vershift+8) {};
	\draw[superedge] (p#1) -- (q#1);
	\draw[superedge] (q#1) -- (qh#1);}

\newcommand{\drawOddEncorcer}[1]{\node[vertex] (pi#1) at (\horshift+\pishift+#1,\vershift+4) {};
	\node[vertex] (pii#1) at (\horshift+\pishift+#1+0.5,\vershift+4) {};
	\node[vertex] (piii#1) at (\horshift+\pishift+#1+1,\vershift+4) {};
	\node[red] (piiicross#1) at (piii#1) {X};
	\draw (pi#1) -- (pii#1);
	\draw (pii#1) -- (piii#1);
	\node[vertex] (pih#1) at (\horshift+\pishift+#1,\vershift+8) {};
	\draw[superedge] (p#1) -- (pi#1);
	\draw[superedge] (pi#1) -- (pih#1);}

\newcommand{\crabConstruction}{
\def\start{6}
\def\endd{12}
\def\enddd{11}
\node[vertex, label={below:{$p^{\start}_{\subscript}$}}] (p\start) at (\horshift+\start,\vershift) {};
\drawEvenEncorcer{\start}
\drawOddEncorcer{\start}

\foreach \i [evaluate=\i as \ii using int(\i+1)] in {0,...,17}{
	\node[vertex] (p\i) at (\horshift+\i,\vershift) {};
	\node[vertex] (p\ii) at (\horshift+\ii,\vershift) {};
	\draw (p\i) -- (p\ii);	
}

\pgfmathsetmacro{\startt}{int(2+\start)}
\foreach \i [evaluate=\i as \j using int(\i+2), evaluate=\i as \ii using int(\i+1)]
	in {\start,\startt,...,\enddd}{
	\node[vertex] (p\j) at (\horshift+\j,\vershift) {};
	\node[vertex] (p\ii) at (\horshift+\ii,\vershift) {};
	\draw (p\i) -- (p\ii);
	\draw (p\ii) -- (p\j);

	\drawEvenEncorcer{\j}
	\drawOddEncorcer{\j}
}
\foreach \i in {0,\endd,18}{
	\node[vertex, label={below:{$p^{\i}_{\subscript}$}}] (p\i) at (\horshift+\i,\vershift) {};
}
}

\foreach \i [evaluate=\i as \angle using (90+\i*360/5)] in {1,...,5}{
	\ifthenelse{\i=2}{
		\node[vertex, shift={(9,10)}, label={175:$y_{1,\sigma_\i}$}] (y\i) at (\angle:3) {};
	}{
		\ifthenelse{\i=3}{
			\node[vertex, shift={(9,10)}, label={5:$\;y_{1,\sigma_\i}$}] (y\i) at (\angle:3) {};
		}{
			\node[vertex, shift={(9,10)}, label={\angle:$y_{1,\sigma_\i}$}] (y\i) at (\angle:3) {};
		}
	}
}
\foreach \i [evaluate=\i as \jbase using (\i+1)] in {1,...,4}{
	\foreach \j in {\jbase,...,5}{
		\draw[superedge] (y\i) -- (y\j);
	}
}
\foreach \todraw in {y1,y3,y4,y5}{
	\node[vertex, fill=white] (\todraw) at (\todraw) {};
}

\def\qshift{-4}
\def\pishift{4}

\def\horshift{-6}
\def\vershift{-4}
\def\subscript{1}
\crabConstruction

\draw[superedge] (qh6) -- (y2);
\draw[superedge] (qh10) -- (y2);
\draw[superedge] (qh12) -- (y2);

\foreach \todrawgray in {y2,p2,p8,p14,q6,qh8,q10,q12,pih6,pih12,pih8,pih10}{
	\node[vertex, fill=gray] (\todrawgray) at (\todrawgray) {};
}
\node[red] (p0) at (p0) {X};

\def\horshift{9.5}
\def\vershift{-1.5}
\def\subscript{2}
\crabConstruction

\draw[superedge] (qh6) -- (y2);
\draw[superedge] (qh12) -- (y2);

\draw[superedge] (pih8) -- (y2);
\draw[superedge] (pih10) -- (y2);

\foreach \todraw in {p3,p9,p15,q6,qh8,qh10,q12,pih6,pih12,pii8,pii10}{
	\node[vertex, fill=gray] (\todraw) at (\todraw) {};
}
\node[red] (p0) at (p0) {X};

\end{tikzpicture}
\end{center}
\caption{Construction for \cref{lemma:is:ppseth:reduction} for a very simple CSP instance with $a=6$.
There are only $2$ variables, both in $V_1$, a single bag and a single constraint with assignments 
$\sigma_1,\dots,\sigma_5$
	with $\sigma_2(1)=2$ and $\sigma_2(2)=3$.
(The construction for the other assignments is omitted).
The bold edge depict super edges, while there is a blocker on every vertex with a red cross.
Note that graph is a $2$-subdivision.
The gray marked vertices form a maximum $6$-independent set.
As $y_{1,\sigma_2}$ is selected (and because of how it is connected to below)
	only $p_1^{6+2}$ and $p_2^{6+3}$ (resembling the assignment $\sigma_2$)
	of $p_1^7,\dots,p_1^{12}$ and $p_2^7,\dots,p_2^{12}$, respectively.}
\label{figure:ind:seth:lb}
\end{figure}
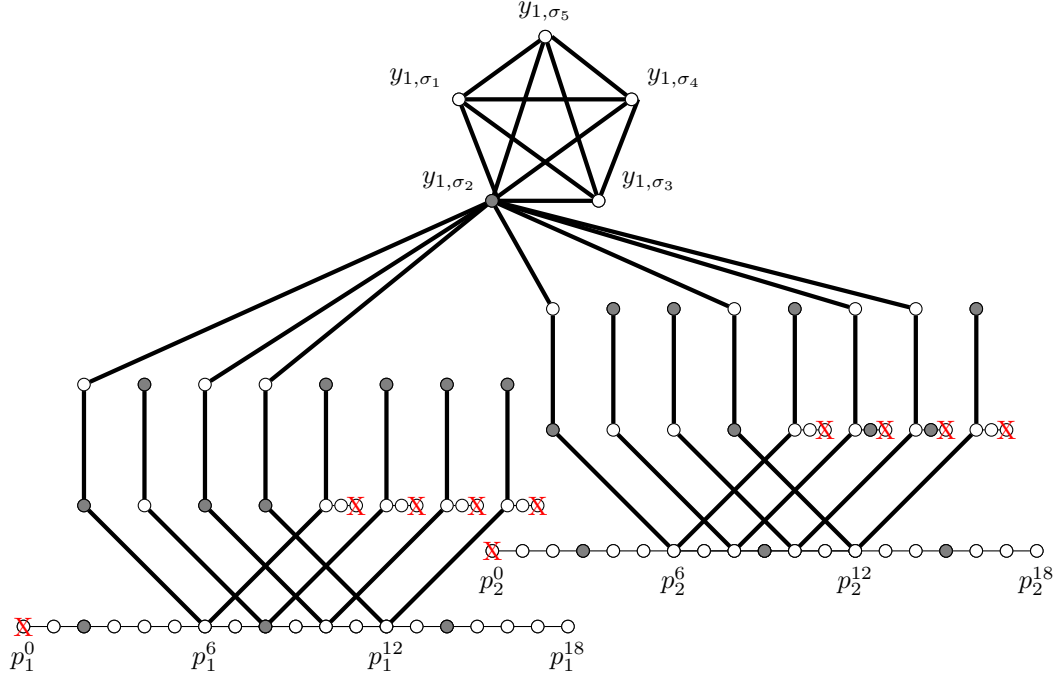

\textit{Construction:}
We construct a graph $G$ that is a $2$-subdivision of some graph $G'$,
	by using super edges $S_{u,v}$, paths $(p^0,\dots,p^j)$ of some even length $j$
	and connecting them only with vertices of even upper index to one another, and only blocking vertices of even 
	index.
\begin{enumerate}
\item
\label{is:step:v2}
For every variable $i \in V_2$ we proceed as follows.
We add a path $P_i = (p_i^0,\dots,p_i^{3a})$ of length $3a$
	with a blocker on $p_i^0$.
For every even position $\alpha = a + \alpha'$ with $\alpha' \in \{0,2,\dots,a\}$
	we add vertices $q_{i}^{\alpha}$, $\hat q_{i}^{\alpha}$, $\pi_{i}^{\alpha}$, $\hat 
	\pi_{i}^{\alpha}$,
	and super edges between the pairs
	$\{p_{i}^{\alpha}, q_{i}^{\alpha}\}$,
	$\{q_{i}^{\alpha}, \hat q_{i}^{\alpha}\}$,
	$\{p_{i}^{\alpha}, \pi_{i}^{\alpha}\}$ and
	$\{\pi_{i}^{\alpha}, \hat \pi_{i}^{\alpha}\}$,
	as well as a path $(\pi_{i}^{\alpha},\pi_{i}^{\alpha,1},\pi_{i}^{\alpha,2})$ of length $2$
	with a blocker on vertex $\pi_{i}^{\alpha,2}$.
	
For convenience, let $p_i^{3a\tau+\alpha}$ refer to the vertex $p_i^{\alpha}$
	for every $\tau\geq 1$ and $\alpha \in [3a]$.
\item
\label{is:step:v1}
For every variable $i \in V_1$, we proceed as follows.
Assume that $i$ occurs exactly in some $t_i$ bags $B_{\tau_i},\dots,B_{\tau_i+t_i-1}$.
We add a path $P_i = (p_{i}^{\tau_i a},\dots,p_i^{(\tau_i+t_i) 3a})$ of length $3at_i$
	with a blocker on $p_i^{\tau_i a}$.
For every even position $\alpha = (3\tau+1) a + \alpha'$ with $\alpha' \in\{0,2,\dots,a\}$
	we add vertices $q_{i}^{\alpha}$, $\hat q_{i}^{\alpha}$, $\pi_{i}^{\alpha}$, $\hat 
	\pi_{i}^{\alpha}$,
	and super edges between the pairs
	$\{p_{i}^{\alpha}, q_{i}^{\alpha}\}$,
	$\{q_{i}^{\alpha}, \hat q_{i}^{\alpha}\}$,
	$\{p_{i}^{\alpha}, \pi_{i}^{\alpha}\}$ and
	$\{\pi_{i}^{\alpha}, \hat \pi_{i}^{\alpha}\}$,
	and further, we add a path $(\pi_{i}^{\alpha,0},\pi_{i}^{\alpha,1},\pi_{i}^{\alpha,2})$	of length $2$
	with a blocker on vertex $\pi_{i}^{\alpha,2}$.

\item
\label{is:step:main}
We proceed for every constraint $\mu \in [m]$ as follows.
Let $\tau = f(\mu)$.
For every assignment $\sigma \in C_\mu$,
	we add a vertex $y_{\tau,\sigma}$.
For every pair of distinct assignments $\sigma,\sigma' \in C_\mu$,
	we add a super edge between $y_{\tau,\sigma}$ and $y_{\tau,\sigma'}$.
We proceed for every assignment $\sigma \in C_\mu$ and variable $i$ in $X_\mu$,
	the set of variables of constraint $\mu$, as follows.
We add a super edge between $y_{\tau,\sigma}$ and $\hat q_{i}^{\alpha}$
	for every position $\alpha = (3\tau+1) a + \alpha'$ where $\alpha' \in \{0,2,\dots,a\}$
	and $|\alpha' - \sigma(i)| \geq 2$.
Similarly, we add a super edge between $y_{\tau,\sigma}$ and $\hat \pi_{i}^{\alpha}$
	for every position $\alpha = (3\tau+1)a + \alpha'$ where $\alpha' \in \{0,2,\dots,a\}$
	and $|\alpha' - \sigma(i)| = 1$,
\end{enumerate}
Finally, we set the budget $k$ to $m + (5+2a) (|V_2|+\sum_{i \in V_1} t_i) + \frac{a-2}{2}s + s'$
	where $s$ is the number of added super edges
	and $s'$ is the number of added blockers.

\smallskip

By the construction, the output graph $G$ is a $2$-subdivision.
Let us further observe that~$G$ does not contain a cycle of length $<a$.
A blocker does not introduce a cycle of length $<a$.
Further, no added super edge $S_{u,v}$ between some vertices $u,v$ is part of a cycle of length $<a$
	since $u,v$ are never joined by a path outside $S_{u,v}$ of length $<2$.
The only remaining possible cycle of length $<a$ must contain a vertex $y_{\tau,\sigma}$,
	for some $\tau \in [t]$ and assignment $\sigma \in C_\mu$ with $f(\mu)=\tau$,
	which, however, is not possible as $y_{\tau,\sigma}$ is only connected via super edges.

Next, let us bound the pathwidth of the constructed graph $G$.
Notably, the graph $H$ resulting from removing every blocker
	and replacing every super edge between some vertices $u,v$
	by an edge $\{u,v\}$ has a pathwidth within an additive constant of the pathwidth of $G$.
Hence let us upper bound the pathwidth of $H$.

Consider a path decomposition $(B_1,\dots,B_t)$ of the primal graph of $\varphi$ of width $p$
	and where each bag contains at most $O( a \log p)$ vertices from $V_2$.
For every bag $\tau \in [t]$ and every variable $i \in V_2$,
	we add all the $O(a)$ vertices of step \ref{is:step:v2} for $\tau$ to a new bag $B_\tau'$.
Further, for every constraint $\mu \in [m]$ with $\tau = f(\mu)$,
	we add every vertex occurring in an edge in step \ref{is:step:main} to $B_\tau'$.
So far every new bag contains at most $O( a^2 \log p + a^{4} )$ vertices,
	as there are at most $a^4$ assignments in $C_\mu$.

For every bag $\tau \in [t]$, let $W_\tau \coloneqq V_1 \cap B_\tau$, which has size at most $p$.
Then let $(B_\tau^1,\dots,B_\tau^{t'})$
	be a straight forward path decomposition of width $p+O(1)$
	that contains, for $i \in W_\tau$,
	the subpath of $P_i$ from vertex $p_{\tau,i,0}^0$ to vertex 
	$p_{\tau,i,2}^{a}$,
	and that starts with a bag $B_\tau^1 = \{ p_{i}^{(3\tau+z)a+0} \mid i \in W_\tau \}$
	and that ends with a bag $B_\tau^{t'} = \{ p_{i}^{(3\tau+z)a+a} \mid i \in W_\tau \}$.
We extend $(B_\tau^1,\dots,B_\tau^{t'})$
	by adding the vertices from the trees attached to vertices of $P_i$ in step \ref{is:step:v1}.
We add every vertex of $B_\tau'$ to the bags $B_\tau^1,\dots,B_\tau^{t'}$.
Finally, joining the path decompositions $(B_\tau^1,\dots,B_\tau^{t'})$ for $\tau \in [t]$
	results in a path decomposition of $H$, and hence also for the constructed graph $G$,
	of width $p+O( a^2 \log p + a^{4} )$.

It remains to show the correctness of the construction.

\begin{lemma}
\label{lemma:is:ppseth:correctness}
There is a monotone increasing multi-assignment $\sigma$ that
	satisfies $\varphi$ and that is consistent for every variable in $V_2$,
	if and only if there is a size $\geq k$ $a$-independent set $I \subseteq V(G)$ in $G$.
\end{lemma}
\begin{proof}
\forward
Assume that there is a monotone increasing multi-assignment $\sigma$
	that satisfies $\varphi$ and that is consistent for every variable in $V_2$.
We define an independent set $I$ of $G$ as follows.
For every super edge $S_{u,v}$, we add the $\frac{a-4}{2}$ internal vertices $u_i'$ for $i \in \{2,4,\dots,a-4\}$.
For every blocker on some vertex $u$, we add the unique internal vertex that has distance $a-1$ to $u$.
For every variable $i \in V_2$,
	we add the vertices $p_i^{az+\alpha}$ for $z \in \{0,1,2\}$ and
	$\alpha$ being the consistent assignment of $\sigma(i,\tau)$ for every $\tau \in [t]$.
For every variable $i \in V_1$, node $\tau \in [t]$ and $z \in \{0,1,2\}$,
	we add the vertex $p_i^{az+\alpha}$ with $\alpha = \sigma(i,\tau)$.
Consider a variable $i \in V_1 \cup V_2$, a node $\tau \in [t]$, $z \in \{0,1,2\}$
	and a position $\alpha' \in \{0,2,\dots,a\}$.
If the difference $|\alpha' - \sigma(i,\tau)| \geq 2$, we add the vertex $q_i^{\alpha}$ to $I$,
	and else we add $\hat q_i^{\alpha}$;
	Similarly, if $|\alpha' - \sigma(i,\tau)| \geq 1$, we add the vertex $\pi_i^{\alpha,1}$ to $I$,
	and else we add $\hat \pi_i^\alpha$.
Finally, for every constraint $\mu \in [m]$ with $f(\mu)=\tau$,
	we add the vertex $y_{\tau,\sigma_\tau}$ to $I$,
	where $\sigma_\tau \in C_\mu$ is the restriction of $\sigma$ to the variables $X_\mu$ of constraint $\mu$.
This way, we have added exactly $k=m + (5+2a) (|V_2|+\sum_{i \in V_1} t_i) + \frac{a-2}{2}s + 2s'$ vertices to $I$.

We claim that $I$ is an $a$-independent set in $G$.
For every super edge $S_{u,v}$ between some vertices $u$ and $v$,
	the inner vertices of $S_{u,v}$ in $I$ have pair-wise distance at least $a$
	and distance at least $a$ to any vertex other than an inner vertex of $S_{u,v}$.
Similarly, for every blocker on some vertex $u$, as $I$ does not select $u$,
	we have that the selected inner vertex of that blocker has distance at least $a$
	to each other and to every other vertex in $I$.
Since $\sigma$ is monotone increasing, for every variable $i \in V_1 \cup V_2$,
	every pair of vertices of $I$ on path $P_i$, has distance at least $a$.
Consider a variable $i \in V_1 \cup V_2$, node $\tau \in [t]$, $z \in \{0,1,2\}$
	and a position $\alpha' \in \{0,2,\dots,a\}$.
The set $I$ contains either $q^{\alpha}_i$ or $\hat q^{\alpha}_i$ with $\alpha = (3\tau+1) a + \alpha'$,
	and, by our choice, the selected one has distance at least $a$ to the vertices $I \cap V(P_i)$.
Similarly, $I$ contains either $\pi^{\alpha}_i$ or $\hat \pi^{\alpha}_i$ with $\alpha = (3\tau+1) a + \alpha'$,
	and, by our choice, the selected one has distance at least $a$ to the vertices $I \cap V(P_i)$.
Finally, consider a constraint $\tau \in [t]$ with $f(\tau)=\mu$,
	where we have added $y_{\tau,\sigma_\tau}$ to $I$
	where $\sigma_\tau \in C_\mu$ is the restriction of $\sigma$ to $X_\mu$.
For every variable $i \in X_\mu$ and position $\alpha' \in \{0,2,\dots,a\}$ defining $\alpha = (3\tau+1)a+\alpha'$,
	by our choice of $I$,
	we have that either $\hat q_i^{\alpha} \notin I$, or $\hat q_i^{\alpha}$ and $y_{\tau,\sigma_\tau}$
	are not connected by a super edge
	and hence have distance at least $a$.
Analogously, either $\hat \pi_i^{\alpha} \notin I$, or $\hat \pi_i^{\alpha}$ and $y_{\tau,\sigma_\tau}$
	are not connected by a super edge
	and hence have distance at least $a$.
We easily observe that the remaining pairs of vertices in $I$ have distance at least $a$
	regardless of the assignment $\sigma$.
Thus we conclude that $I$ is an $a$-independent set of $G$ of size $k$.

\smallskip

\backward
Assume that there is a maximum $a$-independent set $I$ of size at least $k$ in $G$.
For every super edge,
	exactly $\frac{a-2}{2}$ of its inner vertices are contained in $I$.
Similarly, we may assume that every blocker on some vertex $u$
	contains exactly the unique inner vertex that has distance $a-1$ to $u$.
Hence a budget of $m + (5+2a) (|V_2|+\sum_{i \in V_1} t_i)$ remains for the other vertices.

We observe that, for every variable $i \in V_2$, the subpath $(p_i^{za+1},\dots,p_i^{(z+1)a})$ for $z \in 
\{0,1,2\}$,
	contains at most one vertex from $I$.
Similarly, for every variable $i \in V_1$, node $\tau \in [t]$ and $z \in \{0,1,2\}$,
	the subpath $(p_i^{(3\tau+z)a+1},\dots,p_i^{(3\tau+z+a)})$ contains at most one vertex from $I$.
For every added super edge
	between $q_{i}^{\alpha}$ and $\hat q_{i}^{\alpha}$ at most one of them is in $I$.
For every added super edge
	between $\pi_{i}^{\alpha}$ and $\hat \pi_{i}^{\alpha}$
	at most one of $\pi_i^\alpha$, $\pi_{i}^{\alpha,1}$ (the non-super edge neighbor of $\pi_i^\alpha$)
	and $\hat \pi_{i}^{\alpha}$ is in $I$.
Finally, for every constraint $\mu \in [m]$,
	set $I$ contains at most one vertex in $\{ y_{\tau,\sigma_\tau} \mid \sigma_\tau \in C_\mu \}$.
By the budget limit, for each of the above statement about containing at most one vertex from $I$,
	the same statement holds claiming containment of exactly one vertex from $I$.

Thus we can now define the multi-assignment $\sigma$.
For every variable $i \in V_2$ and $\tau \in [t]$,
	let $\sigma(i,\tau) = \alpha$
	where $p_i^{a+\alpha}$ is the unique vertex in $I$ of the subpath $(p_i^{a+1},\dots,p_i^{2a})$.
By definition, this assignment is consistent for every variable of $V_2$.
For every variable $i \in V_1$ and $\tau \in [t]$,
	let $\sigma(i,\tau)=\alpha$
	where $p_i^{(3\tau+1)a+\alpha}$ is the unique vertex in $I$
	of the subpath $(p_i^{(3\tau+1)a+1},\dots,p_i^{(3\tau+2)a})$.
This assignment is monotone increasing for variable $i \in V_1$,
	since, for every $\tau \in [t-1]$,
	between the vertices $p_i^{(3\tau+1)a+\alpha}$ and $p_i^{(3\tau+4)a+\alpha'}$
	with $\alpha = \sigma(i,\tau)$ and $\alpha' = \sigma(i,\tau+1)$
	there are exactly $2$ vertices of $I$ on path $P_i$.

Finally, we observe that every constraint $\mu \in [m]$ is satisfied.
Let $\tau = f(u)$.
Let $\sigma_\tau$ be the local assignment
	such that $y_{\tau,\sigma_\tau}$
	is the unique vertex of $\{y_{\mu,\sigma_\tau} \mid \sigma_\tau \in C_\tau\}$ in $I$.
We claim that $\sigma(i,\tau)=\sigma_\tau$ for every variable $i \in X_\mu$.
Indeed, if $\sigma(i,\tau) = \alpha'$ is even,
	that is $p_i^\alpha \in I$ with $\alpha= (3\tau+1) a + \alpha'$
	we have $\hat q_i^\alpha \in I$.
Hence the difference of the assignments is $|\alpha' - \sigma_\tau(i)| \leq 1$.
Further, we have $\hat \pi_i^\alpha \in I$ and hence $|\alpha' - \sigma_\tau(i)| \neq 1$,
	such that we conclude $\alpha' = \sigma_\tau(i)$.
On the other hand, if $\sigma(i,\tau) = \alpha'$ is odd,
	we may assume that $\pi_i^{\alpha-1,1} \in I$ and $\pi_i^{\alpha+1,1} \in I$.
Then $|(\alpha'-1) - \sigma_\tau(i)| \leq 1$ and $|(\alpha'+1) - \sigma_\tau(i)| \leq 1$,
	which implies that $\alpha' = \sigma_\tau(i)$.
In conclusion, $\sigma$ satisfies $\varphi$.
\end{proof}

\subsection{Parameterized Complexity of Independent Set}
\label{a:section:ind:parameterized}

This section contains the missing proofs from \cref{section:ind:parameterized}.

\begin{lemma}[\cref{lemma:ind:fpt} restated]
\label{a:lemma:ind:fpt}
\lemmaTextIndFpt
\end{lemma}
\begin{proof}
Let the input be a graph $G$ (defining $G_b$) and an integer $k$ as the solution size asked for.
We determine $\nu(G)$, the size of maximum matching in $G$, in polynomial time.
If $k \leq \nu(G)$, by \cref{lemma:ind:leq:matching},
	we may immediately answer `yes'.
Otherwise $k > \nu(G) \geq \frac{\vc(G)}{2}$,
	where $\vc(G)$ is the minimum size of a vertex cover of $G$.
Further, the size of a vertex cover upper bounds the treewidth of $G$.
We may compute a tree decomposition of $G$ in \fpt time~\cite{KorhonenLokshtanov2023},
	which immediately provides a tree decomposition of $G_b$ of same size.
If $a\geq 2$, we may compute a maximum $a$-independent set in \fpt time
	for the treewidth as parameter by using \cref{lemma:a:is:tw}.
Else, $a=1$ and we have $V(G_b)$ as a maximum $1$-independent set.
\end{proof}

We settle the hardness results by two parameter preserving reductions from \indset
	showing \wone-hardness,
	the first for $\frac{a}{b} \in (2,3)$, the second for $\frac{a}{b} \geq 3$;
	similarly as in~\cite{HartmannL22}.

\begin{lemma}%
\label{a:lemma:ind:2:hardness}
\lemmaTextIndTwoHardness
\end{lemma}
\begin{proof}
We show \wone-hardness by a reduction from $\IndSet$ on general graphs with the solution size $k$ as parameter.
Given a graph $G$ and integer $k$ we construct a graph $G' \in \GG_b$ and set $k'=k$
	such that $G$ has an independent set of size at most $k$
	if and only if $G'$ has an independent set of size at most $k'$.
It is easy to see that the following construction works in polynomial time.

\smallskip

\emph{Construction:}
First, we construct an auxiliary graph $G''$.
For every vertex $u\in V(G)$, we add vertices $u_1,u_2$ and edge $\{u_1,u_2\}$ to $G''$.
For every edge $\{u,v\}\in E(G)$, we add edge $\{u_i,v_j\}$ for every $i,j \in \{1,2\}$ to $G''$.
This concludes the construction of $G''$.
Then let $G'$ be the $b$-subdivision of $G''$.
Recall that $\vv(u,v,\beta)$, with $\{u,v\} \in E(G)$ and $\beta \in \{0,\dots,b\}$,
	denotes the unique vertex on the unique shortest $u,v$-path in $G'$
	which has distance $\beta$ to $u$ and distance $b-\beta$ to $v$.
By setting $k'=k$, this completes the construction.

\smallskip

To show correctness,
	for the forward direction,
	consider an independent set $I$ of $G$.
We claim that $I'$, consisting of $\vv(u_1,u_2,\floor{\frac{b}{2}})$ for every $u\in I$,
	is an $a$-independent set of $G'$.
As distinct vertex $u,v \in I$ have distance at least two in $G$,
	the sets $\{u_1,u_2\}$ and $\{v_1,v_2\}$ have distance at least $2b$ in $G_b$,
	and hence $\vv(u_1,u_2,\floor{\frac{b}{2}})$ and $\vv(u_1,u_2,\floor{\frac{b}{2}})$
	have distance at least $3b-1$.
With $\frac{a}{b} < 3$, hence $a < 3b$, and $a,b$ integers,
	the distance is at least $a$.
Since $|I'|=|I|$, this concludes the forward direction.

For the backward direction,
	consider an $a$-independent set $I'$ of $G'$.
For every $u \in V(G)$,
	we define the ball $B(u)$ as the set of vertices in $G'$
	with distance at most $\frac{b}{2}$ to vertex $u_1$ or~$u_2$.
We note that every vertex in $V(G')$ is contained in a ball $B(u')$ for some $u \in V(G)$.
We define $I \subseteq V(G)$ by adding,
	for every $u' \in I'$, a vertex $u \in V(G)$ to $I$ where $u' \in B(u)$.
Note that $|I \cap B(u)|\leq 1$, for every $u \in V(G)$,
	as any two vertices in $B(u)$ have distance at most $a > 2b$.
Hence $|I|\leq|I'|$.
We claim that for every edge $\{u,v\}\in E(G)$, either $u \notin I$ or $v \notin I$.
Assuming otherwise, consider vertices $u' \in B(u)$ and $v' \in B(v)$
	where we have added $u$ and $v$ to $I$.
Then, in $G'$, vertex $u'$ has distance at most $\frac{b}{2}$ to $\{u_1,u_2\}$,
	$v'$ has distance at most $\frac{b}{2}$ to $\{v_1,v_2\}$,
	and $u_i$ has distance to $v_j$ of $b$ for any $i,j \in \{1,2\}$.
In total, $u',v'$ have distance at most $2b$ in $G'$.
Since $\frac{a}{b} > 2$ and hence $a>2b$, we have the contradiction that $u,v' \in I'$.
In conclusion, $I$ is an independent set in~$G$.
\end{proof}

\begin{lemma}%
\label{a:lemma:ind:3:hardness}
\lemmaTextIndThreeHardness
\end{lemma}
\begin{proof}
We show \wone-hardness by a reduction from $\IndSet$ on general graphs with the solution size $k$ as parameter.
Given a graph $G$ and integer $k$ we construct a graph $G' \in \GG_b$ and set $k'=k$
	such that $G$ has an independent set of size at most $k$
	if and only if $G'$ has an independent set of size at most $k'$.
It is easy to see that the following construction works in polynomial time.

\smallskip

\emph{Construction:}
First, we construct an auxiliary graph $G''$.
We begin with a $2$-subdivision $G_{2}$ of $G$
	and making $V(G_2)\setminus V(G)$ to a clique.
For every vertex $u \in V(G)$, we add a path of length $\ceil{\frac{a}{2b}}-2$
	from $u$ to some new vertex $u_1$.
If $\ceil{\frac{a}{b}}$ is even, we add a vertex $u_2$ adjacent to $u_1$ and
	every neighbor of $u_1$ in the constructed graph so far.
This concludes the construction of $G''$.
Then let $G'$ be the $b$-subdivisions of $G''$.
Again, recall that $\vv(u,v,\beta)$, with $\{u,v\} \in E(G)$ and $\beta \in \{0,\dots,b\}$,
	denotes the unique vertex on the unique shortest $u,v$-path in $G'$
	which has distance $\beta$ to $u$ and distance $b-\beta$ to $v$.
By setting $k'=k$, this completes the construction.

\smallskip

To show correctness,
	for the forward direction,
	consider an independent set $I$ of $G$.
In case that $\ceil{\frac{a}{b}}$ is odd,
	let $I'$ consist of $u' = u_1$ for every vertex $u \in I$;
	and in case that $\ceil{\frac{a}{b}}$ is even,
	let $I'$ consist of $u' = \vv(u_1,u_2,\floor{\frac{b}{2}})$ for every vertex $u \in I$.
We claim that $I'$ is an independent set of $G'$.
Consider distinct vertices $u',v' \in I'$.
Then the corresponding vertices $u,v \in I$
	have distance $2$ in $G$,
	and hence distance $3$ in the preliminary graph $G''$.
In case $\ceil{\frac{a}{b}}$ is odd,
	$u',v'$ have distance
	$3b+ 2(\ceil{\frac{a}{2b}}-2)b \geq 2+ \frac{a}{b} = 2\ceil{ \frac{a+1}{2b} }b -2b \geq a$.
In case $\ceil{\frac{a}{b}}$ is even,
	$u',v'$ have distance
	$3b+ 2(\ceil{\frac{a}{2b}}-{2})b +2(\half b) \geq a$.
Thus $I'$ is an $a$-independent set in $G'$.

For the backward direction,
	consider an independent set $I'$ of $G'$.
Fore every vertex $u \in V(G)$,
	let the ball $B(u')$ be the set of vertices with distance at most
	$(\ceil{\frac{a}{2b}}-2)b + b + \floor{\frac{b}{2}} = \ceil{\frac{a}{2b}}b - b + \floor{\frac{b}{2}}$
	to $u'$.
We note that every vertex in $V(G')$ is contained in a ball $B(u')$ for some $u \in V(G)$.
We define $I \subseteq V(G)$ by choosing,
	for every $u' \in I'$, a vertex $u \in V(G)$ where $u' \in B(u)$ and adding $u$ to $I$.
Note that $|I \cap B(u)|\leq 1$, for every $u \in V(G)$,
	as any two vertices in $B(u)$ have distance at most $\ceil{\frac{a}{2b}}b - b + \floor{\frac{b}{2}}$.
Hence $|I|\leq|I'|$.
We claim that for every edge $\{u,v\}\in E(G)$, either $u \notin I$ or $v \notin I$.
Assuming otherwise, consider vertices $u^\star \in B(u)$ and $v^\star \in B(v)$
	where we have added $u$ and $v$ to $I$.
We note that $u^\star, v^\star$ have distance at most that of $u',v'$ in $G'$.
Vertices $u,v$ have distance $2$ in $G''$.
In case $\ceil{\frac{a}{b}}$ is odd,
	$u',v'$ have distance $2(\ceil{\frac{a}{2b}}-2)b+2b \leq a-2b < a$ in $G'$.
In case $\ceil{\frac{a}{b}}$ is even,
	$u',v'$ have distance $2(\ceil{\frac{a}{2b}}-2)b+3b \leq a-b < a$ in $G'$.
Both cases contradict that that $u',v' \in I$ and hence that $u^\star,v^\star \in I'$.
We conclude that $I$ is indeed an independent set of $G$.
\end{proof}

Above \cref{a:lemma:ind:2:hardness} and \cref{a:lemma:ind:3:hardness} combined
	yield the following result.

\begin{lemma}[\cref{lemma:ind:hardness} restated]
\label{a:lemma:ind:hardness}
\lemmaTextIndHardness
\end{lemma}

\subsection{Dispersion with Irrational Distance}
\label{a:section:irrational}

This section provides the missing proof of \cref{section:irrational},
	which is the proof of the following reduction,
	which serves as the first step for the proof of \cref{i:lemma:disp:wone:pw}.

\newcommand{\edgehinted}[2]{
\draw[] (#1) edge ($(#1)!0.10!(#2)$) edge [dotted] ($(#1)!0.20!(#2)$);
}
\begin{figure}
\begin{center}
\begin{tikzpicture}

\node[vertex,label={[yshift=-0.3cm]left:{${\hat w}^{1,2}_{j,j'}$}}] (wjj) at (-1.5,-2) {};
\node[vertex,rectangle,fill=gray,label=left:{$w^{1,2}_{j,j'}$}] (wjjstar) at (-1.5,-3.5) {};
\draw (wjj) -- (wjjstar) node[midway,fill=white] {$14n$};
\node[vertex,label={[yshift=-0.3cm]right:{$\hat w^{1,2}_{j'',j'''}$}}] (wjjmore) at (3,-2) {};
\node[vertex,rectangle,fill=gray,label=left:{$ $}] (wjjmorestar) at (3,-3.5) {};
\draw (wjjmore) -- (wjjmorestar) node[midway,fill=white] {$14n$};
\node[vertex,label={[yshift=-0.3cm]left:{$w^{1,2}$}}] (w) at (1,-3) {};
\draw (wjj) -- (w) node[midway,fill=white] {$12n$};
\draw (wjjmore) -- (w) node[midway,fill=white] {$12n$};
\node[vertex,rectangle,label=left:{$w^{1,2}_\star$}] (wstar) at (1,-4.5) {};
\draw (w) -- (wstar) node[midway,fill=white] {$38n$};

\foreach \i [evaluate=\i as \pos using int((\i*2)-3)*4)] in {1,2}{
	\ifthenelse{\i=1}{\def\jlabel{j}}{\def\jlabel{j'}}
	\node[vertex,rectangle,label=left:{$a^\i_\star$}] (u\i star) at (\pos-2,1.5) {};
	\node[vertex,label=left:{$a^\i$}] (u\i) at (\pos-2,0) {};
	\node[vertex,label=right:{$b^\i$}] (v\i) at (\pos+2,0) {};
	\node[vertex,rectangle,label=right:{$b^\i_\star$}] (v\i star) at (\pos+2,1.5) {};
	\draw (u\i) -- (v\i) node[midway,fill=white] {$46n$};
	\draw (u\i) -- (wjj) node[midway,fill=white] {$30n-6\jlabel$};
	\draw (v\i) -- (wjj) node[midway,fill=white] {$24n+6\jlabel$};
	\foreach \u/\v in {u\i star/u\i,v\i/v\i star}{
		\draw (\u) -- (\v) node[midway,fill=white] {$44n$};
	}
	
	\edgehinted{wjjmore}{u\i};
	\edgehinted{wjjmore}{v\i}
}

\end{tikzpicture}
\end{center}
\caption{A sketch for the construction of \cref{a:lemma:clique:leq:disp}.
Each edge represents a path of length as indicated.
The construction for color classes $1,2$ and for an edge $\{u_j^1,u_{j'}^2\}$ is depicted
	(and hinted for an edge $\{u_{j''}^1,u_{j'''}^2\}$).
The white square-shaped vertices can be assumed to be in any $64n$-independent set.
The gray square-shaped vertices exclude each other, and encode the edge for color classes $1,2$.}
\label{figure:disp:reduction:pathwidth}
\end{figure}
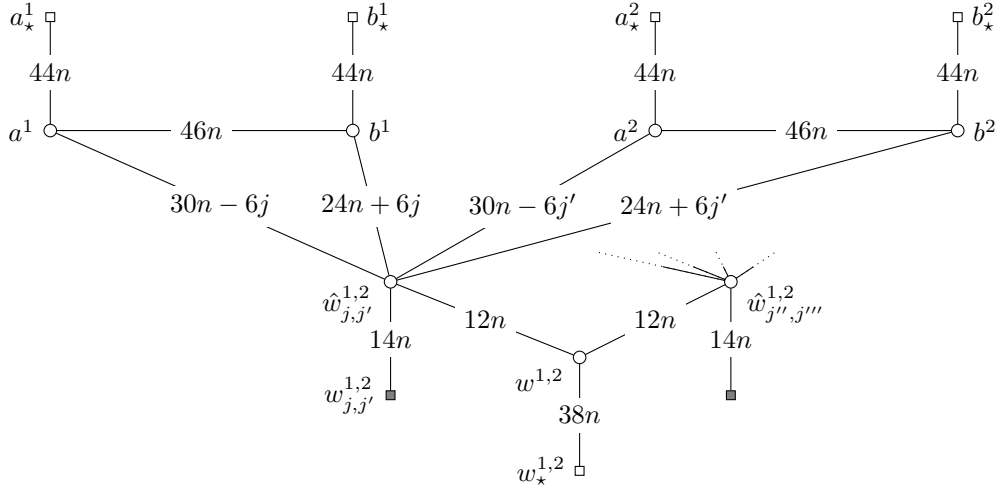

\begin{theorem}[\cref{lemma:clique:leq:disp} restated]
\label{a:lemma:clique:leq:disp}
\ineditnote{\cref{lemma:gt:leq:cover}.}
There is a polynomial time reduction that,
	given a \ColorfulClique-instance $\langle G,k\rangle$ with $k$ color classes of size $n$,
	a graph $G'$ of pathwidth $\Oh(k)$ and integer $k'$, such that: %
$\langle G,k\rangle$ is a yes-instance of \ColorfulClique
	if and only if $\langle G', k'\rangle$ is a yes-instance of $32n$-$\dispersionstar$
	if and only if $\langle G', k'\rangle$ is a yes-instance of $(32n-1)$-$\dispersionstar$.
\end{theorem}
\begin{proof}
We prove the above statement for $64n$-\IndSet and $(64n-2)$-\IndSet
	on $2$-subdivided graphs.
Our output graph $G_2'$ is a $2$-subdivision of a graph $G'$, of pathwidth $\Oh(k)$
	and contains no a cycle of length $<64n-2$.
By \cref{lemma:is:to:disp}, outputting graph $G'$
	yields the claimed reduction to $32n$-$\dispersionstar$ and $(32n-1)$-$\dispersionstar$.
Particularly, the pathwidth of $G'$ is at most the pathwidth of its $2$-subdivision.
Further, since $G_2'$ contains no cycle of length $64n-2$,
	so does $G'$ contain no cycle of length $32n-1$,
	as required.

\medskip

\textit{Construction:}
We set the budget to budget $k'= 2h + 2k$ and construct a graph $G'_2$ constructed as follows.
See \cref{figure:disp:reduction:pathwidth} for an illustration.
\begin{itemize}
\item 
For every color class $i\in[k]$, we add vertices $a^i,a^i_\star,b^i,b^i_\star$,
	as well as a $(a^i,a^i_\star)$-path of length $44n$, %
	a $(b^i,b^i_\star)$-path of length $44n$,
	and a $(a^i,b^i)$-path of length $46n$. %
\item
Let the $n$ vertices of color class $V_i$ to be enumerated as $v_1^i,v_2^i,\dots,v_{n}^i$.
Conveniently, we let $v_{j/6}^i$ for $j\in [6n]$ also be vertex
	on the $(a^i,b^i)$-path with distance $20n+6j$ to $a_i$ %
	(and hence with distance $26n-6j$ to $b^i$).
\item
For every $i,i'\in[k]$ with $i<i'$ and edge $\{u^{i}_{j}, u^{i'}_{j'}\} \in E(G)$,
	we add a path between new vertices $w^{i,i'}$ and $w^{i,i'}_\star$ of length $38n$.
\item
For every $i,i'\in[k]$ with $i<i'$ and edge $\{u^{i}_{j}, u^{i'}_{j'}\} \in E(G)$,
	we connect $w^{i,i'}$ to a new vertex $\hat w^{i,i'}_{j,j'}$ by a path of length $14n$.
We connect $\hat w^{i,i'}_{j,j'}$ to $u^{i}$ by a path of length $30n-6 j$
	as well as $\hat w^{i,i'}_{j,j'}$ to $v^{i}$ by a path of length $24n + 6 j$.
Symmetrically, we connect $\hat w^{i,i'}_{j,j'}$ to $u^{i'}$ by a path of length $30n-6 j'$
	as well as $\hat w^{i,i'}_{j,j'}$ to $v^{i'}$ by a path of length $24n + 6 j'$.
\end{itemize}

\smallskip

This concludes the construction.
It is easy to see, that $\langle G'_2,k'\rangle$ is polynomial time computable given $\langle G,k\rangle$.
Further, the output graph $G_2'$ is indeed a $2$-subdivision of a graph~$G'$.
For the correctness, it suffices to show that
	\forward a colorful $k$-clique of $G$ implies an $64n$-independent set of $G'_2$ and of size $k'$;
	and \backward a $(64n-2)$-independent set of $G'_2$ and of size $k'$ implies a $k$-clique of $G$.

\forward
For the forward direction,
	consider a colorful $k$-clique.
That is, there are vertices $u^i_{j_i}$ of color class $i$ for $i \in[k]$
	such that for all $i,i'\in [k]$ with $i<i'$
	there is an edge $\{u^{i}_{j_i}, u^{i'}_{j_{i'}}\}\in E(G)$.
We observe that
$$ \bigcup_{i \in [k]} \{ a^i_\star, u^i_{j_i} ,b^i_\star \}
	\;\cup \bigcup_{1\leq i < i' \leq k} \left\{ w^{i,i'}_{{j_{i},j_{i'}}} , w^{i,i'}_\star \right\}$$
	is a $64n$-independent set of size $3k + 2\binom{k}{2}$,
	hence especially a $(64n-2)$-independent set.
Particularly, $w^{i,i'}_{j_i,j_{i'}}$ has a shortest path to $u^i_{j_i}$ either via $a^i$ or $b^i$,
	for $i,i'\in[k]$ with ${i<i'}$.
The former has length $14n+30n-6j+20n+6j=64n$,
	and the latter has length $14n+24n+6j+20n+6n-6j=64n$.

\medskip

\backward
Consider a $(64n-2)$-independent set $I$ of $G'_2$ 
	of size $k' = 3k + 2\binom{k}{2}$.
We observe that $a_\star^i,b_\star^i$, for $i\in [k]$
	and $w_\star^{i,i'}$, for $i,i'\in [k], i<i'$,
	are attached to the remaining graph by a path of length at least $44n/2\geq \ceil{\alpha/2}$.
Hence, by \cref{lemma:ind:leaf}, we may assume that $I$
	contains at least the $2k+\binom{k}{2}$ vertices
		$I^\star \coloneqq \bigcup_{i \in [k]} \{ a^i_\star, b^i_\star \} \cup
		\bigcup_{1\leq i<i'\leq k} \{ w^{i,i'}_\star \}$.
Thus the remaining set $I\setminus I^\star$ is a subset of $V_{I^\star}$,
	the set of vertices with distance $\geq \alpha$ to every vertex in $I^\star$.

For every color class $i\in[k]$, let $U^i$
	be the set of vertices with distance at most $3n$ to vertex $u^i_{n/2}$ (the vertex halfway between $a^i$ and 
	$b^i$).
For every color classes $i,i'\in[k]$ with $i<i'$,
	let $U^{i,i'}$
	be the set vertices with distance at most $14n$ to
	at least one of the vertices $w^{i,i'}_{j,j'}$ where $\{u^i_j,u^{i'}_{j'}\} \in E(G)$.
Let $\UU$ be the family of sets $U^i$ and $U^{i,i'}$ for every $i,i'\in[k]$ with $i < i'$.
Notably, any two vertices in a single set $U \in \UU$ have distance at most $2(14n+12n)=52n < \alpha$.
Hence $I$ contains at most one vertex from $U$, for every set $U \in \UU$.
We observe that the union of the sets $U \in \UU$
	is a superset of $V_{I^\star}$.
Thus, as $|I|-|I^\star|=k+\binom{k}{2}=|\UU|$,
	we have that $I$ contains exactly one vertex from each of the sets $U \in \UU$.

Consider color classes $i,i'\in[k]$ with $i<i'$.
Let $u^{i,i'}$ be the unique vertex in $I \cap U^{i,i'}$.
We recall that
	recall that for $j,j'\in[n]$ where $\{u^i_j,u^{i'}_{j'}\}\in E(G)$, 
	the vertex $w^{i,i'}_{j,j'}$ is connected to the remaining graph by a path of length $18n$.
Hence, by arguments analogous to \cref{lemma:ind:leaf},
	we may assume that $u^{i,i'} = w^{i,i'}_{j,j'}$
	for some indices $j,j'\in[n]$
	where $\{u^i_j,u^{i'}_{j'}\}\in E(G)$.

For color classes $i,i' \in [k]$ with $i < i'$, let $V^{i,i'}\subseteq U^i$
	be those vertices of $U^i$ which have distance at least $64n-2$ to $u^{i,i'}$.
Let $j,j'$ be such that $u^{i,i'} = w^{i,i'}_{j,j'}$.
Then $V^{i,i'}$ is the set of vertices with distance $\leq 2$ to vertex $v^i_{j}$.
This is due to that the closed walk $W$ from vertex $w^{i,i'}_{j,j'}$ to $\hat w^{i,i'}_{j,j'}$ to $a^{i}$ to $b^{i}$
	to $\hat w^{i,i'}_{j,j'}$ back to $w^{i,i'}_{j,j'}$ has length $18n+30n-6j + 48n + 18n + 6j + 18n = 2\cdot 64 n$
	and that $u^i_j$ lies exactly halfway on the walk $W$.
We observe that $V^{i,i'} \cap V^{i,i''}\neq\emptyset$
	for pairwise distinct color classes $i,i',i'' \in [k]$
	if and only if
	$u^{i,i''}=w^{i,i''}_{j,j''}$ for some $j''$
	(that is the first lower index $j$ coincides for pairs $i,i'$ and $i,i''$).
Thus, for every color class $i\in[k]$, the vertices $u^{i,i'}$ where $i'\in [k]\setminus\{i\}$
	have a common integer index $j(i)$,
	such that the unique vertex in $U^i$ has distance $\leq 2$ to $u^i_j$.
Hence, for every color classes $i,i'\in[k]$ with $i<i'$,
	there is an edge $\{u^i_{j(i)}, u^{i'}_{j(i')}\}\in E(G)$
	as $I$ contains $w^{i,i'}_{j(i),j(i')}$.
In other words, $\bigcup_{i \in [k]} u^i_{j(i)}$ is a $k$-clique in $G$.

\smallskip

We claim that the the constructed graph has no cycle of length $<64n$
	(and hence no cycle of length $<64n-2$).
Every cycle $C$
	contains $\{a^i,b^i\}$ for some index $i$.
Further, it contains the $a^i,\hat w^{i,i'}_{j,j'}$-path and the $b^i,w^{i,i'}_{j,j'}$-path,
	or it contains the $a^i,\hat w^{i,i'}_{j,j'}$-path and the $b^i,\hat w^{i,i''}_{j'',j'''}$-path
	for some indices $i,i',i'',j',j'',j'''$ with $(i',j,j')\neq(i'',j'',j''')$.
In the former case, $C$ has length at least $46n+30n+24n \geq 64n$.
In the latter case, $C$ has length at least $4\cdot 16n \geq 64n$.

It remains to show that the constructed graph $G'_2$ has pathwidth $\Oh(k)$.
Let $I' = \bigcup_{i\in[k]} \{a^i,b^i\}$.
We have that $G - I'$ is a forest consisting of paths
	and trees $T_{i,i'}$ for color classes $i,i' \in [k]$ with $i<i'$,
	where $T_{i,i'}$ contains vertex $w^{i,i'}_\star$.
For every tree $T_{i,i'}$, removing $\hat w^{i,i'}$ yields a forest with at most one non-leave,
	hence of constant pathwidth.
Thus the pathwidth of a tree $T_{i,i'}$ is also constant.
Since $|I'|=2k$, we obtain that $G'_2$ has pathwidth $\Oh(k)$.
\end{proof}

\begin{lemma}
(\cref{lemma:disp:gamma:delta} restated)
\label{a:lemma:disp:gamma:delta}
For $i\geq2$, and $\gamma_i$, $\delta$, $\delta_i$ as defined in \Cref{i:def:delta}, we have
$
\gamma_{i-1} < \gamma_i < \delta < \delta_i < \delta_{i-1}
$
\end{lemma}
\begin{proof}
First, we note that $\gamma_i < \delta < \delta_i$ follows from the inequalities $\gamma_{i-1} < \gamma_i < \delta_i 
< \delta_{i-1}$.
Indeed, since $\delta_i$ is non-increasing and $\gamma_i$ is non-decreasing,
	the limes (inferior) of $\delta_i$ satisfies $\gamma_i \leq \delta \leq \delta_i$.

To see that $\delta_i$ is non-increasing,
	that is $\delta_i \leq \delta_{i-1}$ for $i\geq 2$,
	we observe that
	$(\delta_{i-1}-\delta_i)b_i b_{i-1} = {a_{i-1}} b_i - {a_{i}} b_{i-1}
	=  {a_{i-1}} (b_{i-1}c_i+1) - c_i{a_{i-1}} b_{i-1} = 1$.

Regarding the inequality $\gamma_i\leq\delta_i$,
	we observe that $(\delta_i - \gamma_i)b_i (b_i - b_{i-1})
	= c_i a_{i-1} (c_{i} b_{i-1} +1 - b_{i-1}) - (c_i a_{i-1} - a_{i-1}) (c_i b_{i-1} +1)
	= a_{i-1}( {c_i} (c_i b_{i-1}+1-b_{i-1}) - (c_{i} + c_i b_{i-1} + 1) )
	= a_{i-1}( {c_i}^2 b_{i-1} + c_i - c_i b_{i-1} - ( {c_i}^2 b_{i-1} + c_{i} - c_{i-1} b_{i-1} - 1) ) = a_{i-1} > 0
	$.

It remains to show that $\gamma_i$ is non-decreasing,
	that is $\gamma_{i} - \gamma_{i-1}$ for $i\geq 2$.
For ease of notation let $c=c_{i-1}$, $d=c_i$, $a=a_{i-2}$ and $b=b_{i-2}$.
We note that $a_{i-1} - a_{i-2} = ca-a$ and $a_i - a_{i-1} = cda -ca$.
Further, $b_{i-1} - b_{i-2} = cb+1$ and $b_i - b_{i-1} = d(cb+1)+1-(cb+1) = cdb + d - cb$.
Then
\begin{align*}
	&(\gamma_i - \gamma_{i-1})(b_{i}-b_{i-1})(b_i - b_{i-1})\\
	=\; &(cda-ca)(cb+1-b) - ((ca-a)(cdb+d-cb))\\
	=\; &(c^2dab + cdb -cdab -c^2ab -ca + cab) - (c^2dab+da -c^2ab- cdab -da +cab)\\
	=\; &a(d-c),
\end{align*}
which is positive since $a$ is positive and the sequence $(c_i)_{i\geq 0}$ is increasing.
\end{proof}

\subsection{Translation Preserving Simplicity}
\label{a:section:disp:translation}

In this section, we prove \cref{a:lemma:disp:translate:b:simple}
	which relates $a$-independent sets in the subdivided graphs $G_b$ and $G_{a+b}$.
The correctness follows with almost the same proof as for \cref{lemma:disp:translate}
	as shown in~\cite{HartmannL22,thesis}.
For a better comparison, we follow the proof of~\cite{HartmannL22,thesis} quite literally
	and highlight changes in bold face.

\begin{lemma}[\cref{lemma:disp:translate:b:simple} restated]
\label{a:lemma:disp:translate:b:simple}
\lemmaTextDispTranslateBSimple
\end{lemma}
Before we prove \cref{lemma:disp:translate:b:simple} let us recall some definitions from~\cite{HartmannL22}.
Let $\delta \coloneqq \tfrac{a}{b}$.
For an edge $\{u,v\} \in E(G)$, let
\[
\dr_S(u,v) \!\coloneqq\! \begin{cases}
	\min\{\lambda \mid p(u,v,\lambda) \in S\}, &
	\!\! \text{if } S \cap P(G[\{u,v\}]) \neq \emptyset, \\
	1 + \min\!\big\{ \dr_S(v,w) \;\big|\; w \in N(v) \setminus \{u\} \big\}, & \!\! \text{else}.
\end{cases}	
\]
Further, recall that
	a set of points $S \subseteq P(G)$ is \define{edge internally $\delta$-dispersed}
	if every edge by itself does not falsify that $S$ is $\delta$-dispersed,
	formally that every distinct points $p,q \in P(G[\{u,v\}])$ for every edges $\{u,v\}\in E(G)$
	are $\delta$-dispersed.

Then for $G$ (since without any cycle of length $<\tfrac{a}{b}$),
	a subset $S \subseteq P(G)$ is $\delta$-dispersed, if and only if
	\begin{itemize}
		\item \labeltext{$($A1$)$}{condition:a1}
		$S$ is edge internally $\delta$-dispersed, and
		\item \labeltext{$($A2$)$}{condition:a2}
		$\dr_{S}(u,v) + \dr_{S}(u,w) \geq \delta$ for every distinct adjacent edges $\{u,v\},\allowbreak \{u,w\} 
		\in E(G)$ where vertex $u$ is not in $S$.
	\end{itemize}

\begin{proof}[Proof of \cref{a:lemma:disp:translate:b:simple}]
We observe that the construction from \cite{HartmannL22,thesis}
	already yields the latter claim,
	that for an $(a+b)$-simple $\frac{a}{a+b}$-dispersed set $S'$,
	there is a $b$-simple $\frac{a}{b}$-dispersed set of size $|S'|-|E(G)|$.
It remains to show the first claim,
	that for an $b$-simple $\frac{a}{b}$-dispersed set $S$,
	there is a $(a+b)$-simple $\frac{a}{a+b}$-dispersed set of size $|S|+|E(G)|$.
We adapt the construction of~\cite{HartmannL22,thesis} slightly.
For a better comparison, we follow the proof of~\cite{HartmannL22,thesis} quite literally
	and highlight changes in bold face.
Let $\delta = \tfrac{a}{b}$.

\medskip

Consider a $\delta$-dispersed set $S \subseteq P(G)$ of size $|S| = \disp{\delta}(G)$.
Let $\rho \coloneqq (\delta+1)^{-1}$, the conversion ratio.
We construct a $(\delta\rho)$-dispersed set $S'$ of size $|S'| = |S|+|E(G)|$. %
A key observation is that
$$\rho + \delta\rho \;=\; 1 .$$

\textbf{For a rational $\mathbold{\delta = \tfrac{a}{b}}$ we have $\mathbold{\rho = \tfrac{b}{a+b}}$. %
That means, for example, if an edge position $\mathbold \lambda$ is $\mathbold b$-simple (i.e., multiple of 
$\mathbold{b^{-1}}$)
	then $\mathbold{\lambda\rho}$ is $\mathbold{(a+b)}$-simple.}

We construct $S'$ by considering each edge $\{u,v\} \in E(G)$ separately.
Set~$S'$ compared to~$S$
will contain one more point in $P(G[\{u,v\}])$ for every edge $\{u,v\} \in E(G)$.
Any point~$p \in S$ on a vertex will also be in~$S'$.
Thus $|S'| = |S| + |E(G)|$.
\begin{itemize}
	\item
	Consider that edge $\{u,v\} \in E(G)$ contains $m \geq 1$ points from $S$,
	which is $|S \cap P(G[\{u,v\}])| = m$.
	Let $p = p(u,v,\lambda)$ be the point among $S \cap P(G[\{u,v\}])$ with minimum distance to~$u$, hence with 
	distance $\lambda$.
	Analogously, let $q= p(u,v, 1- \mu)$ be the point among $S \cap P(G[\{u,v\}])$ with minimal distance to $v$,
	hence with distance $\mu$,
	possibly $p=q$.
	Then $\dr_S(u,v)=\lambda$ and $\dr_S(v,u)=\mu$.
	There are $m-1 \geq 0$ further points between $p$ and $q$ on edge $\{u,v\}$,
	such that $p$ and $q$ have distance $1-(\lambda + \mu) \geq m\delta$.
	In other words $\lambda + \mu \leq 1 - m\delta$.
	
	We add points $p' = p(u,v,\lambda \rho)$ and $q' = p(u,v, 1 - \mu \rho)$ to $S'$,
	which are distinct even when $p=q$.
	Note that if $p$ is positioned on a vertex, also $p'$ is; analogously for $q$ and~$q'$.
	Thus the distance between the new points $p'$ and $q'$ is
	$$1 - (\lambda + \mu)\rho \; \geq \; 1 - (1-m\delta) \rho \; = \; 1 - \rho + m \delta\rho \; = \; 
	(m+1)\delta\rho . $$
	Hence we may add $m$ further points between $p'$ and $q'$ on edge $\{u,v\}$ to set $S'$
	in such a way that the now $m+2$ points in $S' \cap P(G[\{u,v\}])$ have pairwise distance at least~$\delta\rho$.
	Then $S'$ is edge internally $\delta$-dispersed for edge $\{u,v\}$.
\textbf{Given that $\mathbold p$ and $\mathbold q$ are $\mathbold b$-simple, the new points are 
$\mathbold{(a+b)}$-simple.}
	\item
	It remains to consider that edge $\{u,v\} \in E(G)$ contains no point from the set~$S$,
	which is $S \cap P(G[\{u,v\}]) = \emptyset$.
	We fix a neighbor $w_u \in N(u)$ with minimum distance $\dr_S(u,w_u)$.
	Analogously, we fix a neighbor $w_v \in N(v)$ with minimum distance $\dr_S(v,w_v)$.
	By symmetry, we may assume the inequality $\dr_S(u,w_u) \leq \dr_S(v,w_v)$.
	We have that $w_u \neq v$ since otherwise $\dr_S(u,w_u) = 1 + \dr_S(v,w_v)$.
	Then we add a new point $p(u,v, \lambda' )$ at edge position $\lambda' = \max\{ \mathbold{ 
	(\frac{\delta}{2})^\star} , \delta\rho - \dr_S(u,w_u)\rho\}$ to~$S'$,
\textbf{where $\mathbold{ (\frac{\delta\rho}{2})^\star}$
	is the the smallest value $\mathbold{ \geq \frac{\delta\rho}{2} }$
	which is $\mathbold{ (a+b) }$-simple,
	hence either $\mathbold{\frac{\delta\rho}{2}}$ or $\mathbold{\frac{\delta\rho}{2}+\tfrac{\rho}{2b} = 
	\frac{\delta\rho}{2} + \frac{\rho}{2b} = \frac{a+1}{2(a+b)}}$.
This is the only time where we have to adjust the construction.}
\end{itemize}

\sloppy
We show that $S'$ is $(\delta\rho)$-dispersed
by proving the conditions from above.
Set $S'$ is edge internally $(\delta\rho)$-dispersed
since for every edge we have added points that have pairwise distance at least $\delta$.
It remains to show condition~\ref{condition:a2},
which is that $\dr_{S'}(u,v) + \dr_{S'}(u,w) \geq \delta$ for every distinct adjacent edges $\{u,v\},\{u,w\} \in 
E(G)$ where vertex $u$ is not in $S$.

We partition the directed edges $(u,v)$ for $\{u,v\} \in E(G)$ as follows:
\begin{itemize}
	\item %
	A directed edge $(u,v)$ is \define{positive} if it satisfies $\dr_{S'}(u,v) \geq \dr_S(u,v)\rho$.
	This is the case when $S \cap G[\{u,v\}]\neq \emptyset$, by construction.
	Further, we claim that $(u,v)$ is also positive if $S \cap G[\{u,v\}] = \emptyset$
	and point $p(u,v,\lambda')$ with $\lambda' = 1 - (\delta\rho - \dr_S(v,w_v)\rho)$ is added to~$S'$,
	where $w_v \in N_G(v)$ is a neighbor with minimum $\dr_S(v,w_v)$,
	which has $w_v \neq u$.
	Then $\dr_S(u,v) = 1 + \dr_S(v,w_v)$.
	It follows that
	$$
	\dr_{S'}(u,v) \;=\; 1 - \delta\rho + \dr_S(v,w_v)\rho
	\;=\; 1 - \delta\rho + \dr_S(u,v)\rho - \rho
	\;=\; \dr_S(u,v)\rho .
	$$
	\item
	A directed edge $(u,v)$ is \define{neutral} if it is not positive, $S \cap G[\{u,v\}] = \emptyset$ and point 
	$p(u,v,\lambda')$ with $\lambda' = 1-  \mathbold{ (\frac{\delta\rho}{2})^\star }$ is added to $S'$.
	\item
	A directed edge $(u,v)$ is \define{negative} if $S \cap G[\{u,v\}] = \emptyset$ and point $p(u,v, \lambda')$
	with edge position $\lambda' = \max\{ \mathbold{ (\frac{\delta\rho}{2})^\star}, \delta\rho - \dr_S(u,w_u)\rho 
	\}$ is added to $S'$,
	in which case $w_v \neq u$.
\end{itemize}
Now, we observe that distinct adjacent edges $\{u,v\},\{u,w\} \in E(G)$ satisfy condition~\ref{condition:a2},
that is $\dr_{S'}(u,v) + \dr_{S'}(u,w) \geq 2(\frac{\delta\rho}{2}) = \delta\rho$.
We distinguish which type the directed edges $(u,v)$ and $(u,w)$ have.
\begin{itemize}
	\item
	If the directed edges $(u,v),(u,w)$ are positive, then we have
	$\dr_{S'}(u,v) + \dr_{S'}(u,w) \geq (\dr_{S}(u,v) + \dr_{S}(u,w))\rho \geq \delta \rho = \delta\rho$.
	Here we used that conditions \ref{condition:a1} and \ref{condition:a2} apply to $S$ and $\delta$,
	hence that $\dr_{S}(u,v) + \dr_{S}(u,w) \geq \delta$.
	\item
	If the directed edges $(u,v),(u,w)$ are both not positive,
	then $\dr_{S'}(u,v) + \dr_{S'}(u,w) \geq 2\frac{\delta\rho}{2} = \delta\rho$, since $\delta\rho < 1$.
\textbf{Especially for a negative edge $\mathbold{(u,v)}$,
	we have $\mathbold{ \dr_{S'}(u,v) \geq 1-\frac{a+1}{2(a+b)} \geq \frac{\delta\rho}{2} }$ since $\mathbold{b\geq 
	1}$.}

	\item
	Consider that $(u,v)$ is neutral and $(u,w)$ is positive.
	By construction, we have $\dr_S(u,w) \geq \dr_S(u,w_u) \geq \dr_S(v,w_v)$.
Further, since $(u,v)$ is not positive,
	a point is added at $p(v,u,\lambda')$
	where $\lambda' = \delta\rho - \dr_S(v,w_v)\rho < \mathbold{ \frac{\delta\rho}{2}+\frac{\rho}{2b} } $.
	Then
	\begin{align*}
		\dr_{S'}(u,w)+\dr_{S'}(u,v)
		\;\geq&\; \dr_{S}(u,w)\rho + 1 - \mathbold{ ( \tfrac{\delta\rho}{2} + \tfrac{\rho}{2b} ) } \\
		\;\geq&\; \dr_{S}(v,w_v)\rho + 1 - \mathbold{ \tfrac{\delta\rho}{2} - \tfrac{\rho}{2b} } \\
		\;>&\; 1  \mathbold{- \tfrac{\rho}{b} } \\
		\;=&\; \mathbold{ \tfrac{a+b-1}{a+b} } \\
		\;\geq&\; \delta\rho.
	\end{align*}
	\item
	It remains to consider that $(u,v)$ is negative and $(u,w)$ is positive.
	By definition $\dr_{S'}(u,w) + \dr_{S'}(u,v) \geq \dr_{S'}(u,w_u) + \dr_{S'}(u,v)$.
	If $(u,w_u)$ is neutral or negative, we obtain that $ \dr_{S'}(u,w_u) + \dr_{S'}(u,v) \geq \delta\rho$, by the 
	previous cases, as desired.
	Thus consider that $(u,w_u)$ is positive.
	Then
	\begin{align*}
		\dr_{S'}(u,w_u) + \dr_{S'}(u,v)
		\;\geq\;& \dr_{S'}(u,w_u) + \delta\rho - \dr_{S}(u,w_u)\rho \\
		\;\geq\;& \dr_{S'}(u,w_u) + \delta\rho - \dr_{S'}(u,w_u)
		\;=\; \delta\rho. %
	\end{align*}
\end{itemize}
Therefore \ref{condition:a1} and \ref{condition:a2} apply to $S'$ and $\delta\rho$,
which shows that $S'$ is $\delta\rho$-dispersed.
\end{proof}

\section{Details for Domination and Covering}
\label{section:domination:and:covering}

\subsection{Introduction}

There are the following equivalent definitions \ref{def:ds:1}, \ref{def:ds:2} and \ref{def:ds:3}:
\definitionsWalkDomination

\begin{lemma}[\cref{i:lemma:ds:equivalent} restated]
\label{lemma:ds:equivalent}
\lemmaTextDSequivalent
\end{lemma}
\begin{proof}
We show that an \dombarset{a} $D$ according to definition \ref{def:ds:1}
	is also an \dombarset{a} according to definition \ref{def:ds:2},
	analogously for \ref{def:ds:2} to \ref{def:ds:3}, and for \ref{def:ds:3} to \ref{def:ds:1}.

\ref{def:ds:1} $\to$ \ref{def:ds:2}:
Consider a vertex $u \in V'$ where $d_G(u,D) \geq \ahalf$.
Since graph $G$ contains no isolated vertices,
	there is a neighbor $v \in N(u) \cap V'$.
According to \ref{def:ds:1} there are vertices $w_1,w_2 \in D$
	and a $w_1,w_2$-walk $P$ of length at most $a$ that contains $\{u,v\}$.
Hence $d_G(u,D)=d_g(u,w_1) = \ahalf$.
Now property \ref{def:ds:2} follows assuming there are no neighboring vertices $u,v$
	with $d_G(u,D) = d_G(v,D) = \ahalf$.
For neighboring vertices $u,v$, again there are vertices $w_1,w_2 \in D$
	that have a $w_1,w_2$-walk $P$ of length at most $a$ that contains the edge $\{u,v\}$.
Then not both of $u,v$ can have length $a$ to $\{w_1,w_2\}\subseteq D$.

\ref{def:ds:2} $\to$ \ref{def:ds:3}:
Consider a vertex $u \in V'$.
Then by property \ref{def:ds:2}
	indeed $d_G(u,D) \leq \ahalf$ and hence $d_{G_2}(u,D) < a$.
Consider an edge $\{u,v\} \in E'$.
Then by property \ref{def:ds:2}
	indeed $d_G(w,D) \leq \ahalf -\half$ for at least one incident vertex $w \in \{u,v\}$.
Thus $d_{G_2}(\{u,v\},D) \leq 2 \cdot d_G(w,D) + 1 \leq a$.
Therefore $D$ $a$-dominates $V' \cup E'$ in $G_2$.

\ref{def:ds:3} $\to$ \ref{def:ds:1}:
Consider an edge $\{u,v\} \in E'$, which hence occurs as a vertex in $G_2$.
By property \ref{def:ds:3}, we have $d_{G_2}(u,D), d_{G_2}(\{u,v\},D), d_{G_2}(v,D) \leq a$.
We note that $G_2$ is bipartite with $D$ as a subset of $V(G)$ from one partition.
Hence for every two neighboring vertices in $G_2$,
	their distance to $D$ differs by exactly $1$.
We further observe that $d_{G_2}(\{u,v\},D) = 1+ \min\{ d_{G_2}(u,D), d_{G_2}(v,D)\}$.
Thus, up to symmetry, the following two cases remain:
The first case is that $d_{G_2}(u,D) = a' \leq a$ witnessed by some vertex $w_u \in D$
	and $d_{G_2}(\{u,v\},D) = a'-1$ and $d_{G_2}(v,D) = a'-2$ witnessed by some vertex $w_v \in D$.
The second case is that $d_{G_2}(\{u,v\},D) = a' \leq a$
	and $d_{G_2}(u,D) = a'-1$ witnessed by some vertex $w_u \in D$
	and $d_{G_2}(v,D) = a'-1$ witnessed by some vertex $w_v \in D$.
In both cases, the concatenation of a shortest $w_u,u$-path, path $u,\{u,v\},v$ and a shortest $v,w_v$-path
	forms a walk of length $2a'\leq 2a$ in $G_2$.
The corresponding path in $G$ then is a $w_u,w_v$-path of length $a$ that contains the edge $\{u,v\}$,
	hence satisfies property \ref{def:ds:1}.
\end{proof}

Let $\ogamma_a(G)$ be the minimum cardinality of a \dombarset{a} of $G$,
	and for a graph class $\GG$,
	let $\DS_a(\GG)$ be the decision problem for given a graph $G \in \GG$ and integer $k$,
	$\ogamma_a(G) \leq k$.

Let $\ogamma_a(G)$ be the minimum cardinality of an $a$-dominating set of $G$,
	and for a graph class~$\GG$,
	let $a$-$\DS(\GG)$ be the problem, given a graph $G \in \GG$ and integer~$k$,
	deciding whether $\ogamma_a(G) \geq k$.

\subsection{Dominating Set and Covering}
\label{section:ds:and:covering}
\ineditnote{\cref{section:ind:and:disp}}

This section explores the close relationship between
$a$-walk dominating set and $\delta$-dispersion.
Particularly, we relate \dombarsets{a} similarly as the two results below do for covering sets.

\begin{lemma}[\cite{HartmannLW22}]
\label{lemma:cover:subdivide}
For every real $\delta > 0$ and integer $c\geq 1$,
	$\cover{\delta}(G)=\cover{c\delta}(G_c)$.
\end{lemma}

\begin{lemma}[\cite{HartmannLW22}]
\label{lemma:cover:translate}
$\cover{\tfrac{a}{b}}(G) + |E(G)| = \cover{\tfrac{a}{2a+b}}(G)$.
\end{lemma}

To do so, recall that we denote a point $p(u,v,\lambda) \in P(G)$ as $c$-simple
	if $\lambda$ is a multiple of $c$.
Further, recall that $S \subseteq P(G)$ be $c$-simple if all its points are $c$-simple.
Restricted to $b$-simple sets, there is a direct correspondence
	of \dombarsets{a} in the $b$-subdivided of a graph $G$ and 
	$\tfrac{a}{b}$-dispersed sets in $G$ itself.

\begin{observation}[\cref{i:lemma:dom:iff:cov} restated]
\label{lemma:dom:iff:cov}
Let $k \in \NNN$.
For a graph $G$ without isolated vertices, there is an \dombarset{a} of $G_{b}$ of size $k$,
	if and only if
	there is a $b$-simple $\tfrac{a}{2b}$-covering set of $G$ of size~$k$.
\end{observation}
\begin{proof}
Consider an \dombarset{a} $D$ of $G_b$.
Consider the points on an edge $\{u,v\}\in E(G_b)$.
By property \ref{def:ds:2} of $D$,
	we have $d(u,D), d(v,D) < \ahalf$, or $d(u,D)= \ahalf$ and $d(v,D)=\ahalf-1$ up to symmetry.
In both cases, set $D$ $\ahalf$-covers every point of edge $\{u,v\}$ in $G_b$.
Hence $S$, the set of ($b$-simple) points in $P(G)$ corresponding to the vertices in $D$,
	$\frac{a}{2}$-covers $G$.

Vice versa, consider a $b$-simple $\frac{a}{2b}$-covering set $S$ of $G$.
Then $D$, the set of vertices in $G_b$ corresponding to the points $S$,
	forms a $\ahalf$-cover of $G_b$.
Consider an edge $\{u,v\}\in E(G_b)$. %
Then $d(u,D),d(v,D) \leq \ahalf-\half$, or $d(u,D)=\ahalf$ and $d(v,D)=\ahalf-1$ up to symmetry.
Let $w_1,w_2 \in D$ be the closest vertices to $u$ and $v$, respectively.
Then in both cases, a shortest $w_1,w_2$-path via $u$, edge $\{u,v\}$ and $v$
	forms a path of length at most $a$.
Hence $D$ is also an \dombarset{a} of $G_b$.
\end{proof}

A minimum $\tfrac{a}{b}$-covering set $S$ may not be $b$-simple,
	but we can assume it to be $2b$-simple.

\begin{lemma}[\cite{HartmannLW22}]
For an $\tfrac{a}{b}$-covering set $S$ of a graph $G$
	there is $2b$-simple $\tfrac{a}{b}$-covering set $S^\star$ of $G$ of size $|S^\star|=|S|$.
\end{lemma}

This observation is almost analogue to \cref{lemma:covering:simplify} for dispersion.
However, it does not imply a $b$-simple set $S^\star$
	given a set $S$ that contains no point at a position $\tfrac{2i-1}{2b}$ for $i \in \NNN$.
We prove the following, generalized statement that implies such a $b$-simple set $S^\star$
	given a set $S$ that contains no point at a position $\tfrac{2i-1}{2b}$ for $i \in \NNN$.
Particularly, it provides a stronger result for $\tfrac{a}{b}$-covers with even $b$.
For example, a $\half$-cover $S$ implies a $2$-simple $\half$-covering set $S^\star$ of same size

\begin{lemma}[\cref{i:lemma:ocover:simple} restated]
\label{lemma:ocover:simple}
Let $S$ be an $\tfrac{a}{b}$-cover of a graph $G$ for integers $a,b \in \NNN$.
Then there is an $\tfrac{a}{b}$-cover $S^\star$ of $G$ of size $|S^\star|=|S|$
	that is $2b$-simple and, if $b$ is a multiple of $2$, is $b$-simple.
	
Assuming that $S$ contains no point at a position $\tfrac{2i-1}{2b}$ for $i \in \NNN$,
	then $S^\star$ is $b$-simple,
	and, if additionally $b$ is a multiple of $2$,
	$S^\star$ is $\frac{b}{2}$-simple.
\end{lemma}
\begin{proof}
We consider $S$ as an $\tfrac{a'}{b'}$-cover of $G$
	with $b'$ integral and $a'$ minimal half-integral such that 
 	$\tfrac{a}{b} = \tfrac{a'}{b'}$.
In other words, when $b$ is a multiple of $2$, we treat $S$ as an $\frac{a/2}{b/2}$-cover.
Then it suffices to show that there is a $2b'$-simple $\tfrac{a'}{b'}$-cover $S^\star$,
	and that if $S$ contains no point at a position $\tfrac{2i-1}{2b}$ for $i \in \NNN$,
	then $S^\star$ is $b'$-simple.

Further, it suffices to consider the case $b'=1$.
For $b>1$, an ${a'}$-cover of $G_{b'}$,
	corresponds to an $\tfrac{a'}{b'}$-cover of $G$.
Constructing a $2$-simple (or $1$-simple) ${a'}$-cover of $G_{b'}$,
	implies a $2b$-simple (respectively $b$-simple) $\tfrac{a'}{b'}$-cover of $G$,
	by \cref{lemma:cover:subdivide}.

Hence consider an $a'$-cover $S$ of a graph $G$ with half-integral $a'$ (hence that $2a'$ is integer).
Consider point $p = p(u,v,\lambda) \in S$ with, up to symmetry, $\lambda \leq \half$.
If $\lambda < \half$,
	let $p^\star \coloneqq p(u,v,0)$.
Else, in case $\lambda = \half$,
	let $p^\star \coloneqq p$.
Then the resulting set $S^\star = \{ p^\star \mid p \in S\}$ is $2$-simple.
Especially, if $S$ contains no point at edge position $\half$, $S^\star$ is $1$-simple.

It remains to show that $S^\star$ is ${a'}$-covering.
For every edge $\{x,y\} \in E(G)$, there are points $p_x,p_y \in S$ and a real $\mu \in [0,1]$
	such that $p_x$ $a'$-covers the points $p(x,y,\lambda)$ for $\lambda \in [0,\mu]$,
	and $p_y$ $a'$-covers the points $p(x,y,\lambda)$ for $\lambda \in [\mu,1]$.
By symmetry, we may assume that $\mu \le \half$.
First consider that $\mu$ can be chosen to be half-integral,
	hence there is a $\mu \in \{0,\half\}$
	such that $p_x$ $a'$-covers the points $p(x,y,\lambda)$ for $\lambda \in [0,\mu]$,
		and $p_y$ $a'$-covers the points $p(x,y,\lambda)$ for $\lambda \in [\mu,1]$.
Then since $p_x^\star,p_y^\star$ respectively are not moved across a half-integral point
	and ${a'}$ is half-integral,
	still $p_x^\star$ covers the points $p(x,y,\lambda)$ with $\lambda \in [0,\mu]$
		and $p_y^\star$ covers the points $p(x,y,\lambda)$ with $\lambda \in [\mu,1]$.
It remains to consider that $\mu \in (0,\half)$.
Let $p_x = p(u_x,v_x,\lambda_x)$ where $v_x$ is on a shortest path from $p_x$ to $u$,
	and let $p_y = p(u_y,v_y,\lambda_y)$ where $v_y$ is on a shortest path from $p_y$ to $v$.
Because the covering distance ${a'}$ is half-integral (and $\mu$ is not in $\{0,\half\}$),
	we have $\lambda_x > \half$ or $\lambda_y < \half$.
In case $\lambda_x > \half$, it follows that $p_x^\star = p(u_x,v_x,1)$ and ${a'}$ is not integer.
And in case $\lambda_y > \half$, we have $p_y^\star = p(u_y,v_y,0)$ and ${a'}$ is integer.
In both cases, $p_x^\star$ and $p_y^\star$ cover all points of edge $\{u,v\}$.
\end{proof}

\begin{corollary}[\cref{i:lemma:ds:to:cov} restated]
\label{lemma:ds:to:cov}
$\cover{\tfrac{a}{b}}(G) = \ogamma_{4a}(G_{2b}) = \ogamma_{4ca}(G_{2cb})$;
$\cover{\tfrac{a}{2b}}(G) = \ogamma_{2a}(G_{2b}) = \ogamma_{2ca}(G_{2cb})$,
	for any $a,b,c \in \NNN$ and graph $G$.
\end{corollary}

The analogous result to \cref{lemma:alpha:sub} follows for \dombarset{a}.

\begin{theorem}[\cref{lemma:ds:sub} restated]
\label{a:lemma:ds:sub}
\lemmaTextDsSub
\end{theorem}
\begin{proof}
First consider that $a,b$ are even.
Then $\ogamma_{a}(G_b) = \ogamma_{2a'}(G_{2b'})$ for some integers $a',b'$.
Then $\ogamma_{2a'}(G_{2b'}) = \ogamma_{2ca'}(G_{2cb'}) = \ogamma_{ca}(G_{cb})$,
	because of \cref{lemma:ds:to:cov}.

Now consider that $c$ is odd.
Clearly, any \dombarset{a} in $G_b$ corresponds to a \dombarset{ca} in $G_{cb}$.
For the reverse direction,
	let $D \subseteq V(G_{cb})$ be an \dombarset{ca} set of $G_{cb}$.
Consider the corresponding $a$-covering set $S \subseteq V(G_b)$ in $G_b$,
	which is $bc$-simple.
Apply the construction of \cref{lemma:ocover:simple}.
Because $c$ is odd,
	there is no point $p \in S$ with edge position $\tfrac{1}{2}$.
Hence the construction only produces points with edge position $0$ or $1$.
In other words, the constructed set $S^\star$ is $1$-simple,
	and hence $S^\star$ corresponds to a \dombarset{a} in~$G_b$.
\end{proof}

The proof of \cref{lemma:cover:translate} is constructive
	and preserves simplicity of the point sets.
As stated in~\cite{HartmannLW22},
	the proof of \cref{lemma:cover:translate} considers \emph{maximum} $\frac{a}{b}$-covering sets
	and \emph{maximum} $\frac{a}{2a+b}$-covering sets
	and translates them to one another with the size difference of $|E(G)|$.
However, in fact the construction in the forward direction does not require
	the $\frac{a}{b}$-covering set to be maximum,
	and the construction in the backward direction only requires
	that the $\frac{a}{2a+b}$-covering set $S'$
	contains for every edge $\{u,v\}\in E(G)$
	at least one point $p(u,v,\lambda)$ with $\lambda\in(0,1)$.
Since $\frac{a}{2a+b}<\half$, this has to be the case in order to cover the point $p(u,v,\half)$.
Thus we obtain the following result.

\begin{lemma}[\cite{HartmannLW22}]
\label{lemma:cover:translate:simplicity}
\ineditnote{\cref{lemma:disp:translate:b:simple}}
Let $G$ be a graph.
A $b$-simple $\tfrac{a}{b}$-covering set $S$ of $G$
	implies a $(2a+b)$-simple $\tfrac{a}{2a+b}$-covering set of $G$ of size $|S|+|E(G)|$.
Further, a $(2a+b)$-simple $\tfrac{a}{2a+b}$-covering set $S'$
	implies a $b$-simple $\tfrac{a}{b}$-covering set of size $|S'|-|E(G)|$.
\end{lemma}

Hence the following  relation of \dombarset{a} in $G_b$ and $G_{a+b}$ holds,
	analogously to \cref{lemma:alpha:translate}.

\begin{theorem}[\cref{i:lemma:trans:ogamma} restated]
\label{lemma:trans:ogamma}
\lemmaTextDsTranslation
\end{theorem}
\begin{proof}
Let $D$ be a \dombarset{a} of $G_{b}$.
Then $D$ correspond to a $b$-simple $\frac{a}{2b}$-covering set $S$ of $G$ of size $|D|$,
	by \cref{lemma:dom:iff:cov}.
\cref{lemma:cover:translate:simplicity} maps $D$ to an $(a+b)$-simple $\tfrac{a}{2a+2b}$-covering set $S'$ of $G$ of size $|D|+|E(G)|$.
Then $S'$ corresponds to a \dombarset{a} of $G_{a+b}$ of size $|D|+|E(G)|$.
That means $\ogamma_a(G_b) + |E(G)| \leq \ogamma_{a}(G_{a+b})$.

Vice versa, let $D'$ be an \dombarset{a} of $G_{a+b}$.
Then $D'$ corresponds to an $(a+b)$-simple $\tfrac{a}{2a+2b}$-covering set $S'$ of $G$ of size $|D'|$,
	by \cref{lemma:dom:iff:cov}.
\cref{lemma:cover:translate:simplicity} maps $S'$ to an $a$-simple $\tfrac{a}{2b}$-covering set $S$ of $G$ of size $|D'|-|E(G)|$.
Then $S$ corresponds to an \dombarset{a} of $G_{b}$ of size $|D'|-|E(G)|$, by \cref{lemma:dom:iff:cov}.
That means $\ogamma_a(G_b) + |E(G)| \geq \ogamma_{a}(G_{a+b})$.
\end{proof}

\subsection{Domination with Parameter Solution Size}
\label{section:dom:parameterized}

This section settles the parametrized complexity
	of $a$-\DS on $b$-subdivided graphs
	when parameterized by the solution size,
	for every integers $a,b$.
These results are also marked by the color of the cells in \cref{figure:ds:cells}
	and summarized as follows.

\begin{lemma}
\label{a:lemma:ds:param}
\ineditnote{\cref{lemma:ind:param}}
\lemmaTextDsParam
\end{lemma}

The polynomial time solvable cases are already settled by \cref{lemma:ds:p}.
As is well-known, \textsc{DominatingSet},
	or as denoted in this work $\DS_3(\GG_1)$, is \np-hard and \wtwo-hard~\cite{GareyJohnson1979,DowneyF95}.
This also puts all cases with $\frac{a}{b}=2$ and $b$ odd to \np-hard and \wtwo-hard,
	by applying \cref{lemma:ds:sub}.
The cases solvable in \fpt-time rely on bounding the solution size
	by the size of a maximum matching in a graph $G$, denoted as $\nu(G)$;
	similarly as in~\cite{HartmannLW22}.

\begin{lemma}
\label{lemma:ds:leq:matching}
\ineditnote{\cref{lemma:ind:leq:matching}}
$\frac{\nu(G)}{2} \leq \ogamma_{a}(G_b)$
	for every $a,b$ with $\frac{a}{b} < 3$.
\end{lemma}
\begin{proof}
Let $M$ be a maximum matching in the original graph $G$,
	and let $D$ be an \dombarset{a} of the subdivided graph $G_b$.
We use the definition \ref{def:ds:3} of an \dombarset{a},
	which is that 
	for every vertex $u \in G_{2b}$ (hence of the $2$-subdivision of $G_b$),
	there is a vertex $w \in D$ with $d_{G_{2b}}(u,w)\leq a$
	(measured in the $2$-subdivision $G_{2b}$ of $G_b$).
For every edge $\{u,v\} \in M$,
	we define $m(\{u,v\})$ as the vertex $\vv(u,v,b)$ in the graph $G_{2b}$.
We claim that for every vertex $d \in D$,
	there are no three edges $e_1,e_2,e_3 \in M$
	such that $d$ has distance $< \ahalf$ to a vertex in each of $m(e_1), m(e_2), m(e_3)$.
Consider that $d = \vv(u,v,\beta)$ with $\beta \in \{0,\dots,2b\}$.
Then there is an edge $e \in e_1,e_2,e_3$ not incident to $u$ nor $v$.
Then $d$ has distance to $m(e)$ of at least $3b<a$ in $G_{2b}$,
	since $\frac{a}{b}<3$.
Therefore $|M| \leq 2|D|$.
\end{proof}

Since we maximum size of a matching in a graph upper bounds the treewidth,
	we have a win-win situation for our \fpt-algorithm.

\begin{lemma}
For every $a,b$ with $\frac{a}{b} < 3$,
	$a$-$\DS(\GG_b)$ is \fpt for the parameter solution size.
\end{lemma}
\begin{proof}
Let the input be a graph $G$ (defining $G_b$) and an integer $k$ as the solution size asked for.
We determine $\nu(G)$, the size of maximum matching in $G$, in polynomial time.
If $k \leq \frac{\nu(G)}{2}$,
	we may immediately answer `yes', according to \cref{lemma:ds:leq:matching}.
Otherwise $k > \frac{\nu(G)}{2} \geq \frac{\vc(G)}{4}$,
	where $\vc(G)$ is the minimum size of a vertex cover of $G$.
Further, the size of a vertex cover upper bounds the treewidth of $G$.
We may compute a tree decomposition of $G$ in \fpt time~\cite{KorhonenLokshtanov2023},
	which immediately provides a tree decomposition of $G_b$ of same size.
If $a\geq 2$, we may compute a minimum \dombarset{a} in  \fpt time
	for the treewidth as parameter by
	using \cref{lemma:ds:classic} for even $a \geq 4$,
	and \cref{lemma:ds:lit:seth} for $a=2$ and odd $a$.
Else, $a=1$ and we have $V(G_b)$ as a minimum \dombarset{1} of $G_b$.
\end{proof}

We settle the remaining cases by two parameter preserving reductions from
	a colorful version of \textsc{Dominating Set}
	showing \wtwo-hardness,
	the first for $\frac{a}{b} \in (3,4)$, the second for $\frac{a}{b} \geq 4$.

\begin{lemma}
For every $a,b \in \NNN_+$ with $\frac{a}{b} \in [3,4)$,
	$a$-$\DS(\GG_b)$ is \np-hard and \wtwo-hard for the parameter solution size.
\end{lemma}
\begin{proof}
We use the definition \ref{def:ds:3} of a \dombarset{a} $D$ of $G_b$,
	that is $D \subseteq V(G_b)$ is an $a$-dominating set in $G_{2b}$.

We show \wtwo-hardness by a reduction from \textsc{Colorful Dominating Set},
	which is known to be \np-hard and \wtwo-hard~\cite{HartmannLW22}.
The input is a graph $G$ with a partition of the vertex set into \emph{color classes} $V_1,\dots,V_k$
	and an integer $k$.
The task is to output a dominating set $\{u_1,\dots,u_k\}$ of $G$ where $v_i \in V_k$ for $i\in[k]$.
Given a graph $G$ with a partition of $V(G)$ into $V_1,\dots,V_k$ and integer $k$,
	we construct a graph $G' \in \GG_b$ and set $k'=k$,
	such that $G$ has a dominating set $\{u_1,\dots,u_k\}$ with $u_i \in V_i$ for $i \in [k]$
	if and only if $G'$ has a \dombarset{a} of size $k$.
The following construction is possible in polynomial time.

\smallskip

\emph{Construction:}
First, we construct an auxiliary graph $G''$.
Initially, let $V(G'')$ be $V(G)$.
For every color class $i \in [k]$,
	we add vertices $x_i$ and $y_i$ with edge $\{x_i,y_i\}$
	and, for every vertex $u \in V_i$,
	the edges $\{u,x_i\}$ and $\{y_i,u\}$.
For every vertex $u \in V(G)$, we add vertices $u^\star,u^{\star\star}$
	and add the edges of the triangle $u,u^\star,u^{\star\star}$.
Further, for every edge $\{u,v\}\in E(G)$, we add
	cross-edges $\{u,v^\star\}$ and $\{u,v^{\star\star}\}$
	and by symmetry $\{v,u^\star\}$ and $\{v,u^{\star\star}\}$.
Finally, let $G'$ be the $b$-subdivision of $G''$.

\smallskip

To show correctness, for the forward direction,
	consider a dominating set $D=\{u_1,\dots,u_k\}$ of $G$
	where $u_i \in V_k$ for $i\in[k]$.
We claim that $D$ is also an \dombarset{a} of $G'$.
In~$G$, for every vertex $v \in V(G)$, there is a vertex $u \in D$ with $u=v$ or $\{u,v\}\in E(G)$.
Hence in the $2$-subdivision $(G'')_{2}$ vertex $u$ $3$-dominates
	$\{v\} \cup N(v)$ and
	because of the cross-edges also $3$-dominates $\{v^\star\} \cup N(v^\star)$ and $\{v^{\star\star}\} \cup N(v^{\star\star})$.
Similarly, for $i \in [k]$, vertex $u_i$ $3$-dominates $\{x_i,y_i\}\cup N(x_i) \cup N(y_i)$ in $(G'')_{2}$.
Since $a \geq 3b$,
	$D$ is an $a$-dominating set of the $2b$-subdivision $(G'')_{2b}$
	and hence an \dombarset{a} of $G' = (G'')_b$.

For the backward direction,
	consider an \dombarset{a} $D'$ of $G'$ of size $k$.
For every color class $i \in [k]$, let $U_i$ be the set of vertices
	with distance $< 4a$ to the midpoint $\vv(x_i,y_i,b)$ in the graph $G_{2b}$.
Since $\frac{a}{b} < 4$,
	the set $D'$, as an $a$-dominating set of $(G'')_{2b}$,
	contains at least one vertex $U_i$,
	for every color class $i\in [k]$.
Because $U_1,\dots,U_k$ are pairwise disjoint,
	$D'$ contains exactly one vertex $u_i'$ in $U_i$ for every $i\in[k]$.
Let $D \subseteq V(G)$ contain a vertex $u_i$, for every color class $i\in [k]$,
	that has minimum distance to $u_i'$ in $G'$.
Then $D$ contains exactly one vertex $u_i$ from every color class $V_i$ for $i \in [k]$.
It remains to show that $D$ is a dominating set of $G$.
Consider a vertex $u \in V(G)$.
Then the midpoint $\vv(u^\star,u^{\star\star},b)$ is $a$-dominated by some vertex $v'$ in $(G'')_{2b}$.
Since $a < 4b$, vertex $v'$ has distance $<b$ to a vertex $v \in V(G)$ in the graph $G'$.
Thus $v \in D$ and $v$ dominates $u$ in $G$.
\end{proof}

\begin{lemma}
For every $a,b \in \NNN_+$ with $\frac{a}{b} \geq 4$,
	$a$-$\DS(\GG_b)$ is \np-hard and \wtwo-hard for the parameter solution size.
\end{lemma}
\begin{proof}
Given a graph $G$ with a partition of $V(G)$ into $V_1,\dots,V_k$ and integer $k$,
	we construct a graph $G' \in \GG_b$ and set $k'=k$,
	such that $G$ has a dominating set of size $k$,
	if and only if $G'$ has an \dombarset{a} of size $k$.
The following construction is possible in polynomial time.

\emph{Construction:}
First, we construct an auxiliary graph $G''$.
Initially, let $V(G'')$ be $V(G)$.
For every vertex $u \in V(G)$,
	a new vertex $u_1$ and
	a $u,u_1$-path of length $\floor{\frac{a}{2b}}-1$.
If $\floor{\frac{a}{b}}$ is odd, we additionally add a vertex $u_2$
	and make it adjacent to $\{u_1\}\cup N(u_1)$.
Finally, let $G'$ be the $b$-subdivision of $G''$.

To show correctness, for the forward direction,
	consider a dominating set $D$ of $G$ of size $k$.
We claim that $D$ is also an $a$-dominating set of the $2b$-subdivision $(G'')_{2b}$
	and hence an \dombarset{a} of $G' = (G'')_b$.
Because $\frac{a}{b}\geq 3$, every vertex $\vv(u,v,\beta)$ for $\{u,v\}\in E(G)$ and $\beta \in \{0,\dots,2b\}$
	is $a$-dominated by $D$ in $(G'')_{2b}$.
In case $\floor{\frac{a}{b}}$ is even, it remains to show that $u_1$ for every $u \in V(G)$
	is $a$-dominated by $D$ in $G_{2b}$.
Indeed, since $u$ is dominated by $v\in \{u\}\cup N(u)$ in $G$,
	$v$ has distance to $u_1$ of at most $2b + (\floor{\frac{a}{2b}}-1)2b \leq a$ in $(G'')_{2b}$.
Else, in case $\floor{\frac{a}{b}}$ is odd, it remains to show that $\vv(u_1,u_2,b)$
	for $u \in E(G)$ is $a$-dominated by $D$ in~$G_{2b}$.
Indeed, since $u$ is dominated by $v\in \{u\}\cup N(u)$ in $G$,
	$v$ has distance to $\vv(u_1,u_2,b)$ of at most 
	$3b + (\floor{\frac{a}{2b}}-1)2b \leq 3b + (\frac{a}{2b}-\frac{1}{2}-1)2b = a$ in $(G'')_{2b}$.

For the backward direction,
	consider an \dombarset{a} $D'$ of the $b$-subdivision $G' = (G'')_b$,
	hence an $a$-dominating set $D'\subseteq V(G_b)$ of the $2b$-subdivision $(G'')_{2b}$.
For every vertex $u' \in D'$, let $u$ be a closest vertex in $V(G)$,
	and let $D = \{ u \mid u' \in D'\}$.
We claim that $D$ is a dominating set of $G$.
First consider that $\floor{\frac{a}{b}}$ is even.
Then for every vertex $u\in V(G)$, vertex $u_1$ in $(G'')_{2b}$
	is $a$-dominated by a vertex whose closest vertex in $V(G)$ is $u$
	or by a vertex $v'$ with distance to $u$ of at most
	$a - (\floor{\frac{a}{2b}}-1)2b < a - ((\frac{a}{2b}+\frac{1}{2})2b) = b$.
In the former case, $u \in D$.
In the latter case, $v'$ has distance $<b$ to a vertex $v \in N(u)$,
	and hence $v \in D$ dominates $u$.
Now consider that $\floor{\frac{a}{b}}$ is odd.
Then for every vertex $u\in V(G)$, vertex $u_1$ in $(G'')_{2b}$
	is $a$-dominated by a vertex whose closest vertex in $V(G)$ is $u$
	or by a vertex $v'$ with distance to $u$ of at most
	$a - (\floor{\frac{a}{2b}}-1)2b - b = b$.
In the former case, $u \in D$.
In the latter case, $v'$ has distance $<b$ to a vertex $v \in N(u)$,
	and hence $v \in D$ dominates $u$.
\end{proof}

\subsection{Covering for Irrational Distances}
\label{section:cov:irrational}

This section derives the following result, analogously to \cref{lemma:disp:wone:pw}.
The main task to give another initial reduction introducing a leeway in the solutions size,
	analogous to \cref{lemma:clique:leq:disp}.
Then the core tools, subdividing and translation, exists in almost the same format for covering.

\begin{theorem}[\cref{lemma:cover:irrational} restated]
\label{a:lemma:cover:irrational}
\ineditnote{\cref{lemma:disp:wone:pw}.}
\lemmaTextCoverIrrational
\end{theorem}

Our starting point is \GridTiling.
The input are $k$ by $k$ `tiles' in form of
	a subset of tuples $S_{i,j} \subseteq [n]^2$ for every $i,j \in [k]$.
The task is to select a tuple for every tile
	such that corresponding entries for each column and each row are non-decreasing.

\problemdefSimple{\GridTiling}{
Integers $n,k$ and $\SSS$, a list of tuples $S_{i,j} \subseteq [n] \times [n]$
	for $i,j \in [k]$.}{
Are there tuples $(s_{i,j},t_{i,j}) \in S_{i,j}$, for every $i,j \in [k]$,
	such that $s_{i,j}\leq s_{i+1,j}$ for $i,i+1,j \in [k]$
	and $t_{i,j}\leq t_{i,j+1}$ for $i,j,j+1 \in [k]$?
	}

\GridTiling is \wone-hard for parameter $k$ and, assuming \ETH,
	has no $f(k) \cdot n^{o(k)}$ time algorithm, for any computable function $f$, see~\cite{bookParameterized}.
Our first reduction builds upon the reduction of Feldmann et al.~\cite{FeldmannM20}.
We adapt it to work for unit-distance graphs and to allow a leeway in the covering distance.

\begin{theorem}
\label{lemma:gt:leq:cover}
\ineditnote{\cref{lemma:clique:leq:disp}.}
There is a polynomial time reduction that,
	given a \GridTiling-instance $\langle n,k, \SSS \rangle$,
	outputs $\langle G',k'\rangle$, a graph $G'$ of pathwidth $\Oh(k)$ and integer $k'$, such that:
$\langle n,k,\SSS\rangle$ is a yes-instance of \GridTiling
	if and only if $\langle G', k'\rangle$ is a yes-instance of $16n^4$-$\covering$
	if and only if $\langle G', k'\rangle$ is a yes-instance of $(16n^4-1)$-$\covering$.
\end{theorem}
\begin{proof}
We prove the above statement for $32n^4$-\domset and $(32n^4-2)$-\domset
	with an output graph $G'_2$ that is a $2$-subdivision of a graph $G'$.
By \cref{lemma:ds:to:cov}, outputting graph $G'$ instead of $G'_2$
	yields the claimed result for $16n^4$-$\covering$ and $(16n^4-1)$-$\covering$.
Particularly, the pathwidth of $G'$ is at most the pathwidth of $G_2'$.

\medskip

\textit{Construction:}
Let $k'=4k^2$.
Let $a=32n^4$.
We assume in the following the (redundant) inequalities
	$6n^2 + (12n^4 +5 n^2) \leq (a-2)/4$, %
	and $a - 12n^4-8n^2-2\geq 4n^2$ %
	and $a-12n^4-2 > 12n^4$.
Otherwise $n$, and hence the number of distinct inputs,
	is bounded above by a constant
	and our reduction may output a trivial no or yes-instance depending on the input.

We define $G_2'$ as follows.
For every tile, defined by integers $i,j \in [k]^2$, we add a cycle $C_{i,j}$ on vertices
	$c_{i,j}^0,c_{i,j}^1,\dots,c_{i,j}^{4a-8}=c_{i,j}^0$ of length $4a-8$.
Let $\UU$ be the family of sets $U_{i,j}^z = \{ c_{i,j}^{z(a-2)+2}, c_{i,j}^{z(a-2)+12n^4} \}$ for every
	$i,j \in [k]$ and ${z\in\{0,1,2,3\}}$.
For every of the two vertices $c \in U_{i,j}^z$, we add a path
	between $c$ and new vertex $\bar c$.
Let $\bar \UU$ be the family of sets analogous to~$\UU$ but with the bar-version of vertices,
	which is $\bar U_{i,j}^{z} = \{\bar c_{i,j}^{z(a-2)+2}, \bar c_{i,j}^{z(a-2)+12n^4} \}$
	for $i,j \in [k]$ and $z\in\{0,1,2,3\}$.
For every `vertically' neighboring tiles, defined by integers $i,i+1,j \in [k]$,
	we add a path between new vertices $u_{i,j}^0$ and $u_{i+1,j}^2$ of length $a-12n^4-8n^2-2 \geq 0$. %
For every `horizontally' neighboring tiles, defined by integers $i,j,j+1 \in [k]$,
	we add a path between new vertices $u_{i,j}^1$ and $u_{i,j+1}^3$ of length $a-12n^4-8n^2-2$.
Consider $S_{i,j}$ with $m_{i,j}=m\leq n^2$ tuples $(s^1, t^1),\dots, (s^{m},t^{m})$,
	omitting subscripts $i,j$.
For $\mu \in [m]$, add
\begin{itemize}
\item
a path from $c_{i,j}^{12n^2\mu}$ to $u_{i,j}^0$ of length $12n^4 + 4(n^2-s^\mu)$,
\item
a path from $c_{i,j}^{(a-2)+12n^2\mu}$ to $u_{i,j}^1$ of length $12n^4 + 4(n^2-t^\mu)$,
\item
a path from $c_{i,j}^{2(a-2)+12n^2\mu}$ to $u_{i,j}^2$ of length $12n^4 + 4(n^2+s^\mu)$,
\item
a path from $c_{i,j}^{3(a-2)+12n^2\mu}$ to $u_{i,j}^3$ of length $12n^4 + 4(n^2+t^\mu)$,
\end{itemize}
This concludes the construction.
We note that the used vertices of the circle in the first bullet point,
	$c_{i,j}^{12n^2\mu}$ for $\mu \in [m]$,
	lie on the shortest path
	between $\bar c_{i,j}^{2}$ and $\bar c_{i,j}^{12n^4}$.
This is analogously the case for the other bullet points.

It is easy to see that $\langle G_2',k \rangle$ is polynomial-time computable.
Further, the output graph $G_2'$ is indeed a $2$-subdivision of a graph $G'$.
Regarding the pathwidth,
	let $G_{i,j}$ be the graph after removing the vertices $U_{i,j}=\{u^0_{i,j},\dots,u^3_{i,j}\}$.
Then the component $\Gamma_{i,j}$ containing the cycle $C_{i,j}$ in $G_{i,j}$
	has constant pathwidth
	since it is a vertex plus a tree where all degree $\geq 3$ vertices lie on a path.
Let $G''$ be graph resulting from $G'_2$ by, for each $i,j \in [k]$,
	removing $\Gamma_{i,j}$
	and contracting the four vertices $U_{i,j}$ to a single vertex $u_{i,j}$.
Then $G''$ as a subdivision of a $k$ by $k$ grid has a path-decomposition of width $\Oh(k)$.
Then we obtain a path-decomposition of $G'_2$ of width $\Oh(k)$ by the following extension:
We extend each bag containing $u_{i,j}$ to include $U_{i,j}$
	and at the first occurrence of a bag containing $U_{i,j}$
	we insert bags that follow the constant width path-decomposition of $\Gamma_{i,j}$.

\medskip

For the correctness, it suffices to show that
	\forward a grid-tiling solution of $\langle G,k\rangle$ implies a $(32n^4-2)$-dominating set of $G'$;
	and that \backward a $32n^4$-dominating set of $G'$ implies a grid-tiling solution of $\langle G,k\rangle$.

\forward
Let indices $\mu_{i,j}$ for every tile $i,j \in [k]$ define a solution of the \GridTiling-instance.
That is $(s_{i,j}^{\mu_{i,j}},t_{i,j}^{\mu_{i,j}}) = (s_{i,j},t_{i,j}) \in S_{i,j}$ for $i,j \in [k]$
	satisfy $s_{i,j}\leq s_{i+1,j}$ for $i,i+1,j \in [k]$
	and $t_{i,j}\leq t_{i,j+1}$ for $i,j,j+1 \in [k]$.
We claim that $D=\bigcup_{i,j\in[k]} D_{i,j}$
	with $D_{i,j} = \{ c_{i,j}^{z(a-2)+12n\mu_{i,j}} \mid z \in \{0,1,2,3\} \}$
	is a $(32n^4-2)$-dominating set.
Indeed, as for each cycle $C_{i,j}$ the vertices $D_{i,j}$ are consecutively in distance $a-2$,
	they $(a-2)$-dominate the cycle.
Further, $c_{i,j}^{z(a-2)+12n\mu_{i,j}}$ does $(a-2)$-dominate $\bar U_{i,j}^z$
	since its lies on a shortest path, of length $a-2$,
	between these two vertices of $\bar U_{i,j}^z$.
The distance from $D_{i,j}$ to $u_{i,j}^z$ for some $z \in \{0,1,2,3\}$
	is at most $(12n^4 +5 n^2) \leq (a-2)/4$.
The distance of any vertex on a path between $C_{i,j}$ and $u_{i,j}^z$
	is at most $(12n^4 +5 n^2) +6n^2  \leq (a-2)/4$. %
Thus the vertices $D_{i,j}$ do $(a-2)$-dominate every vertex between $C_{i,j}$
	and the vertices $u_{i,j}^z$ for $z \in \{0,1,2,3\}$.

Consider `vertically' connected cycles defined by indices $i,i+1,j \in [k]$.
Then $D_{i,j}$ and $D_{i+1,j}$ have distance
	$12n^4 + 4(n^2 - s^\mu_{i,j}) + (a - 24n^4-8n^2-2) + 12n^4 + 4(n^2 + s_{i+1,j}^{\mu'})
	= a - 2 - s_{i,j} + s_{i+1,j}$.
Since $s_{i,j} \leq s_{i+1,j}$, this distance is at most $a-2$
	and hence $D_{i,j}$ and $D_{i+1,j}$ do $(a-2)$-dominate the `vertical' $u_{i,j}^0$ to $u_{i+1,j}^2$ path.
Analogously, $D_{i,j}$ and $D_{i,j+1}$,
	for indices $i,j,j+1 \in [k]$,
	do $(a-2)$-dominate the `horizontal' $u_{i,j}^1$ to $u_{i,j+1}^3$ path.
In conclusion, $D$ does $(a-2)$-dominate every vertex of~$G'_2$.

\smallskip

\backward
Let $D \subseteq V(G')$ be an $a$-dominating set of $G'_2$ and of size $4k^2$.
We show that there is a solution the original \GridTiling-instance.
We claim that distinct $U,U' \in \bar \UU$
	have pairwise distance of more than $a$.
It suffices to observe that distinct $U,U' \in \UU$
	have pairwise distance of more than $12n^4$.
Distinct $U,U' \in \UU$ from the same cycle $C_{i,j}$
	have distance at least $a-12n^4-2 > 12n^4$.
Further, $U,U' \in \UU$ from distinct cycles
	have distance at least $2\cdot 12n^4+(a-24n^4-4n^2-2)= 64n^4-4n^2+6> 12n^4$.

For each set $U_{i,j}^z \in \UU$,
	an $a$-dominating set $D$ of size $4k^2$ contains exactly one vertex
	in distance $\leq \ahalf$ to both of its vertices.
The vertices $V_{i,j}^z$ satisfying this condition have distance $\leq 2$ to $c_{i,j}^{za-2}, c_{i,j}^{za-1},\dots, c_{i,j}^{za+12n^4+2}$,
	located on the cycle.
Indeed, any vertex not on a cycle $C_{i,j}$ has in its $12n^4-2$ radius
	at most one vertex in this cycle.
We observe that $V_{i,j}^z$ are disjoint for each $i,j\in[k]$ and $z\in\{0,1,2,3\}$,
	and hence $D$ contains exactly one vertex from each set $V_{i,j}^z$.
Now, since each cycle $C$ has length $4a-8$ and no shortcuts
	(i.e., any path connecting two vertices $u,v$ of $C$ has length at least the distance of $u,v$ within the cycle),
	the $4$ vertices of $D \cap C_{i,j}$
	have consecutively distance in the interval $\{a-8,a-7,\dots,a\}$.
More precisely, these $4$ vertices have distance $\leq 2$ to $c_{i,j}^{za+c}, z \in \{0,1,2,3\}$,
	respectively, for some $c=c(i,j) \in \{0,\dots,12n^4\}$ depending on~$i,j$.

Recall that the distance between any two cycles is at least $64n^4-4n^2+6$.
Hence index $c(i,j)$ is within distance $4n^2-6+2 < 6n^2$
	to the index $za+12n^2\mu$ for some $\mu \in [m_{i,j}]$,
	and consequently in distance $4n^2-6+2$ to exactly one $\mu=\mu(i,j) \in [m_{i,j}]$.
We may modify $D$ by setting $c(i,j)$ to $12n^2\mu(i,j)$,
	since the corresponding $4$ vertices still cover $C_{i,j}$
	and cover every vertex outside $C_{i,j}$ which was covered by the original vertices of $D\cap C_{i,j}$.

We claim for indices $i \in [k-1]$ and $j\in[k]$ that $s^{\mu(i,j)} \leq s^{\mu(i+1,j)}$.
Indeed, the distance between the vertices $c_{i,j}^{12n^2\mu(i,j)}$ and $c_{i+1,j}^{12n^2\mu(i+1,j)}$,
	defined by the path via $u^0_{i,j}$ and $u^3_{i,j}$ used in the `vertical' connection,
	is $a-24n^4-8n^2-2 + 2\cdot (12n^4 + 4 n^2) - 4s^{\mu(i,j)}_{i,j} + 4s^{\mu(i+1,j)}_{i+1,j} = a + 4(s^{\mu(i+1,j)}_{i+1,j} - s^{\mu(i,j)}_{i,j}) - 2$.
Assuming $s^{\mu(i,j)} > s^{\mu(i+1,j)}$ hence results in a distance larger than $a$,
	in contradiction to an $a$-dominating set $D$.
Analogously, for $j \in [k-1]$ and $i \in [k]$ we have $t^{\mu(i,j)} \leq t^{\mu(i+1,j)}$.
In conclusion, tuples $(s^{\mu(i,j)},t^{\mu(i,j})\in S_{i,j}$ form a solution
	for the original \GridTiling-instance.
\end{proof}

\newcommand{\coverstar}{cover${}^\star$\xspace}
\newcommand{\coverstarr}[1]{#1\text{-}\textrm{cover}^\star}
\newcommand{\coveringstar}{\textsc{Covering}${}^\star$\xspace}

To simplify the following reduction, we tweak the notion of covering as follows.
For a real~$\delta$, a subset of points $S \subseteq P(G)$
	is a $\delta$-\coverstar of $G$ if it is a $\frac{\delta}{2}$-cover of $G$.
Let $\delta$-\coveringstar be the decision problem, given a graph $G$ and integer $k$,
	to decide whether $\delta$-cover${}^\star(G) \leq k$.
Hence we may restate \cref{lemma:gt:leq:cover}
	with some more leeway as follows.

\begin{corollary}
\label{lemma:gt:leq:coverstar}
There is a polynomial time reduction that,
	given a \GridTiling-instance $\langle n,k, \SSS \rangle$,
	outputs $\langle G',k'\rangle$, a graph $G'$ of pathwidth $\Oh(k)$ and integer $k'$, such that:
$\langle n,k,\SSS\rangle$ is a yes-instance of \GridTiling
	if and only if $\langle G', k'\rangle$ is a yes-instance of $2^8n^4$-\coveringstar
	if and only if $\langle G', k'\rangle$ is a yes-instance of $(2^8n^4-1)$-\coveringstar.
\end{corollary}

Further the subdivision lemma translation lemma \cref{lemma:cover:translate} yields
	the following translation analogous to $\delta$-dispersion.
Hence there is a reduction of $\delta$-covering
	analogous to the reduction of $\delta$-dispersion
	in \cref{lemma:disp:reduction:delta:gamma}.

\begin{observation}
\label{lemma:covstar:translate}
${\tfrac{a}{b}}$-cover${}^\star(G) + |E(G)| = {\tfrac{a}{a+b}}$-cover${}^\star(G)$.
\end{observation}
\begin{proof}
${\tfrac{a}{b}}$-cover${}^\star(G) + |E(G)|
	= {\tfrac{a}{2b}}$-cover$(G) + |E(G)|
	= {\tfrac{a}{2a+2b}}$-cover$(G) = {\tfrac{a}{a+b}}$-cover${}^\star(G)$,
	by applying \cref{lemma:cover:translate}.
\end{proof}

\begin{observation}
\label{lemma:covstar:subdivide}
For every real $\delta > 0$ and integer $c \geq 1$,
$\delta$-cover${}^\star(G) = c\delta$-cover${}^\star(G_b)$.
\end{observation}

\begin{lemma}
\label{lemma:cover:reduction:delta:gamma}
\ineditnote{\cref{lemma:disp:reduction:delta:gamma}.}
Let $a_i, b_i$ be defined by \Cref{i:def:delta} and sequence $(c_i)_{i\geq 1}$.
There is a polynomial time reduction that,
	given integers $c_i,k'$ and a graph $G'$,
	outputs a subdivision $G''$ of $G'$ and integer $k''$, such that:
$\langle G',k'\rangle$ is a yes-instance of $c_i$-\coveringstar,
	if and only if $\langle G'',k''\rangle$ is a yes-instance of $\delta_i$-\coveringstar.
$\langle G',k'\rangle$ is a yes-instance of $(c_i-1)$-\coveringstar,
	if and only if $\langle G'',k''\rangle$ is a yes-instance of $\gamma_i$-\coveringstar.
\end{lemma}
\begin{proof}
The proof follows analogously with the same construction as \cref{lemma:disp:reduction:delta:gamma}.
This is due to whenever the proof uses \cref{lemma:disp:subdivide}
	we may use \cref{lemma:covstar:subdivide} instead,
	and whenever the proof uses \cref{lemma:disp:translate}
	we may use \cref{lemma:covstar:translate}.
\end{proof}

\begin{proof}[Proof of \cref{a:lemma:cover:irrational}]
In the following, we show the statement for $\delta$-\coveringstar with $\delta = \frac{\delta'}{2}$.
Then the actual statement follows for $\delta'$-\covering.

Let $\delta$ be defined by integer sequence $c_i = 2^{2^i}$ for $i\geq 1$.
Then $2\delta$, and hence $\delta$, is efficiently comparable by \cref{lemma:disp:computing:a:b}.
Consider a \GridTiling-instance $\langle n,k,\SSS \rangle$ with a universe of size $\hat n$.
Let $i$ be such that $\hat n \leq (c_i/2^8)^{1/4} = 2^{2^{i-2}-2} \eqqcolon n$, hence $2^8n^4=c_i$.
Note that $c_{j+1} = {c_j}^2$, for $j\geq 0$,
	and hence $n \leq {\hat n}^2$.
Thus $n,c_i$ are polynomial time computable and $n,c_i$ are polynomial in $\hat n$.
Equivalently to $\langle \hat n,k,\SSS \rangle$ we may consider
	the instance $\langle n,k,\SSS \rangle$.
We apply the reduction of \cref{lemma:gt:leq:cover} to $\langle n,k,\SSS \rangle$,
	and apply on the output $\langle G',k'\rangle$ the reduction of \cref{lemma:cover:reduction:delta:gamma}
	which outputs $\langle G'',k''\rangle$,
	and finally output $\langle (G'')_2,k''\rangle$.

We note that the reductions of \cref{lemma:gt:leq:cover} and \cref{lemma:cover:reduction:delta:gamma}
	are polynomial time computable.
The former outputs a graph of pathwidth $\Oh(k)$,
	the latter does not change the pathwidth up to subdividing.
Hence the overall algorithm outputs a graph $G''$ of pathwidth $\Oh(k)$.

By \cref{lemma:gt:leq:cover} and \cref{lemma:cover:reduction:delta:gamma},
	$\langle G,k\rangle$ is a yes-instance of \GridTiling,
	if and only if $\langle G'',k''\rangle$ is a yes-instance of $\delta_i$-\coveringstar,
	if and only if $\langle G'',k''\rangle$ is a yes-instance of $\gamma_i$-\coveringstar.
By \cref{lemma:disp:gamma:delta} we have $\gamma_i < \delta < \delta_i$
	and hence the covering numbers satisfy
	$\coverstarr{\gamma_i}(G'')=\coverstarr{\delta}(G'')=\coverstarr{\delta}(G'')$. 	
Thus the output $\langle G'',k'' \rangle$ is a yes-instance of $\delta$-\coveringstar,
	if and only if $\langle G,k \rangle$ is a yes-instance of \GridTiling.

Finally we use that $\langle G'',k'' \rangle$ is a yes-instance of $\delta$-\coveringstar
	if and only if $\langle (G'')_2,k'' \rangle$ is a yes-instance of $\delta$-\covering,
	hence if and only if $\langle G,k \rangle$ is a yes-instance of \GridTiling.

Now since \GridTiling is \wone-hard parameterized by $k$,
	also $\frac{\delta}{2}$-\covering is \wone-hard parameterized by the pathwidth of the input graph.
For the lower bound under \ETH, assume an $f(\pw(G)) \cdot n^{o(\pw(G))}$ time algorithm for $\delta$-\coveringstar
	for any computable function~$f$.
Then using the above reduction on a \GridTiling-instance
	yields an $f(k) \cdot n^{o(k)}$ time algorithm for \GridTiling, in contradiction to \ETH.
\end{proof}

\subsection{Domination for Rational Distances}
\label{section:ds:treewdith}

This section derives the complete classification of the complexity of
	\dombarset{a} on $b$-subdivided graphs
	as well as $\frac{a}{b}$-\covering
	on graphs of bounded treewidth
	for all values $a,b$.
The following two statements are the technical results needed.
First, we give a lower bound for $a=4$
	on general graphs and a lower bound for even $a\geq 6$ on $2$-subdivided graphs.

\begin{lemma}[\cref{section:ds:ppseth:lb}]
\label{lemma:dom:lb:6}
\label{lemma:ds:ppset:lb}
\ineditnote{\cref{lemma:alpha:sub:lower:bound}}
Assume SETH and an $\varepsilon>0$.
Then, for every even $a\geq 4$,
	 $a$-\DomSet has no $O( (a-\varepsilon)^{\pw(G)} )$ time algorithms,
	 in case $a\geq 6$, even on $2$-subdivided graphs.
\end{lemma}

Further, we give an upper bound for every $a\geq 2$.
Previously, only an upper bound for odd $a$ was known,
	by the work of Borradaile et al.~\cite{BorradaileL16}.

\newcommand{\lemmaTextBoundsDsClassic}{
	For $a \geq 2$,
	given a tree decompostion of width $t$ of an $n$ vertex input graph,
	$a$-$\DS$ can be solved in time $a^{t} \cdot n^{O(1)}$.
}

\begin{theorem}[\cref{section:ds:domination}]
\label{lemma:ds:classic}
\lemmaTextBoundsDsClassic
\end{theorem}

Then the following classification follows.

\begin{theorem}[\cref{lemma:i:ds:tw:domination} and \cref{lemma:i:ds:tw:covering} combined]
\label{lemma:ds:tw:summary}
\ineditnote{\cref{p:lemma:is:tw:summary}}
Let $\aaa,\bbb$ integers with $\gcd(\aaa,\bbb)=c$ and $c \aac=\aaa$ and $c \bbc=\bbb$.
Let $n$ be the number of vertices in the input graph.
Assume \SETH, $\varepsilon>0$, and that a tree decomposition of width $t$ is part of the input.
\begin{itemize}
\item
If $\gcd(\aaa,\bbb)$ is odd:
If $\aac=1$, $\aaa$-$\DS(\GG_{\bbb})$ is in \p,
else
	$\aaa$-$\DS(\GG_{\bbb})$ can be solved in time $a^{t} \cdot n^{O(1)}$
	but not in $(a-\varepsilon)^{\pw(G)} \cdot n^{O(1)}$.
\item
If $\gcd(\aaa,\bbb)$ is even:
If $\aac\in\{1,2\}$, $\aaa$-$\DS(\GG_{\bbb})$ is in \p,
else
	$\aaa$-$\DS(\GG_{\bbb})$ can be solved in time $(2a)^{t} \cdot n^{O(1)}$
	but not $(2a-\varepsilon)^{\pw(G)} \cdot n^{O(1)}$.
\item
$\frac{\aaa}{\bbb}$-\covering for $\aac =1$ is in $\p$;
	if $\aac \geq 2$ and $\bbc$ is odd,
	can be solved in time $(4\aac)^{t} \cdot n^{O(1)}$
	but not in $(4\aac-\varepsilon)^{\pw(G)} \cdot n^{O(1)}$;
	if $\aaa \geq 2$ and $\bbc$ is even,
	can be solved in time $(2\aac)^{t} \cdot n^{O(1)}$
	but not in $(2\aac-\varepsilon)^{\pw(G)} \cdot n^{O(1)}$.
\end{itemize}
\end{theorem}
\begin{proof}
First consider that $c = \gcd(\aaa,\bbb)$ is odd.
Then $c \aac$-$\DS(\GG_{c \bbc})$ is equivalent to $\aac$-$\DS(\GG_{\bbc})$
	by \cref{lemma:ds:sub}.
In case $\aac=1$,
	the minimum \dombarset{1} of any graph $G$ is $|V(G)|$.
In case $\aac \geq 2$, there is an $\aac^{t} \cdot n^{O(1)}$ time algorithm
	using \cref{lemma:ds:classic} for even $\aac \geq 4$
	and \cref{lemma:ds:lit:seth} for $\aac=2$ and odd $\aac \geq 3$.

For the lower bound we use that $y \bbc = 1+x \aac$ for some integers $x,y\in\NNN$.
Further assume \SETH.
Then we know that $\aac$-$\DS(\GG_1)$ has no $(\aac-\varepsilon)^{\pw(G)} \cdot n^{O(1)}$ time algorithm 
	for any $\varepsilon>0$,
	by \cref{lemma:dom:lb:6} for even $\aac \geq 4$ and \cref{lemma:ds:lit:seth} for $a=2$ and all odd $\aac \geq 
	3$.
\cref{lemma:trans:ogamma} applied $x$ times
	yields that $\aac$-$\DS(\GG_{1+xa})$,
	and equivalently $\aac$-$\DS(\GG_{yb})$,
	has no $(\aac-\varepsilon)^{\pw(G)} \cdot n^{O(1)}$ time algorithm
	for any $\varepsilon>0$.
Thus especially $\aac$-$\DS(\GG_{\bbc})$ has no $(\aac-\varepsilon)^{\pw(G)} \cdot n^{O(1)}$ time algorithm
	for any $\varepsilon>0$.
This settles the cases for \dombarsetvoid with odd $\gcd(\aaa,\bbb)$.

Consider that $c = \gcd(\aaa,\bbb)$ is even,
	hence we have $c=2\hat c$ for some integer $\hat c$.
Then $2\hat c \aac$-$\DS(\GG_{2\hat c \bbc})$ is equivalent to $(2 \aac)$-$\DS(\GG_{2 \bbc})$
	by \cref{lemma:ds:sub}, since $2\aac$ and $2\bbc$ are even.
In case $a \in \{1,2\}$, \cref{lemma:ds:p} provides a polynomial time algorithm.
Otherwise, for $a\geq 3$,
	we have that $2a$ is even
	and hence \cref{lemma:ds:classic} provides a $(2a)^{t} \cdot n^{O(1)}$ time algorithm.

For the lower bound we again use that $yb = 1+xa$ for some $x,y\in\NNN$.
Assume \SETH.
We know, for $a\geq 2$, that
	$2a$-$\DS(\GG_2)$
	has no $(2a-\varepsilon)^{\pw(G)} \cdot n^{O(1)}$ time algorithm
	for any $\varepsilon>0$, by \cref{lemma:dom:lb:6}.
By applying \cref{lemma:trans:ogamma} $x$ times,
	$2a$-$\DS(\GG_{2+2ax})$,
	and equivalently $2a$-$\DS(\GG_{2yb})$,
	has no $(2a-\varepsilon)^{\pw(G)} \cdot n^{O(1)}$ time algorithm
	for any $\varepsilon>0$.
Thus especially there is no $(2a-\varepsilon)^{\pw(G)}$ time algorithm for $2a$-$\DS(\GG_{2b})$
	for any $\varepsilon > 0$.
This settles the cases for \dombarsetvoid with even $\gcd(\aaa,\bbb)$.

Finally, $\frac{\aaa}{\bbb}$-\covering is equivalent to $\frac{\aac}{\bbc}$-\covering
	with coprime $\aac,\bbc$
	by definition.
In case $\aac=1$, there is polynomial time algorithm for $\frac{\aac}{\bbc}$-\covering, by \cref{lemma:ds:p}.
Consider that $\aac\geq 2$.
In case $\bbc$ is odd,
	$\frac{\aac}{\bbc}$-\covering is equivalent to $4\aac$-$\DS(\GG_{2\bbc})$, by \cref{lemma:ds:to:cov}.
Then the greatest common devisor of $4\aac$ and $2\bbc$ (as two even numbers with $\aac,\bbc$ coprime) is $2$.
Hence (by the earlier discussion)
	$\frac{\aac}{\bbc}$-\covering has a $(4\aac)^t \cdot n^{O(1)}$ time algorithm
	and, assuming \SETH, there is no $(4\aac-\varepsilon)^t \cdot n^{O(1)}$ time algorithm
	for any $\varepsilon>0$.
Otherwise, if $\bbc$ is even,
	$\frac{\aac}{\bbc}$-\covering is equivalent to $2\aac$-$\DS(\GG_{\bbc})$, by \cref{lemma:ds:to:cov}.
Then the greatest common devisor of $2\aac$ and $\bbc$ (as two even numbers with $\aac,\bbc$ coprime) is again $2$.
Hence
	$\frac{\aac}{\bbc}$-\covering has a $(4\aac)^t \cdot n^{O(1)}$ time algorithm
	and, assuming \SETH, there is no $(4\aac-\varepsilon)^t \cdot n^{O(1)}$ time algorithm
	for any $\varepsilon>0$.
\end{proof}

\subsubsection{Lower Bound under SETH}
\label{section:ds:ppseth:lb}

This section shows the lower bound of \dombarset{a} under \ppseth for even $a \geq 4$
	on graphs of bounded pathwidth,
	and, in case of $a \geq6$, also for $2$-subdivided graphs.

\begin{theorem}
\label{a:lemma:ds:ppset:lb}
Assume \ppseth and an $\varepsilon>0$.
Then, for every even $a\geq 4$,
	 $a$-\DomSet has no $(a-\varepsilon)^{\pw(G)} \cdot n^{O(1)}$ time algorithms,
	 in case $a\geq 6$, even on $2$-subdivided graphs.
\end{theorem}

The idea of our proof origins from Borradaile et al.,
	who gave a lower bound for \textsc{Distance Dominating Set} under SETH
	(hence for $a$-\DomSet for odd $a$)~\cite{BorradaileL16}.
Recently, Lampis introduced the `Primal-Pathwidth \SETH' (\ppseth)
	and showed the same lower bound assuming \ppseth instead~\cite{Lampis2024pwseth}.
This result is stronger, as assuming \ppseth is a weaker prerequisite,
	and, at the same time, allows for a nicer presentation
	by deferring many technicalities to an intermediate
	Constraint Satisfaction Problem (CSP) problem,
	which we also use here.
Our reduction is similar but uses some adaptations to output a $2$-subdivided graph.
For details about the intermediate problem, we refer to \cref{section:is:ppseth:lb}.

As also done in \cref{section:is:ppseth:lb},
	the basic idea is to deploy a long path for every variable $i$
	where an $a$-independent set can at most contain every $a$-th vertex
	and thereby encoding an assignment of variable $i$ to $[a]$.
For every constraint, we attach a gadget to these paths that enforce an encoding satisfying the constraint.
To construct a $2$-subdivided graph,
	we may only add connections to the even position of a long path.
We use two kinds of connections,
	that together can enforce an encoding at even and odd positions of such a long path.
In contrast to the construction for $a$-independent set
	as well as to that of ordinary $a$-dominating set,
	we have to locally repeat these enforcing connecting three times
	to assure our connection properly enforce the encoding at an odd position.

\smallskip

We give a reduction from
	\technicalProblem to $a$-\DomSet
	that outputs a $2$-subdivided graph of pathwidth at most $p+O(a^2 \log p + a^5)$.
Thus, assuming \ppseth and an $(a-\varepsilon)^{\pw(G)} \cdot n^{O(1)}$ time algorithm for an $\varepsilon>0$,
	$a$-\DomSet yields a contradiction to \cref{lemma:lampis:technical}.

\begin{lemma}
\label{lemma:ds:ppseth:reduction}
\ineditnote{\cref{lemma:is:ppseth:reduction}}
Let $a\geq 4$ be even.
There is a polynomial time reduction from
	\technicalProblem to $a$-\DomSet
	that outputs a graph $G$ of pathwidth at most $p + O(a^2 \log p + a^5)$,
	which, if $a \geq 6$, is a $2$-subdivision.
\end{lemma}

We output a graph $G$ that is $2$-subdivision of some graph $G'$.
Since we consider even $a$,
	we may work with a simplified version of the definition \ref{def:ds:1} of an \dombarset{a} $D$.
It suffices to require that at least one end-vertex of each edge has distance at most $\ahalf-1$ to~$D$.
Further, we conveniently reduce to a slightly more general problem,
	which treats some edges as already being dominated.
These edges are specified as those in distance $i$ to an $i$-marked vertex for some $i \in \{1,2\}$.

That is, we ask for an \dombarset{a}
	that does not need to cover the edges within an $i$ radius of an $i$-marked vertex, where $i\in \{1,2\}$.
To be precise, our modified problem is the following.

\problemdefSimple{\textsc{Marked $a$-Walk Dominating Set}, $a\geq 4$ even}{
	An integer $k$, a graph $G$ that is a $2$-subdivision of a graph $G'$,
	a set of $1$-marked vertices in $V(G')$ and a set of $2$-marked vertex in $V(G')$
	that define $E' \subseteq E(G)$
	as the set of edges that are incident to $1$-marked vertex or a neighbor of a $2$-marked vertex.
	}{
Is there a size $\leq k$ set $D\subseteq V(G)$,
	such that
	every edge $\{u,v\}\in E(G) \setminus E'$
	has $d(u,D) \leq \ahalf - 1$ or $d(v,D) \leq \ahalf-1$.
}

This problem is in one-to-one correspondence
	with the $a$-\DomSet on $2$-subdivided graphs by the following construction,
	which only increases the pathwidth by a constant.
We increase the budget by $2$ for every marked vertex.
For $i\in\{1,2\}$ and every $i$-marked vertex $u$,
	we add a path $P_u$ of length $2a-i\geq 6$ starting from $u$ and ending in a sequence of vertices 
	$u_3,u_2,u_1,u_0$.
If $2a-i$ is odd, we add a length $3$ path $P_u'$ from $u_3$ to $u_0$ (hence forming length $6$ cycle).
Then, without loss of generality any \dombarset{a} contains,
	for every $i\in\{1,2\}$ and every $i$-marked vertex $u$,
	the two vertices in distance $\ahalf-i$ and in distance $a+\ahalf-i$ to $u$ on the added path $P_u$.
Then these two vertices imply
	that every edge of $P_u$ is dominated as well as every edge within $i$ hops of $u$.

Finally, we note that the original graph has at most the pathwidth of the constructed graph.
Vice versa, we observe that the constructed graph has the pathwidth of the original graph increased by only a 
constant.
To see this, consider a pathwidth decomposition of the original graph.
Then for every $1$-marked or $2$-marked vertex $u$,
	consider a bag $B_u$ that contains $u$.
Then replace $B_u$ by a sequence of bags
	that additionally include the vertices of a constant width path decomposition
	of the path $P_u$ (or path with loop $P_u,P_u'$) attached to~$u$.

\smallskip	

Now, we are ready for the reduction from \technicalProblem
	to our modified problem, \textsc{Marked $a$-Walk Dominating Set}.
We begin by presenting our reduction for even $a \geq 12$.

\smallskip

\emph{Construction:}
Assume that $a\geq 12$.
We construct a graph $G$ that is a $2$-subdivision of a graph $G'$,
	by adding paths $(p^0,\dots,p^j)$ of some even length $j$
	and connecting them only with vertices of even upper index.
\begin{enumerate}
\item For every variable $i \in V_2$,
	we add a path $(p_i^0,\dots,p_i^{3a})$ of length $3a$  and $1$-mark $p_i^0$.
For every node $\tau \in [t]$ and $z \in \{0,1,2\}$,
	let $P_{i,\tau,z} = P_{i,z} = (p_{i,\tau,z}^{1},\dots,p_{i,\tau,z}^{a})$ refer to the subpath
	$(p_i^{za+1},\dots,p_i^{za+a})$
	(where, as $\tau$ is redundant, we may omit the second index.)
We $1$-mark every vertex of even upper index $\leq \ahalf-1$ of $P_{i,1,0}$.
If $\ahalf$ is even, we also $2$-mark the vertex of upper index $\ahalf-2$.
Analogously, we $1$-mark every vertex of even upper index $\geq \ahalf+2$ of $P_{i,1,2}$.
If $\ahalf$ is odd, we also $2$-mark the vertex of upper index $\ahalf+3$.
(These `strange' indices are due to that $0$ is not part of the domain of $\varphi$.)

\item For every variable $i \in V_1$,
	occurring exactly in some $t_i$ bags $B_{\tau_i},\dots,B_{\tau_i+t_i-1}$, 
	we add a path $p_i^{3a\tau_i},\dots,p_i^{3a(\tau_i+t_i)}$ of length $3at_i$.
For a node $\tau \in \{\tau_i,\dots,\tau_i+t_i-1\}$ and $z \in \{0,1,2\}$,
	let $P_{i,\tau,z} = (p_{i,\tau,z}^1,\dots,p_{i,\tau}^{3a})$ refer to the subpath 
	$(p_i^{(3\tau+z)a},\dots,p_i^{(3\tau+z+1)+a})$.
We $1$-mark every vertex of even upper index $\leq \ahalf-1$ of $P_{i,\tau,0}$.
If $\ahalf$ is even, we also $2$-mark the vertex of upper index $\ahalf-2$.
Analogously, we $1$-mark every vertex of even upper index $\geq \ahalf+2$ of $P_{i,\pi,2}$.
If $\ahalf$ is odd, we also $2$-mark the vertex of upper index $\ahalf+3$.
\item \label{ds:step:main}
We proceed for every constraint $\mu \in [m]$ as follows.
Let $\tau = f(\mu)$ and let $X_\mu = \{v_1,\dots,v_4\}$ be the set of variables of constraint $\mu$.
We add a path of length $\ahalf$
	starting in a vertex $\hat z_\tau$
	and ending in vertices $z_\tau', z_\tau$.
If $\ahalf$ is odd, make $z_\tau$ adjacent to a new vertex $z_\tau''$ and $1$-mark $z_\tau''$.
For every $\sigma \in C_\mu$, the set of satisfying assignments of constraint $\mu$,
	we add a vertex $y_{\tau,\sigma}$.
We add a path of length $2$
	between $\hat z_\tau$ and $y_{\tau,\sigma}$ for every $\sigma \in C_\tau$.
For every $\sigma \in C_\mu$,
	for every variable $i \in X_\mu$ where $\sigma(i)$ is even,
		we add a vertex $x_{\tau,\sigma,i}$
	and, for every variable $i \in X_\mu$ where $\sigma(i)$ is odd,
		we add vertices $x_{\tau,\sigma,i}^{-1}$ and $x_{\tau,\sigma,i}^{1}$.
		
Further, for every distinct $\sigma,\sigma' \in C_\mu$ and variable $i \in X_\mu$,
	we add a path of length $2$
	between $y_{\tau,\sigma'}$
	and each of the added vertices $x_{\tau,\sigma,i}, x_{\tau,\sigma,i}^{-1}, x_{\tau,\sigma,i}^{1}$.
For every $\sigma \in C_\tau$, variable $i \in X_\mu$ and $z \in \{0,1,2\}$,
	we proceed as follows.
\begin{itemize}
\item 
If $\alpha = \sigma(i)$ is even,
	we add a path $Q_{\tau,\sigma,i,z} = (q_0,\dots,q_{a-4})$ of length $a-4$
	between $q_0 = p^\alpha_{i,\tau}$ and $q_{a-4} = x_{\tau,\sigma}^i$.
Further, we $2$-mark every vertex of $Q_{\tau,\sigma,i,z}$ of even positive index $\leq \ahalf-4$.
Finally, we $1$-mark $q_{a/2-3}$ if $\ahalf$ is odd.
As a result, the set $E'$ will contain every edge of the subpath $(q_0,\dots, q_{a/2-2})$.
Let $e_{\tau,\sigma,i,z}$ be the edge $\{ q_{a/2-1}, q_{a/2} \}$
	and let $e_{\tau,\sigma,i,z}'$ be the edge below, namely $\{ q_{a/2-2}, q_{a/2-1} \}$.
(We note that neither of $e_{\tau,\sigma,i,z}$, $e_{\tau,\sigma,i,z}'$
	has an end-vertex in distance $\ahalf-1$ to $y_{\tau,\sigma}$.)

\item
If $\alpha = \sigma(i)$ is odd,
	then for $s \in \{-1,1\}$ we proceed as follows.
We add a path $Q_{\tau,\sigma,i,z}^s = (q_0,\dots,q_{a-4})$ of length $a-4$
	between $q_0=p^{\alpha+s}_{i,\tau}$ and $q_{a-4} = x_{\tau,\sigma}^i$.
Further, we $2$-mark every vertex of $Q_{\tau,\sigma,i,z}^s$ of even positive index $\leq \ahalf-4$.
Moreover, we $1$-mark $q_{a/2-3}$ if $\ahalf$ is odd.
Finally, if $\ahalf$ is even, we $1$-mark $q_{a/2}$,
	and if $\ahalf$ is odd, we $2$-mark $q_{a/2+1}$.
As a result, the set $E'$ will contain every edge of the subpath $(q_0,\dots, q_{a/2-2})$
	and particularly the edge $\{q_{a/2-1},q_{a/2}\}$.
Let $e_{\tau,\sigma,i,z}^s$ be the edge $\{ q_{a/2-2}, q_{a/2-1} \}$.
(We note that $e_{\tau,\sigma,i,z}^{s}$
	has no end-vertex in distance $\ahalf-1$ to $y_{\tau,\sigma}$.)

\end{itemize}
\end{enumerate}
Finally, we set the budget $k$ to $m + 3\sum_{i \in V_1} t_i + 3a|V_2|$.

\smallskip

Let us first bound the pathwidth of the constructed graph $G$.
Consider a path decomposition $(B_1,\dots,B_t)$ of the primal graph of $\varphi$ of width $p$
	and where each bag contains at most $O( a \log p)$ vertices from $V_2$.
For every bag $\tau \in [t]$ and every variable $i \in V_2$,
	we add the subpath $P_i$ to a new bag $B_\tau'$.
Further, for every constraint $\mu \in [m]$ with $\tau = f(\mu)$,
	we add to $B_\tau'$ every vertex occurring in an edge in step \ref{ds:step:main}.
So far every new bag contains at most $O( a^2 \log p + a^{5} )$ vertices,
	as there are at most $a^4$ assignments in $C_\mu$.
For every bag $\tau \in [t]$, let $W_\tau \coloneqq V_1 \cap B_\tau$, which has size at most $p$.
Then let $(B_\tau^1,\dots,B_\tau^{t'})$
	be a straight forward width $p+O(1)$ path decomposition
	containing, for $i \in W_\tau$, the subpath of $P_i$ from vertex $p_{\tau,i,0}^0$ to vertex $p_{\tau,i,2}^{a}$,
	and starting with the bag $B_\tau^1 = \{ p_{\tau,i,0}^0 \mid i \in W_\tau \}$
	and ending with the bag $B_\tau^{t'} = \{ p_{\tau,i,2}^a \mid i \in W_\tau \}$.
We add every vertex of $B_\tau'$ to the bags $B_\tau^1,\dots,B_\tau^{t'}$.
Finally, joining the path decompositions $(B_\tau^1,\dots,B_\tau^{t'})$ for $\tau \in [t]$
	results in a path decomposition of the constructed graph of width $p+O( a^2 \log p + a^{5} )$.

\begin{lemma}
\label{lemma:ds:ppseth:correctness}
There is a monotone decreasing multi-assignment $\sigma$
	that is satisfying for $\varphi$ and that is consistent for every variable in $V_2$,
	if and only if there is a size $\leq k$ \dombarset{a} $D \subseteq V(G)$ for $ E(G) \setminus E'$.
\end{lemma}
\begin{proof}
\forward
Assume that there is a satisfying monotone-decreasing multi-assignment $\sigma$ that is consistent for every variable in 
$V_2$.
We construct a size $k$ subset of vertices $D \subseteq V(G)$
	such that every edge $\{u,v\} \in E(G) \setminus E'$
	has an end-vertex that has distance at most $\ahalf-1$ to $D$.
For every variable $i \in V_2$ and $z \in \{0,1,2\}$,
	we add the vertex $p_{i,1,z}^{\alpha}$ to $D$
	where $\alpha = \sigma(\tau,i)$,
	which forms an assignment that is consistent for every variable in $V_2$.
For every variable $i \in V_1$, node $\tau \in [t]$ and $z \in \{0,1,2\}$,
	we add the vertex $p_{i,\tau,z}^{\alpha}$ where $\alpha = \sigma(\tau,i)$.
Finally, for every constraint $\mu \in [m]$,
	say with $\tau = f(\mu)$, variables $X_\mu$ and constraints $C_\mu$,
	we add the vertex $y_{\tau,\sigma_\mu}$ to $D$ where $\sigma_\mu \in C_\mu$
	is the restriction of $\sigma$ to variables $X_\mu$.
The resulting set has size $k$.

We claim that $D$ is an \dombarset{a} for $E(G)\setminus E'$.
Indeed, for every variable $i \in V_2$,
	since the first $\ahalf-1$ edges and the last $\ahalf-1$ edges of $P_i$ are part of $E'$,
	the whole path is \dombared{D} by $D$.
Similarly, for every variable $i \in V_1$,
	since $\sigma$ is monotone decreasing and the first $\ahalf-1$ edges
	and the last $\ahalf-1$ edges of $P_i$ are part of $E'$,
	the whole path is \dombared{a} by $D$.
Further, every edge other than $e_{\tau,\sigma,i,z}$, $e_{\tau,\sigma,i,z}'$ and $e_{\tau,\sigma,i,z}^s$,
	where $\mu \in [m]$, $\tau = f(\mu)$, $\sigma \in C_\mu$, variable $i \in X_\mu$, $z \in \{0,1,2\}$ and $s \in 
	\{-1,1\}$,
	is \dombared{a} by $D \setminus \bigcup_{i \in V_1 \cup V_2} V(P_i)$.
Consider such an edge $e_{\tau,\sigma,i,z}$
	that is not \dombared{a} by $D \setminus \bigcup_{j \in V_1 \cup V_2} V(P_j)$.
In this case, $y_{\tau,\sigma_\mu} \in D$ with $\sigma \in C_\mu$ for $f(\mu) = \tau$.
That means $D$ contains the vertex $p_{i,\tau}^\alpha$ with $\alpha = \sigma(i)$,
	which has distance at most $\ahalf-1$ to the lower index end-vertex
	of $e_{\tau,\sigma,i,z}$ and $e_{\tau,\sigma,i,z}'$.
Finally, consider such an edge $e_{\tau,\sigma,i,z}^s$
	that is not \dombared{a} by $D \setminus \bigcup_{j \in V_1 \cup V_2} V(P_j)$.
That means $D$ contains the vertex $p_{i,\tau}^\alpha$ with $\alpha = \sigma(i)$,
	which has distance at most $\ahalf-1$ to the lower index end-vertex
	of $e_{\tau,\sigma,i,z}$.

\medskip

\backward
Assume that there is a size $\leq k$ set $D \subseteq V(G)$
	such that every edge in $E(G) \setminus E'$
	has at least one end-vertex in distance at most $\ahalf-1$ to $D$.
We show that there is a monotone-decreasing satisfying multi-assignment $\sigma$
	that is consistent for every variable in $V_2$.
Let $S$ be the set of edges consisting of, for every constraint $\mu \in [m]$
	the edge $\{z_i',z_i\}$,
	and, for every variable $i \in V_1$, node $\tau \in [t]$ and $z \in \{0,1,2\}$,
	the edge $\{p^{a/2}_{i,\tau,z}, p^{a/2+1}_{i,\tau,z}\}$,
	and, for every variable $i \in V_2$ and $z \in \{0,1,2\}$, the edge $\{p^{a/2}_{i,z}, p^{a/2+1}_{i,z}\}$.
We observe that every pair of edges in $S$
	has that their end-vertices have distance at least $a-1$.
Since also $D$ has size at most $k=|S|$,
	every edge in $S$ must be \dombared{a} by a unique vertex of $D$.

We modify $D$ to a vertex set $D'$ that still is an \dombarset{a} of $E(G)\setminus E'$.
Consider a constraint $\mu \in [m]$.
We replace the unique vertex $u \in D$ that \dombares{a} the edge $\{z_\tau', z_\tau\}$,
	with $y$, the unique vertex in $\{ y_{\mu,\sigma} \mid \sigma \in C_\mu \}$
	that has closest distance $u$.
For every such replacement, the resulting set still \dombares{a} $E(G)\setminus E'$
	since every in distance $\ahalf-1$ to $u$
	also have distance $\ahalf-1$ to $y$.
(Particularly $y_{\mu,\sigma}$ dominates more edges than its neighbor towards $\hat z_r$.)
Further, for every variable $i \in V_1$, node $\tau \in [t]$ and $z \in \{0,1,2\}$,
	for the vertex $u$ in distance at most $\ahalf-1$ to $\{p^{a/2}_{i,\tau,z}, p^{a/2+1}_{i,\tau,z}\}$
	let $p$ be the unique vertex on the path $P_{i,\tau,z}$ with minimum distance to $u$.
Since especially the edge $\{p^{a/2}_{i,\tau,z}, p^{a/2+1}_{i,\tau,z}\}$
	is not part of a cycle of length at most $a$,
	the vertex $p$ is unique. %
Further, as the modified set $D$ contains a vertex in $\{ y_{\mu,\sigma} \mid \sigma \in C_\mu \}$,
	we have that $D$ where $u$ is replaced by $y$ still is an \dombarset{a} of $E(G)\setminus E'$.
Hence we replace $u$ with $p$ in $D$.
Let the resulting \dombarset{a} be $D'$.
We have that for every variable $i \in V_2$ and $z \in \{0,1,2\}$,
	path $P_{i,z}$ contains exactly one vertex from $D'$,
	and for every variable $i \in V_1$, every node $\tau \in [t]$ and $z \in \{0,1,2\}$,
	path $P_{i,\tau,z}$ contains exactly one vertex from $D'$.

Now, given $D'$, we state the a monotone satisfying assignment $\sigma$
	that is consistent for the variables $V_2$.
For every variable $i \in V_2$ and every node $\tau \in [t]$,
	we set $\sigma(i,\tau)$ to the upper index
	of a the vertex of $D'$ on path $P_{i,1}$.
By definition, this assignment is consistent for $V_2$.
For every variable $i \in V_1$ and every node $\tau \in [t]$,
	we set $\sigma(i,\tau)$ to the the upper index
	of a the vertex of $D'$ on path $P_{i,\tau,1}$.
This assignment is monotone, since
	the edges of the paths $P_{i,\tau,1}, P_{i,\tau,2}, P_{i,\tau+1,0}, P_{i,\tau,-1}$
	including their connecting edges on the path $P_i$
	may only be \dombared{a} by the $4$ vertices in $D' \cap V(P)$ with $P \in \{P_{i,\tau,1}, P_{i,\tau,2}, 
	P_{i,\tau+1,0}, P_{i,\tau,-1}\}$.
	
Finally, we observe that every constraint $\mu \in [m]$ is satisfied.
Let $\tau = f(\mu)$.
Let $\sigma_\tau$ be the local assignment such that $y_{\tau,\sigma_\tau}$
	is the unique vertex of $\{y_{\tau,\sigma_\tau} \mid \sigma_\tau \in C_\mu\}$ in $D'$.
We claim that $\sigma(i,\tau) = \sigma_\tau$ for every variable $i \in \{1,\dots,4\}$.
Indeed, if $\sigma_\tau(i)$ is even, and assuming that $\sigma(i,\tau) \neq \sigma_\tau$
	then the edge $e_{\tau,\sigma,i,1}$ has neither end-vertex in distance at most $\ahalf-1$ to
	a vertex in $D'$, hence is not \dombared{a}.
Otherwise, if $\sigma_\tau(i)=\alpha$ is odd,
	then, for $z \in \{0,1,2\}$ and $s\in\{-1,1\}$,
	the unique vertex in $D' \cap P_{i}$, which we denote as $v_{z,s}$,
	must have distance $1$ to $p_{i,\tau,z}^{\alpha+s}$
	as otherwise the edge $e_{\tau,\sigma,i,z}^s$
	is not \dombared{a}.
Assuming that the vertices $v_{1,-1}, v_{1,1}$ are distinct,
	then the $3$ paths $P_{i,\tau,0},P_{i,\tau,1},P_{i,\tau,2}$ together contain
	the $4$ distinct vertices $v_{0,1}, v_{1,-1}, v_{1,1}, v_{2,-1} \in D'$.
Hence we have $v_{1,-1} = v_{1,1}$.
The unique vertex of $D' \cap V(P_{i,\tau,1})$
	must be $p_{i,\tau,z}^{\alpha}$,
	and hence we have $\sigma(i,\tau)=\sigma_\tau$.
In conclusion, $\sigma$ satisfies $\varphi$.
\end{proof}

This completes the proof for \cref{lemma:ds:ppseth:correctness},
	and hence for \cref{a:lemma:ds:ppset:lb}, for the case $a \geq 12$.
We proceed to show the statement for $a \in \{4,6,8,10\}$.
First, consider that $a=10$.
We adapt the construction as follows.

\smallskip

\textit{Construction for $a=10$:}
We do the same construction as for $a\geq 12$
	but do not add the paths $Q_{\tau,\sigma,i,z}$, $Q_{\tau,\sigma,i,z}^s$ as described in step \ref{ds:step:main}.
Instead, for every assignment $\sigma \in C_\tau$, variable $i \in X_\mu$ and $z \in \{0,1,2\}$,
	we proceed as follows.
\begin{itemize}
\item 
If $\alpha = \sigma(i)$ is even,
	we add a path $Q_{\tau,\sigma,i,z} = (q_0,\dots,q_{16})$ of length $16$
	between $q_0 = p_{i,\tau}^\alpha$ and $q_{16} = x_{\tau,\sigma}^i$.
We $2$-mark the vertex $q_{12}$.
\item
If $\alpha = \sigma(i)$ is odd, then for $s \in \{-1,1\}$ we proceed as follows.
We add a path $Q_{\tau,\sigma,i,z}^s = (q_0,\dots,q_{{36}})$ of length $36$
	between $q_0=p^{\alpha+s}_{i,\tau}$ and $q_{36} = x_{\tau,\sigma}^i$.
We $2$-mark the vertices $q_{26}$ and $q_{32}$.
\end{itemize}
Finally, we set the budget $k$ to as in the construction for $a\geq 12$
	and, for every assignment $\sigma \in C_\tau,$ variable $i \in X_\mu$ and $z \in \{0,1,2\}$,
	if $\sigma(i)$ is even, plus $1$, and else plus $3$.

\smallskip

We claim that \cref{lemma:ds:ppseth:correctness} also holds for this construction for $a=10$.
We easily see that the construction above outputs a graph
	that asymptotically has the same pathwidth as in the construction for $a\geq 12$.
Instead of repeating the whole proof, we focus on the role of the paths $Q_{\tau,\sigma,i,z}$
	and $Q_{\tau,\sigma,i,z}^s$, $s \in \{-1,1\}$,
	for $\mu \in [m]$ and $\tau = f(\mu)$, $i \in X_\mu$ and $z \in \{0,1,2\}$.
The proof of \cref{lemma:ds:ppseth:correctness} relies on the following.
\begin{itemize}
\item 
For even $\alpha = \sigma(i)$,
	if a \dombarset{a} $D$ of $E(G)\setminus E'$ contains $y_{\tau,\sigma}$,
	then $D$ must contain $p_{i,\tau,z}^\alpha$,
	and otherwise no vertex of $P_{i}$ is required to have 
	that every edge of $Q_{\tau,\sigma,i,z}$ has at least one end-vertex in distance $\ahalf-1$ to $D$.
\item
For odd $\alpha = \sigma(i)$ and $s \in \{-1,1\}$,
	if a \dombarset{a} $D$ of $E(G)\setminus E'$ contains $y_{\tau,\sigma}$,
	then $D$ must contain a vertex in distance $1$ to the vertex $p_{i,\tau,z}^{\alpha+s}$,
	and otherwise no vertex of $P_{i}$ is required to have 
	that every edge of $Q_{\tau,\sigma,i,z}^s$ has at least one end-vertex in distance $\ahalf-1$ to $D$.
\end{itemize}
The alternative construction has an additional budget of $1$ respectively $3$
	for every construction replacing such a path $Q_{\tau,\sigma,i,z}$ and $Q_{\tau,\sigma,i,z}^s$ for $s \in 
	\{-1,1\}$,
	depending on whether $\alpha = \sigma(i)$ is even.
With these budget constraints, we observe that the constructions
	replacing the paths $Q_{\tau,\sigma,i,z}$ and $Q_{\tau,\sigma,i,z}^s$
	have the same properties.
Particularly, for an \dombarset{a} $D$ of $E(G)\setminus E'$,
	in case $\alpha = \sigma(i)$ is even,
		then, if $y_{\tau,\sigma} \in D$, we assume that vertex $q_{10}$ of path $Q_{\tau,\sigma,i,z}$ is in $D$
		and else $q_5$ of path $Q_{\tau,\sigma,i,z}$ is in $D$;
	and in n case $\alpha = \sigma(i)$ is odd,
		then, if $y_{\tau,\sigma} \in D$, we assume that $q_{9},q_{19}$ and $q_{30}$ is in $D$,
		and else $q_5,q_{15}$ and $q_{25}$ is in $D$.
With this observation, the same proof of \cref{lemma:ds:ppseth:correctness} also applies for $a=10$.

\smallskip

The next two cases, $a=8$ and $a=4$, work analogously as the discussion for $a=10$,
	and hence we only provide the construction here.

\smallskip

\textit{Construction for $a=8$:}
We do the same construction as for $a\geq 12$
	but do not add the paths $Q_{\tau,\sigma,i,z}$, $Q_{\tau,\sigma,i,z}^s$ as described in step \ref{ds:step:main}.
Instead, for every assignment $\sigma \in C_\tau$, variable $i \in X_\mu$ and $z \in \{0,1,2\}$,
	we proceed as follows.
\begin{itemize}
\item 
If $\alpha = \sigma(i)$ is even,
	we add a path $Q_{\tau,\sigma,i,z} = (q_0,\dots,q_{38})$ of length $38$
	between $q_0 = p_{i,\tau}^\alpha$ and $q_{38} = x_{\tau,\sigma}^i$.
We $2$-mark the vertices $q_{6}$, $q_{18}$ and
	we $1$-mark the vertices $q_{30}$, $q_{36}$.
\item
If $\alpha = \sigma(i)$ is odd, then for $s \in \{-1,1\}$ we proceed as follows.
We add a path $Q_{\tau,\sigma,i,z}^s = (q_0,\dots,q_{{36}})$ of length $36$
	between $q_0=p^{\alpha+s}_{i,\tau}$ and $q_{36} = x_{\tau,\sigma}^i$.
We $2$-mark the vertex $q_{16}$ and
	we $1$-mark the vertices $q_{28}$, $q_{34}$.
Additionally, we add a pat of length $2$ to the vertex~$q_{10}$.
\end{itemize}
Finally, we set the budget $k$ to as in the construction for $a\geq 12$
	and, for every assignment $\sigma \in C_\tau,$ variable $i \in \{1,\dots,4\}$ and $z \in \{0,1,2\}$,
	plus $4$.

\smallskip

\textit{Construction for $a=6$:}
We do the same construction as for $a\geq 12$
	but do not add the paths $Q_{\tau,\sigma,i,z}$ as described in step \ref{ds:step:main}.
Instead, for every assignment $\sigma \in C_\tau$, variable $i \in X_\mu$ and $z \in \{0,1,2\}$,
	we proceed as follows.
\begin{itemize}
\item 
If $\alpha = \sigma(i)$ is even,
	we add a path $Q_{\tau,\sigma,i,z} = (q_0,\dots,q_{18})$ of length $18$
	between $q_0 = p_{i,\tau}^\alpha$ and $q_{18} = x_{\tau,\sigma}^i$.
We add a path of length $2$ to the vertex $q_{8}$,
	ending in a new vertex $q_{8}'$,
	to the vertex $q_{18}$, ending in a new vertex $q_{18}'$ and to the vertex $q_{10}$.
Then we $1$-mark the vertex $q_{8}'$ and $q_{18}'$.
Further, we add a path of length $4$ between $q_{12}$ and $q_{14}$.
\item
If $\alpha = \sigma(i)$ is odd, then for $s \in \{-1,1\}$ we proceed as follows.
We add a path $Q_{\tau,\sigma,i,z}^s = (q_0,\dots,q_{{12}})$ of length $12$
	between $q_0=p^{\alpha+s}_{i,\tau}$ and $q_{12} = x_{\tau,\sigma}^i$.
We add a path of length $2$ to the vertex $q_{4}$,
	and to the vertex $q_{12}$, ending in a new vertex $q_{12}'$.
Then we $1$-mark the vertex $q_{12}'$.
Further, we add a path of length $4$ between $q_{6}$ and $q_{8}$.
\end{itemize}
Finally, we set the budget $k$ to as in the construction for $a\geq 12$
	and, for every assignment $\sigma \in C_\tau$, variable $i \in X_\mu$ and $z \in \{0,1,2\}$,
	if $\sigma(i)$, plus $3$, and else plus $2$.

\smallskip

Finally, we consider the case that $a=4$.
This time we do not need to output a $2$-subdivided graph.
We reduce to \textsc{Marked $4$-Walk Dominating Set}
	where we allow graphs that are not $2$-subdivisions
	and where every vertex may be $1$-marked.
It easily follows that $4$-\DomSet
	is in one-to-one correspondence to this version of \textsc{Marked $4$-Walk Dominating Set}
	by a reduction that increases the pathwidth only by a constant.
Now, for the most part,
	we do the original construction as if $a=8$ and all variable assignments have twice their value.
That is, we do the same construction as for $a\geq 12$,
	however, pretending that $a=8$ and
	for $\sigma \in C_\mu$ and a constraint $\mu \in [m]$ we 
	instead of $\sigma$ use the mapping $\sigma'$
	where $\sigma'(i)=2\sigma(i)$ for every variable $i\in V_1 \cup V_2$.
Notably, only the vertices $p_{i}^0$ for $i \in V_1 \cup V_2$ become $1$-marked by the constructed graph $G$ so far.
Let $G'$ be such that $G$ is the $2$-subdivision of $G'$.
We output a modified graph $G'$ (hence a distance of $8$ in $G$ constitutes a distance of $4$ in $G'$).
Instead of $1$-marking $p_{i}^0$ for $i \in V_1 \cup V_2$, we remove these vertices.
For every vertex $u$ that was $2$-marked in the construction of $G$,
	we $1$-mark $u$ instead.
As the vertices $p_i^0$ for $i \in V_1 \cup V_2$ can be assumed
	to not be contained in any minimum \dombarset{a} of $E(G)\setminus E'$,
	\cref{a:lemma:ds:ppset:lb} follows for $a=4$ (on general graph)
	analogously as for the case $a\geq 12$.
This completes the proof of \cref{a:lemma:ds:ppset:lb}
	for all $a \geq 4$.

\subsubsection{Dynamic Program for Even Distances}
\label{section:ds:domination}

\newcommand{\state}{f}
\newcommand{\State}{F}

This section derives an algorithm for the following statement.

\begin{theorem}[\cref{lemma:ds:classic} restated]
\label{a:lemma:ds:classic}
\lemmaTextBoundsDsClassic
\end{theorem}

Previously, such an algorithm was only known for \emph{odd $a$}, given by Borradaile et al.~\cite{BorradaileL16}.
We give a dynamic program (DP)
	that computes a minimum \dombarset{a} $D$ of $G$
	by using the definition \ref{def:ds:2}.
That is, a subset $D\subseteq V(G)$ is a \dombarset{a} of graph $G$ if:
\begin{itemize}
\item 
Every vertex $u \in V(G)$ has $d(u,D) \leq \ahalf$,
	and the set vertices $u \in V(G)$ where $d(u,D)=\ahalf$ forms an independent set.
\end{itemize}

Our DP over a tree decomposition uses a labeling of the vertices
	that roughly is the integer distance $< a$ to the nearest known vertex in a potential \dombarset{a}.
Doing so, we have to handle assigned integer distances $<\ahalf$ and $>\ahalf$ differently.
In case that $a$ is even, we also have to consider that the distance is exactly $\ahalf$.

Let $(\TT,\beta)$ be the given tree decomposition of width $t$ of the input graph.
For every bag $x$, let $G_x$ be the graph induced by the union of $\beta(x')$
	over $x'$ that is a descendant of $x$ including $x$ itself.

Bottom-up (towards some root node) we compute for every bag $x$
	and every labeling of $x$ and size $\kappa$,
	whether there is an \dombarset{a} $D$ that realizes the labeling at node $x$ (in the following sense)
	and where $D$ has size $\kappa$ when restricted to the subgraph $G_x$.
A labeling $\state$ of $x$ maps each vertex $u$ of the bag $\beta(x)$
	to a distance $0,1,\dots,a-1$.
A value of $0$ indicates that $u$ is part of the \dombarset{a}.
A \emph{low} value $c \in \{1,\dots, \floor{\ahalf-\half} \}$
	(hence all positive integers smaller than $\ahalf$)
	indicates a distance of $\leq c$ to a seen vertex of 
\dombarset{a},
	i.e., a vertex in the subgraph $G_x$.
And a \emph{high} value $c > \ahalf$ requires a distance of a $\leq a-c$ to a not seen vertex of \dombarset{a},
	i.e., a vertex outside the subgraph $G_x$.
For an assigned $\ahalf$ to $u$, we only require that $u$ is not adjacent to a vertex with another assigned $\ahalf$.
We use that the neighbors of $u$ then imply proper domination of $u$.
Since we may reprocess isolated vertices, $u$ has at least one neighbor.

\newcommand{\A}{A}
\subparagraph{Labeling}
For a bag $x$, recall that $G_x$ be the graph induced by the union of $\beta(y)$
	over all $y$ where $y$ is a descendant of $x$ in $\TT$ of $y$ itself is $x$.
We use $L \coloneqq \{0,1,\dots,a-1\}$ as the set of labels.
A \emph{labeling} of a node $x$ is a mapping $f: \beta(x) \to L$.
A subset $D\subseteq V(G)$ \emph{realizes} a labeling
	$\state$ if
	$D$ realizes $f(u)$
	for every vertex $u \in \beta(x)$, which is
\begin{itemize}
\item \labeltext{$($L0$)$}{def:labeling:0} if $f(u) = 0$, then $u \in D$;
\item \labeltext{$($L1$)$}{def:labeling:1} if $f(u) \in \{1,\dots, \floor{\ahalf-\half}\}$,
	then $f(u) \geq d(u,D \cap V(G_x))$;
\item \labeltext{$($L2$)$}{def:labeling:2} if $f(u) = \ahalf$, then $N(u) \cap f^{-1}(\ahalf)=\emptyset$;
\item \labeltext{$($L3$)$}{def:labeling:3} if $f(u) > \ahalf$, then $a-f(u) \geq d(u, D\cap V(G-G_x) )$. 
\end{itemize}

Notably, a  single set $D$ may realizes more than one labeling.
Further, if $a$ is odd, then the value $\ahalf$ is never assigned.
For a node $x$, an integer $\kappa \in \{0,\dots,n\}$ and a labeling $f$,
	we define a table entry of values in $\{0,1\}$, defined as
$$
\A_x[\kappa,f] \coloneqq
\begin{cases*}
1,&  \text{an \dombarset{a} }$D \subseteq V(G)
	\text{ realizes $f$ and has } |D \cap V(G_x) |=\kappa $. \\
0,& \text{else.}
\end{cases*}
$$
We refer to an \dombarset{a} that causes $\A_x[\kappa,f]$ in the above definition as \emph{witness}.
Conveniently, we say that a labeling $f: \beta(x) \to L$ is \emph{sane} if
\begin{itemize}
\item
$f^{-1}(\ahalf)$ is an independent set in $G$.
\item 
for every vertex $v \in \beta(x)$ with $f(v) >\ahalf$
	there is a neighbor $w \in N(v)$
	where $f(w) \in \{f(v)+1,\dots, a-1\}$ or $w \in V(G-G_x)$.
\end{itemize}
We easily observe that $\A[\kappa,f]=0$ for an in-sane labeling $f$.
Our DP computes the entries $A_x[\kappa',f']$ for integers $\kappa' \in \{0,\dots,n\}$ and
	labelings $f$ of $x$, bottom-up for every node $x$,
	depending on the type of $x$, as follows.
Notably, our DP is always aware of the whole graph $G$
	regarding the distances between vertices.
Introducing and forgetting vertices only influences the currently labeled vertices.

To simplify the DP, we first transform the given tree decomposition to a nice tree decomposition $(T,\beta)$,
	as shown by Bodlaender~\cite{Bodlaender1993}, where
	$T$ is rooted at a degree $1$ node $r$ with $\beta(r)=\emptyset$
	and every node $x \in V(T)$ is of one of the following types:
\begin{itemize}
\item \emph{leaf node} where $x$ has no children and $\beta(x)=\emptyset$,
\item \emph{introduce node} where $x$ has a single child $y$
	with $\beta(x) = \beta(y)\cup\{u^\star\}$ for some vertex $u^\star \notin \beta(y)$,
\item \emph{forget node} where $x$ has a single child $y$
	with $\beta(x) = \beta(y)\setminus\{u^\star\}$ for some vertex $u^\star \in \beta(y)$,
	and
\item \emph{join node} where $x$ has exactly two children $x_1$ and $x_2$ where $\beta(x)=\beta(x_1)=\beta(x_2)$.
\end{itemize}

\subparagraph{Leaf Node}
Assume that $x$ is a leaf node, hence with $\beta(x) = \emptyset$.
Then the only labeling of $x$ is the empty labeling $f'$.
Clearly, we have $\A_x[\kappa,f']=1$, if and only if $\kappa'=0$.

\subparagraph{Introduce Node}
Let $x$ be an introduce node with child node $y$,
	and with bag $\beta(x) = \beta(y) \cup \{u^\star\}$.
Let $f': \beta(x) \to L$ be a labeling of $x$.
Let $\kappa$ be an integer and
	$\kappa' = \kappa+1$ if $f'(u^\star)=0$, and else $\kappa' = \kappa$.
(We set $A_x[0,f']=0$ in case $f'(u^\star)=0$.)
We set $\A_x[\kappa', f']=1$, if $f' \in \FF(f)$ as defined below
	for some labeling $f: \beta(y) \to L$.
Otherwise we set $\A_x[\kappa', f']=0$.
We have $f' \in \FF(f)$ if $f'$ is sane and:
\begin{itemize}
\item $f'(u^\star)=0$
	or $f'(u^\star) \geq \ahalf$
	or $f'(u^\star) \geq \min_{w \in \beta(y)} d(u^\star,w)+f(w) \geq 1$,
and
\item for every vertex $u \in \beta(y)$, $f'(u) = 0$ if and only if $f(v)=0$, and
$$ f'(u) \;\geq\;
\begin{cases*}
\min\{f(u), d(u,u^\star) \}, & \text{if $f'(u^\star)=0$ and $f(u) + d(u,u^\star) \leq a$,} \\
f(u), & \text{else.}
\end{cases*}$$
\end{itemize}

\begin{lemma}
$\A_x[\kappa', f']=1$; if and only if $\A_y[\kappa, f]=1$ for some $f' \in \FF(f)$,
	and $\kappa'=\kappa+1$ if $f(u^\star)=0$ and else $\kappa'=\kappa$.
\end{lemma}
\begin{proof}
\backward
Let $\A_y[\kappa,f]=1$,
	hence there is a witness $D \subseteq V(G)$ for $\A_y[\kappa,f]=1$,
	which means that $D$ is an \dombarset{a}
	that realizes $f$ and has $|D \cap V(G_y)|=\kappa$.
Let $f' \in \FF(f)$.
If $f'(u^\star)=0$, then let $\kappa' = \kappa+1$ and else $\kappa'=\kappa$.
We claim that $\A_x[\kappa', f']=1$ because of the witness $D' = D \cup V(G-G_x)$.
We have $|D' \cap V(G_x)|=\kappa'$ by the definition of $\kappa'$.
Particularly, if $f'(u^\star)=0$, then indeed $u^\star \in D'$.
Whenever a vertex $u \in \beta(x)$ has $f'(u)=0$,
	then also $f(u)=0$ such that $u \in D \subseteq D'$, as required.
Further, whenever a vertex $u \in \beta(x)$ has $f'(u)=\ahalf$,
	then since $f'$ is sane, $N(u) \cap {f'}^{-1}(\ahalf)=\emptyset$.
Moreover, whenever a vertex $u \in \beta(x)$ has $f'(u)>\ahalf$,
	then since $f'$ is sane,
	either there is a neighbor $w \in N(u) \cap V(G_x)$ with $a > f'(w) > f'(u)$
	or there is a neighbor in $N(u) \setminus V(G_x)$,
	such that inductively we conclude that $a - f'(u) \geq d(u,D' \cap V(G-G_x))$.

For a vertex $u \in \beta(x)$ with ${f'}(u) \in \{1,\dots, \floor{\ahalf-\half} \}$,
	we distinguish whether $u=u^\star$ or $u \in \beta(y)$.
In case $u=u^\star$, hence considering $f'(u^\star) \in \{1,\dots,\floor{\ahalf-\half} \}$, we have
\begin{equation}
	f'(u^\star)
	\geq \min_{w \in \beta(y)} d(u^\star,w)+f(w)
	\overset{\text{\ref{def:labeling:1}}}{\geq} d(u^\star,w) + d(u,D \cap V(G_y))
	\geq d(u^\star,D\cap V(G_x).
\end{equation}

It remains to consider that $u \in \beta(y)$ with $f'(u) \in \{1,\dots,\floor{\ahalf-\half}\}$.
We have $f'(u) \geq f(u)$,
	or we have $f'(u) \geq d(u,u^\star)$ and $f'(u^\star)=0$ and $d(u,u^\star) \leq a - f(u)$.
In the former case, $f'(u) \geq f(u) \geq d(u, D \cap V(G_y)) \geq d(u, D' \cap V(G_x))$, as required.
In the latter case, particularly $u^\star \in D'$,
	and hence $f'(u) \geq d(u,u^\star) = d(u,D' \cap V(G_x))$, as required.
In summary, $f$ realizes $D$ and $f' \in \FF(d)$.

\smallskip

\forward
Let $A_x[\kappa',f']=1$,
	hence there is a witness $D' \subseteq V(G)$ for $A_x[\kappa',f']=1$,
	which means that $D'$ is an \dombarset{a} that realizes $f'$ and has $|D'\cap V(G_x))|=\kappa'$.
If $f'(u^\star)=0$, let $\kappa=\kappa'$, and else let $\kappa=\kappa'-1$.
We claim that $A_y[\kappa,f]=1$ witnessed by the same set $D=D'$
	for some labeling $f \in F$ which we define in the following.
We have $|D \cap V(G_x)| = \kappa$ by the definition of $\kappa$.
Particularly, if $f'(u^\star)=0$, then indeed $u^\star \in D$.
Whenever $f'(u)=0$, we also set $f(u)=0$.
Then $f$ has that if $f(u)=0$, also $u \in D$.
Further, whenever a vertex $u \in \beta(x)$ has $f'(u)=\ahalf$,
	we also set $f(u)=\ahalf$.
Then since $N(u) \cap {f'}^{-1}(\ahalf)=\emptyset$ also $N(u) \cap {f}^{-1}(\ahalf)=\emptyset$.

Next, consider a vertex $u \in \beta(y)$ with $f'(x) > \ahalf$.
Then $D \cap V(G-G_x)$ contains some vertex $w$
	where $a-f'(u) \geq d(u,w)$.
Particularly, $w \in D \cap V(G-G_y)$,
	and we may set $f(u)=f'(u)$.
Clearly, $f'(u)=f(u)$ satisfies the conditions for $f$ and $f'$.
Next, consider a vertex $u \in \beta(y)$ with ${f'}(u) \in \{1,\dots,\floor{\ahalf-\half}\}$.
If $D \cap V(G_x)$ contains a vertex $w$ with $d(u,w)\leq f'(u)$,
	we may set $f(u)=f'(u)$.
Otherwise, we have $d(u,u^\star) \leq f'(u) < \ahalf$ and $u^\star \in D$.
Then by setting $f(u) = a - d(u,u^\star) > \ahalf$
	we satisfy $a - f(u) \geq d(u, u^\star) \geq d(u,D \cap V(G-G_y))$.
Indeed $f'(u)$ may result from $f(u)$
	in the first case as $f'(u)=f(u)$
	and in the second case as $f'(u^\star)=0$ and having $d(u,u^\star) \leq a-f(u)$ by definition of $f(u)$.
In summary, $f'$ realizes $D$ and we have $f' \in \FF(f)$.
\end{proof}

\subparagraph{Forget Node}
Let $x$ be a forget node with child $y$.
Let $\beta(x) = \beta(y) \setminus \{u^\star\}$.
Consider a labeling $f': \beta(y) \to L$.

\begin{lemma}
$\A_x[\kappa,f'] = 1$, if and only if
	$\A_y[\kappa,f] = 1$
	for some sane labeling $f$ that extends $f'$.
\end{lemma}
\begin{proof}
\forward
Consider that $\A_x[\kappa',f'] = 1$ because of witness $D$.
Since we require that $f$ is sane and further $G_x = G_y$,
	we have $\A_y[\kappa,f]$ witnessed by $D$ as well.

\backward
Let $\A_y[\kappa,f] = 1$ because of witness $D$.
We have that $N(u^\star)\subseteq V(G_y)$
	and hence $f$ restricted to $\beta(x)$ is sane.
Since further $G_x = G_y$, we have $\A_x[\kappa,f']=1$ witnessed by $D$.
\end{proof}

\subparagraph{Join Node}
Consider a join node $x$ with children $x_1$ and $x_2$.
Let $\#_0(f) = | f^{-1}(0) |$, which is the
	the number of vertices counted twice
	when adding up the solutions for the subgraphs $G_{x_1}$ and $G_{x_2}$.
Wether there is a solution $D$ for $\A_x[\kappa,f]$
	depends on the existence of solutions for $\A_{x_1}[\kappa_1,f_1]$ and $\A_{x_2}[\kappa_2,f_2]$
	with $\kappa_1 + \kappa_2 = \kappa + \#_0(f)$
	and labelings $f_1,f_2: \beta(x) \to L$ compatible with $f$ in the following sense.
Intuitively, we differentiate between whether $f(u)$ is \emph{high},
	that is $f(u) \in \{\ceil{\ahalf+\half},\dots,a-1\} \cup \{0\}$,
	or \emph{low}, that is $f(u) \in \{1,\dots, \floor{\ahalf-\half}\}$.
If $f(u)$ is high, a realizing dominating set $D$ contains a vertex in distance $\leq a-f(u)$
	introduced in an ancestor node $x'$ of $x$.
As then $x'$ is an ancestor for $x_1$ \emph{and} $x_2$, we require that $f(u)=f_1(u)=f_2(u)$. 
If $f(u)$ is low, then a realizing dominating set $D$
	contains a vertex in distance $\leq f(u)$ in descendant of $x$,
	hence a descendant of $x_1$ \emph{or} $x_2$.
Hence we require that $f_i(u)=f_i(u)$ and $f_{3-i}(u) \leq a-f_i(u)$ for an $i \in \{1,2\}$.

Formally, we define for
	labelings $\state_1,\state_2: \beta(x) \to L$
	where, for every vertex $u \in \beta(x)$,
	$\{\state_1(u),\state_2(u) \} \subseteq \{ c, a-c \}$ for some $c \in \{0,\dots,a-1\}$,
	the combined labeling $\state_1 \oplus \state_2$
	as
\begin{itemize}
\item $\state(u) = \min\{ \state_1(u), \state_2(u) \}$ for every node $u \in \beta(x)$.
\end{itemize}
For convenience, let us denote $f_1(u) \oplus f_2(u)=f(u)$
	if $f_1 \oplus f_2 = f$ when restricted to input $u$.

\begin{lemma}
For integer $\kappa$ and sane labeling $\state$ of a node $x$ with children $x_1$ and $x_2$,
$$ \A_x[\kappa,\state] \;= \max_{\substack{\kappa_1 + \kappa_2 \;=\; \kappa + \#_0(\state), \\
	f_1 \oplus f_2 = f}}
	\A_{x_1}[\kappa_1,\state_1] \cdot \A_{x_2}[\kappa_2,\state_2] $$
\end{lemma}
\begin{proof}
($\leq$)
Assume that $D \subseteq V(G)$ realizes $f$.
We show that there are labelings $f_1$ and $f_2$ where $f_1 \oplus f_2 = f$,
	such that $f_i$ realize $D$ for $i \in \{1,2\}$.
The size constraint holds, since, for every vertex $u \in \beta(u)$,
	we have $f(u)=0$ if and only if $f_1(u)=f_2(u)=0$,
	and hence $D \cap V(G_x)$ has size
	$|D \cap V(G_{x_1})| + |D \cap V(G_{x_2})| - \#_0(f)$.

Consider a vertex $u \in \beta(x)$.
If $f(u)=0$, meaning $u \in D$, we set $f_1(u)=f_2(u)=f(u)$,
	which clearly satisfies $f_1(u) \oplus f_{2}(u) = f(u)$.
Then $D$ realizes $f_i(u)=0$ for $i \in \{1,2\}$.
For every vertex $u$ with $f(u)=\ahalf$,
	we set $f_1(u)=f_2(u)=\ahalf$, which satisfies $f_1(u) \oplus f_{2}(u) = f(u)$.
Then $N(u) \cap f_i^{-1}(\ahalf) = N(u) \cap f^{-1}(\ahalf)=\emptyset$,
	such that $D_i$ realizes $f_i(u)=\ahalf$ for $i \in \{1,2\}$.
For a vertex $u$ with $f(u)\in\{ \ceil{\ahalf+\half},\dots,a\}$,
	we set $f_1(u)=f_2(u)=f(u)$,
	which satisfies $f_i(u) \oplus f_{3-i}(u) = f(u)$.
Then for every $i\in\{1,2\}$,
	we have $d(u, D \cap V(G-G_{x_i})) \leq d(u, D \cap V(G-G_x)) \leq a - f(u)$,
	such that $D$ realizes $f_i(u)=f(u)$.
Finally, consider that $f(u) \in \{1,\dots, \floor{\ahalf- \half} \}$.
That means $f(u) \geq d(u,v)$ for some vertex $v \in D\cap V(G_x)$.
We set $f_i(u)=f(u)$ and $f_{3-i}(u) = a-f(u)$, which satisfies $f_i(u) \oplus f_{3-i}(u) = f(u)$.
Regarding the child $i$,
	we have $d(u,D \cap V(G_{x_i})) \leq d(u,v) = d(u,D \cap V(G_{x_i}) \leq f(u)=f_i(u)$,
	satisfying \ref{def:labeling:1}.
Regarding the other child $3-i$,
	we have $d(u,D \cap V(G-G_{x_{3-i}})) \leq d(u,v) \leq a - (a - f(u)) = a - f_{3-i}(u)$,
	satisfying \ref{def:labeling:3}.
In conclusion, $D$ realizes $f_1$ and $f_2$ with $f_1 \oplus f_2 = f$.

($\geq$)
Assume that there are labelings $f,f_1,f_2: \beta(x) \to L$ with
	$f_1 \oplus f_2 = f$ and $\kappa_1+\kappa_2 = \kappa + \#_0(f)$
	such that $D_{i}$ realizes $f_i$ for $i \in \{1,2\}$.
We show that then $D = V(G-G_x) \cup \bigcup_{i\in\{1,2\}} (D_i \cap V(G_{x_i})) $ realizes $f$.
We recall that $f(u)=0$ if and only if $f_1(u)=f_2(u)=0$,
	and hence $|D\cap V(G_x)|$ has size $|D_1 \cap V(G_{x_1})|+|D_2 \cap V(G_{x_2})| - \#_0(f) = \kappa$.

Consider a vertex $u \in \beta(x)$.
If $f(u)=0$, then $f_1(u)=f_2(u)=0$, and hence $u \in D$ such that $D$ realizes $f(u)$.
Similarly, for every vertex $u \in \beta(x)$,
	we have $f(u)=\ahalf$ if and only if $f_1(u)=f_2(u)=\ahalf$.
It follows that $f^{-1}(\ahalf) = f_1^{-1}(\ahalf)$, and hence that $D$ realizes $f(u)=\ahalf$.
For a vertex $u$ with $f(u) \in \{ \ceil{\ahalf + \half},\dots,a-1\}$,
	then since $f$ is sane, vertex $u$ has either a neighbor $w \in V(G-G_x) \subseteq D$
	or a neighbor $w \in \beta(x)$ with $f(w) \geq f(u)+1$,
	such that inductively it follows that $d(u,D) \leq a-f(u)$.
Finally, we consider that $f(u) \in \{1,\dots, \floor{\ahalf-\half} \}$.
Then for at least one child $i \in \{1,2\}$,
	we have $f_i(u)=f(u)$,
	hence that there is a vertex $v \in D_1 \cap V(G_{x_1})$ with $d(u,v)\leq f_i(u)=f(u)$.
Since $D_1 \cap V(G_{x_1}) \subseteq D \cap V(G_x)$, and hence that $D$ realizes $f(u)$.
In conclusion, $D$ realizes the labeling $f$.
\end{proof}

Iterating over all $f_1, f_2$ with $f_1 \oplus f_2 = f$ is slow.
To be faster, we convolute labels in accordance of the $\oplus$ operation.
Each of the following labels in $\LL$ is now a set,
	the set of labels of $L$ which it replaces.
	$$\LL \;\coloneqq\; \big\{ \{1,a-1\},\{2,a-2\},\dots,\{\ahalf-1,\ahalf+1\} \big\}
		\cup \big\{ \{\ahalf\}, \dots, \{a-1\} \big\} \cup \big\{ \{0\} \big\} .$$
For a labelings $f: X \to L$ and $F: X \to \LL$,
	we conveniently say that $f \in F$ if $f(u) \in F(u)$ for every vertex $u \in \beta(x)$.
Then let us abstract the notion $f_1 \oplus f_2 = f$ to labels $\LL$.

\begin{lemma}
\label{lemma:dp:f:f:in:F}
Let $F: X \to \LL$ and $f_1,f_2: X \to L$.
Then $(f_1 \oplus f_2) \in F$,
	if and only if $f_1,f_2 \in F$.
\end{lemma}
\begin{proof}
\forward
Assume $f_1 \oplus f_2 = f \in F$.
Consider a vertex $u \in X$.
If $f(u) \in \{\ahalf\} \cup \{ \ceil{\ahalf+\half},\dots,a-1\} \cup \{0\}$,
	then also $f_1(u),f_2(u) = f(u) \in \{f(u)\}$.
If $f(u) \in \{1,\dots, \floor{\ahalf-\half}\}$, we have $\{f_1(u),f_2(u)\} \subseteq \{f(u),a-f(u)\}$.
Hence $f_1,f_2 \in F$.

\backward
Assume $f_1,f_2 \in F$.
Consider a vertex $u \in X$.
If $\State(u) = \{c\}$, in which case $c \geq \ahalf$,
	we have $f_1(u)=f_2(u)=c$.
If $\State(u) = \{c,-c\}$,
	we have $\{\state_1(u),\state_2(u)\} \subseteq \{c,-c\}$.
Hence $f_1 \oplus f_2 = f \in F$
	for $f$ where $u \mapsto \min\{f_1(u),f_2(u)\}$ for every $u \in X$.
\end{proof}

Let $\#_0(F) \coloneqq |F^{-1}(\{0\})|$, analogously to $\#_0(f)$.

\begin{observation}
For $f \in F$, we have $\#_0(F) = \#_0(f)$.
\end{observation}

\medskip

Our \emph{convolution table} $\AAA_x$ for a node $x$ has an entry for every integer $\kappa \in [0,n]$
	and labeling $\State: \beta(x) \to \LL$,
	of value at most the constant $2^{a}$,
	defined as
$$
\AAA_x[\kappa, \State] \;\coloneqq\; \sum_{\state \in \State}
	\A_x[\kappa,\state]
. $$

\begin{lemma}
Let $F: X \to \LL$.
Then
$$
\AAA_x[\kappa,F] \;= \sum_{\kappa_1 + \kappa_2 \;=\; \kappa - \#_0(F)}
	\AAA_1[\kappa_1,F] \cdot \AAA_2[\kappa_2,F] .
$$
\end{lemma}
\begin{proof}
\begin{align*}
\AAA_x[\kappa, \State] =& \sum_{\state \in \State} \A_x[\kappa,\state] & \\
	=& \sum_{\substack{\state \in \State}} \;
	\sum_{\substack{\kappa_1 + \kappa_2 = \kappa + \#_0(\state), \\ \state_1 \oplus \state_2 = \state}}
	\; \big( \A_1[\kappa_1,\state_1] \cdot \A_2[\kappa_2,\state_2] \big) & \\
	=&
	\sum_{\substack{\kappa_1 + \kappa_2 = \kappa + \#_0(\State) }} \;
	 \sum_{\substack{ \state \in \State, \state_1 \oplus \state_2 = \state }}
	\;\big( \A_1[\kappa_1,\state_1] \cdot \A_2[\kappa_2,\state_2] \big) & \\
	\overset{}{=}&
	\sum_{\substack{\kappa_1 + \kappa_2 = \kappa + \#_0(\State) }} \;
	 \sum_{\substack{ f_1,f_2 \in F }}
	\; \big( \A_1[\kappa_1,\state_1] \cdot \A_2[\kappa_2,\state_2] \big) & \mid \text{\cref{lemma:dp:f:f:in:F}} \\
	=&
	\sum_{\substack{\kappa_1 + \kappa_2 = \kappa + \#_0(\State) }} \;
	 \sum_{\substack{ \state_1 \in F }} 	\A_1[\kappa_1,\state_1] \cdot
	 \sum_{\substack{ \state_2 \in \State }} 	\A_1[\kappa_2,\state_2] & \\
	=& \sum_{\kappa_1 + \kappa_2 = \kappa + \#_0(F)}
		\AAA_1[\kappa_1,\State] \cdot \AAA_2[\kappa_2,F] .
\end{align*}
\end{proof}

\begin{lemma}
Consider a bag $x$ of size $w$ and $\kappa \in [0,n]$.
Given the entries$\A_x[\kappa,\state]$ for every $\state: \beta(x) \to L$,
	the entries $\AAA_x[\kappa, \state]$ for every $F: \beta(x) \to \LL$,
	can be computed in time $O(2^{w})$.

Vice versa, given the entries $\AAA_x[\kappa, \state]$ for every $F: \beta(x) \to \LL$,
	the entries $\A_x[\kappa,\state]$ for every $\state: \beta(x) \to L$,
	can also be computes in time $O(2^{w})$.
\end{lemma}
\begin{proof}
Let $u_1,\dots,u_{w}$ be an ordering of the vertices of the bag $\beta(x)$.
For every labeling $f: \beta(x) \to L$,
	we define $\AAA_x^0[\kappa,f(u_1),\dots,f(u_{w})]$ as $\A_x[\kappa,f]$.
For $i \in \{1,\dots,\beta(x)\}$,
	for every label $C_1,\dots,C_i \in \LL$ and label $c_{i+1},\dots,c_{w} \in L$
	we define
	$$
	\AAA_x^i[\kappa,C_1,\dots, C_i, c_{i+1}, \dots,c_{w}]
		\;\coloneqq\; \sum_{c_i \in C_i}
		\AAA_x^{i-1}[\kappa,C_1,\dots,C_{i-1}, c_i, \dots, c_{w}]
	.$$
Eventually, we obtain that $\AAA[\kappa,F] = \AAA_x^{w}[\kappa,F(u_1),\dots,F(u_{w}]$
	for every labeling $F: \beta(x) \to \LL$.
We compute $\A_i$ for increasing $i$ in $w$ rounds, each with $2^{w}$ computations
	of values up to the constant $2^a$;
	hence in overall time $O( 2^{w} )$.

For the other direction, we may undo each operation in reverse order.
That is for $i$ from $w$ to $1$, we compute for
	every label $C_1,\dots,C_{i-1} \in \LL$ and label $c_{i},\dots,i_{w} \in L$,
	in case $c_i \in \{1,\dots, \floor{\ahalf-\half}\}$,
\begin{align*}
\AAA_x^{i-1}[\kappa,C_1,\dots,C_{i-1},c_i,\dots,c_{w}]
	\;=&\; \AAA_x^i[\kappa,C_1,\dots,C_{i-1},\{c_i,a-c_i\},c_{i+1},\dots,c_{w}] \\
	& -\AAA_x^i[\kappa,C_1,\dots,C_{i-1},\{c_i\},c_{i+1},\dots,c_{w}] ,
\end{align*}
and in case $c_i \in \{\ahalf\} \cup \{ \ceil{\ahalf+\half},\dots,a-1\} \cup \{0\}$, 
$$\AAA_x^{i-1}[\kappa,C_1,\dots,C_{i-1},c_i,\dots,c_{w}]
	\;=\; \AAA_x^{i}[\kappa,C_1,\dots,C_{i-1},\{c_i\},c_{i+1},\dots,c_{w}] .$$
Eventually, we obtain that $\A_x[\kappa,f] = \AAA_x^0[\kappa,f(u_1),\dots,f(u_{w})]$
Analogously, to the first direction,
	we obtain entries $\A[\kappa,f]$ for all labelings $f:\beta(x)\to L$ in time $O( 2^{w} )$.
\end{proof}

The above correctness of computing the table entries
	for every node type,
	yields that $A_r[k,f]$, where $r$ is the root of $T$ and $f$ is the empty labeling,
	is $1$ if and only if there is an \dombarset{a} of size $k$.
Let the input treewidth decomposition have width~$t$.
For every node $x$, computing the up to $(n+1) \cdot a^{t+1}$ entries of $A_x$
	is possible in time $n^{O(1)} \cdot a^{w} \cdot O(2^{w} n w)$,
	particularly, by computing the join in the convolution tables $\AAA_x$ instead of the original tables $A_x$.
The overall run time then is $a^t \cdot n^{O(1)}$.
This concludes the proof of \cref{a:lemma:ds:classic}.

\newpage
\bibliography{literature}

\begin{thebibliography}{10}

\bibitem{AlaviBLN77}
Yousef Alavi, M.~Behzad, Linda~M. Lesniak-Foster, and E.~A. Nordhaus.
\newblock Total matchings and total coverings of graphs.
\newblock {\em Journal of Graph Theory}, 1(2):135--140, 1977.
\newblock URL: \url{http://dx.doi.org/10.1002/jgt.3190010209}, \href
  {https://doi.org/10.1002/jgt.3190010209} {\path{doi:10.1002/jgt.3190010209}}.

\bibitem{AlaviLWZ1992}
Yousef Alavi, Jiuqiang Liu, Jianfang Wang, and Zhongfu Zhang.
\newblock On total covers of graphs.
\newblock {\em Discret. Math.}, 100(1-3):229--233, 1992.
\newblock \href {https://doi.org/10.1016/0012-365X(92)90643-T}
  {\path{doi:10.1016/0012-365X(92)90643-T}}.

\bibitem{AlvaradoDR2015}
Jos{\'{e}}~D. Alvarado, Simone Dantas, and Dieter Rautenbach.
\newblock Distance k-domination, distance k-guarding, and distance k-vertex
  cover of maximal outerplanar graphs.
\newblock {\em Discret. Appl. Math.}, 194:154--159, 2015.
\newblock URL: \url{https://doi.org/10.1016/j.dam.2015.05.010}, \href
  {https://doi.org/10.1016/J.DAM.2015.05.010}
  {\path{doi:10.1016/J.DAM.2015.05.010}}.

\bibitem{bacso2019}
G{\'a}bor Bacs{\'o}, Daniel Lokshtanov, D{\'a}niel Marx, Marcin Pilipczuk,
  Zsolt Tuza, and Erik~Jan Van~Leeuwen.
\newblock Subexponential-time algorithms for maximum independent set in
  {$P_t$}-free and broom-free graphs.
\newblock {\em Algorithmica}, 81:421--438, 2019.

\bibitem{Bodlaender1993}
Hans~L. Bodlaender.
\newblock A tourist guide through treewidth.
\newblock {\em Acta Cybern.}, 11(1-2):1--21, 1993.
\newblock URL:
  \url{https://cyber.bibl.u-szeged.hu/index.php/actcybern/article/view/3417}.

\bibitem{BorradaileL16}
Glencora Borradaile and Hung Le.
\newblock Optimal dynamic program for r-domination problems over tree
  decompositions.
\newblock In Jiong Guo and Danny Hermelin, editors, {\em 11th International
  Symposium on Parameterized and Exact Computation, {IPEC} 2016, August 24-26,
  2016, Aarhus, Denmark}, volume~63 of {\em LIPIcs}, pages 8:1--8:23. Schloss
  Dagstuhl - Leibniz-Zentrum f{\"{u}}r Informatik, 2016.
\newblock \href {https://doi.org/10.4230/LIPIcs.IPEC.2016.8}
  {\path{doi:10.4230/LIPIcs.IPEC.2016.8}}.

\bibitem{bookParameterized}
Marek Cygan, Fedor~V. Fomin, Lukasz Kowalik, Daniel Lokshtanov, D{\'{a}}niel
  Marx, Marcin Pilipczuk, Michal Pilipczuk, and Saket Saurabh.
\newblock {\em Parameterized Algorithms}.
\newblock Springer, 2015.
\newblock \href {https://doi.org/10.1007/978-3-319-21275-3}
  {\path{doi:10.1007/978-3-319-21275-3}}.

\bibitem{DallardKM2021}
Cl{\'{e}}ment Dallard, Mirza Krbezlija, and Martin Milanic.
\newblock Vertex cover at distance on {H}-free graphs.
\newblock In Paola Flocchini and Lucia Moura, editors, {\em Combinatorial
  Algorithms - 32nd International Workshop, {IWOCA} 2021, Ottawa, ON, Canada,
  July 5-7, 2021, Proceedings}, volume 12757 of {\em Lecture Notes in Computer
  Science}, pages 237--251. Springer, 2021.
\newblock \href {https://doi.org/10.1007/978-3-030-79987-8\_17}
  {\path{doi:10.1007/978-3-030-79987-8\_17}}.

\bibitem{Dearing1974}
Perino~M. Dearing and Richard~L. Francis.
\newblock A minimax location problem on a network.
\newblock {\em Transportation Science}, 8(4):333--343, 1974.

\bibitem{DowneyF95}
Rodney~G. Downey and Michael~R. Fellows.
\newblock Fixed-parameter tractability and completeness {I:} basic results.
\newblock {\em {SIAM} J. Comput.}, 24(4):873--921, 1995.
\newblock \href {https://doi.org/10.1137/S0097539792228228}
  {\path{doi:10.1137/S0097539792228228}}.

\bibitem{DowneyFellows1995}
Rodney~G. Downey and Michael~R. Fellows.
\newblock Fixed-parameter tractability and completeness {II:} on completeness
  for {W[1]}.
\newblock {\em Theor. Comput. Sci.}, 141(1{\&}2):109--131, 1995.
\newblock \href {https://doi.org/10.1016/0304-3975(94)00097-3}
  {\path{doi:10.1016/0304-3975(94)00097-3}}.

\bibitem{DubloisLP21}
Louis Dublois, Michael Lampis, and Vangelis~T. Paschos.
\newblock New algorithms for mixed dominating set.
\newblock {\em Discret. Math. Theor. Comput. Sci.}, 23(1), 2021.
\newblock URL: \url{https://doi.org/10.46298/dmtcs.6824}, \href
  {https://doi.org/10.46298/DMTCS.6824} {\path{doi:10.46298/DMTCS.6824}}.

\bibitem{ErdosM1977}
Paul Erd{\H{o}}s and Amram Meir.
\newblock On total matching numbers and total covering numbers of complementary
  graphs.
\newblock {\em Discret. Math.}, 19(3):229--233, 1977.
\newblock \href {https://doi.org/10.1016/0012-365X(77)90102-9}
  {\path{doi:10.1016/0012-365X(77)90102-9}}.

\bibitem{EtoGM14}
Hiroshi Eto, Fengrui Guo, and Eiji Miyano.
\newblock Distance-\( d \) independent set problems for bipartite and chordal
  graphs.
\newblock {\em J. Comb. Optim.}, 27(1):88--99, 2014.
\newblock \href {https://doi.org/10.1007/s10878-012-9594-4}
  {\path{doi:10.1007/s10878-012-9594-4}}.

\bibitem{EtoILM17}
Hiroshi Eto, Takehiro Ito, Zhilong Liu, and Eiji Miyano.
\newblock Approximation algorithm for the distance-3 independent set problem on
  cubic graphs.
\newblock In Sheung{-}Hung Poon, Md.~Saidur Rahman, and Hsu{-}Chun Yen,
  editors, {\em {WALCOM:} Algorithms and Computation, 11th International
  Conference and Workshops, {WALCOM} 2017, Hsinchu, Taiwan, March 29-31, 2017,
  Proceedings}, volume 10167 of {\em Lecture Notes in Computer Science}, pages
  228--240. Springer, 2017.
\newblock \href {https://doi.org/10.1007/978-3-319-53925-6\_18}
  {\path{doi:10.1007/978-3-319-53925-6\_18}}.

\bibitem{FeldmannM20}
Andreas~E.\ Feldmann and D{\'{a}}niel Marx.
\newblock The parameterized hardness of the k-center problem in transportation
  networks.
\newblock {\em Algorithmica}, 82(7):1989--2005, 2020.
\newblock URL: \url{https://doi.org/10.1007/s00453-020-00683-w}, \href
  {https://doi.org/10.1007/S00453-020-00683-W}
  {\path{doi:10.1007/S00453-020-00683-W}}.

\bibitem{FreiGHHM2024}
Fabian Frei, Ahmed Ghazy, Tim~A. Hartmann, Florian H{\"{o}}rsch, and
  D{\'{a}}niel Marx.
\newblock From chinese postman to salesman and beyond: Shortest tour
  {\(\delta\)}-covering all points on all edges.
\newblock In Juli{\'{a}}n Mestre and Anthony Wirth, editors, {\em 35th
  International Symposium on Algorithms and Computation, {ISAAC} 2024, Sydney,
  Australia, December 8-11, 2024}, LIPIcs, pages 31:1--31:16. Schloss Dagstuhl
  - Leibniz-Zentrum f{\"{u}}r Informatik, 2024.
\newblock URL: \url{https://doi.org/10.4230/LIPIcs.ISAAC.2024.31}, \href
  {https://doi.org/10.4230/LIPICS.ISAAC.2024.31}
  {\path{doi:10.4230/LIPICS.ISAAC.2024.31}}.

\bibitem{GareyJohnson1979}
Michael~R. Garey and David~S. Johnson.
\newblock {\em Computers and Intractability: {A} Guide to the Theory of
  NP-Completeness}.
\newblock W. H. Freeman, 1979.

\bibitem{GrigorievHLW21}
Alexander Grigoriev, Tim~A. Hartmann, Stefan Lendl, and Gerhard~J. Woeginger.
\newblock Dispersing obnoxious facilities on a graph.
\newblock {\em Algorithmica}, 83(6):1734--1749, 2021.
\newblock \href {https://doi.org/10.1007/s00453-021-00800-3}
  {\path{doi:10.1007/s00453-021-00800-3}}.

\bibitem{thesis}
Tim~A. Hartmann.
\newblock {\em {F}acility location on graphs}.
\newblock Dissertation, RWTH Aachen University, Aachen, 2022.
\newblock URL: \url{https://publications.rwth-aachen.de/record/951030}, \href
  {https://doi.org/10.18154/RWTH-2023-01837}
  {\path{doi:10.18154/RWTH-2023-01837}}.

\bibitem{HartmannJanssen2024}
Tim~A. Hartmann and Tom Jan{\ss}en.
\newblock Approximating {\(\delta\)}-covering.
\newblock In Marcin Bienkowski and Matthias Englert, editors, {\em
  Approximation and Online Algorithms - 22nd International Workshop, {WAOA}
  2024, Egham, UK, September 5-6, 2024, Proceedings}, Lecture Notes in Computer
  Science, pages 61--75. Springer, 2024.
\newblock \href {https://doi.org/10.1007/978-3-031-81396-2\_5}
  {\path{doi:10.1007/978-3-031-81396-2\_5}}.

\bibitem{HartmannL22}
Tim~A. Hartmann and Stefan Lendl.
\newblock Dispersing obnoxious facilities on graphs by rounding distances.
\newblock In Stefan Szeider, Robert Ganian, and Alexandra Silva, editors, {\em
  47th International Symposium on Mathematical Foundations of Computer Science,
  {MFCS} 2022, August 22-26, 2022, Vienna, Austria}, volume 241 of {\em
  LIPIcs}, pages 55:1--55:14. Schloss Dagstuhl - Leibniz-Zentrum f{\"{u}}r
  Informatik, 2022.
\newblock URL: \url{https://doi.org/10.4230/LIPIcs.MFCS.2022.55}, \href
  {https://doi.org/10.4230/LIPICS.MFCS.2022.55}
  {\path{doi:10.4230/LIPICS.MFCS.2022.55}}.

\bibitem{HartmannLW22}
Tim~A. Hartmann, Stefan Lendl, and Gerhard~J. Woeginger.
\newblock Continuous facility location on graphs.
\newblock {\em Math. Program.}, 192(1):207--227, 2022.
\newblock \href {https://doi.org/10.1007/s10107-021-01646-x}
  {\path{doi:10.1007/s10107-021-01646-x}}.

\bibitem{JainJPS2017}
Pallavi Jain, Jayakrishnan Madathil, Fahad Panolan, and Abhishek Sahu.
\newblock Mixed dominating set: {A} parameterized perspective.
\newblock In Hans~L. Bodlaender and Gerhard~J. Woeginger, editors, {\em
  Graph-Theoretic Concepts in Computer Science - 43rd International Workshop,
  {WG} 2017, Eindhoven, The Netherlands, June 21-23, 2017, Revised Selected
  Papers}, volume 10520 of {\em Lecture Notes in Computer Science}, pages
  330--343. Springer, 2017.
\newblock \href {https://doi.org/10.1007/978-3-319-68705-6\_25}
  {\path{doi:10.1007/978-3-319-68705-6\_25}}.

\bibitem{KatsikarelisLP23}
Ioannis Katsikarelis, Michael Lampis, and Vangelis~T. Paschos.
\newblock Improved (in-)approximability bounds for d-scattered set.
\newblock {\em J. Graph Algorithms Appl.}, 27(3):219--238, 2023.
\newblock URL: \url{https://doi.org/10.7155/jgaa.00621}, \href
  {https://doi.org/10.7155/JGAA.00621} {\path{doi:10.7155/JGAA.00621}}.

\bibitem{KatsikarelisLP2022}
Ioannis Katsikarelis, Michael Lampis, and Vangelis~Th. Paschos.
\newblock Structurally parameterized d-scattered set.
\newblock {\em Discret. Appl. Math.}, 308:168--186, 2022.
\newblock \href {https://doi.org/10.1016/j.dam.2020.03.052}
  {\path{doi:10.1016/j.dam.2020.03.052}}.

\bibitem{KorhonenLokshtanov2023}
Tuukka Korhonen and Daniel Lokshtanov.
\newblock An improved parameterized algorithm for treewidth.
\newblock In Barna Saha and Rocco~A. Servedio, editors, {\em Proceedings of the
  55th Annual {ACM} Symposium on Theory of Computing, {STOC} 2023, Orlando, FL,
  USA, June 20-23, 2023}, pages 528--541. {ACM}, 2023.
\newblock \href {https://doi.org/10.1145/3564246.3585245}
  {\path{doi:10.1145/3564246.3585245}}.

\bibitem{Lampis2024pwseth}
Michael Lampis.
\newblock The primal pathwidth {SETH}.
\newblock In Yossi Azar and Debmalya Panigrahi, editors, {\em Proceedings of
  the 2025 Annual {ACM-SIAM} Symposium on Discrete Algorithms, {SODA} 2025, New
  Orleans, LA, USA, January 12-15, 2025}, pages 1494--1564. {SIAM}, 2025.
\newblock \href {https://doi.org/10.1137/1.9781611978322.47}
  {\path{doi:10.1137/1.9781611978322.47}}.

\bibitem{lewis2010vertex}
Jason Lewis, Stephen~T. Hedetniemi, Teresa~W. Haynes, and Gerd~H. Fricke.
\newblock Vertex-edge domination.
\newblock {\em Utilitas mathematica}, 81:193--213, 2010.

\bibitem{LokshtanovMS18}
Daniel Lokshtanov, D{\'{a}}niel Marx, and Saket Saurabh.
\newblock Known algorithms on graphs of bounded treewidth are probably optimal.
\newblock {\em {ACM} Trans. Algorithms}, 14(2):13:1--13:30, 2018.
\newblock \href {https://doi.org/10.1145/3170442} {\path{doi:10.1145/3170442}}.

\bibitem{MadathilPSS2019}
Jayakrishnan Madathil, Fahad Panolan, Abhishek Sahu, and Saket Saurabh.
\newblock On the complexity of mixed dominating set.
\newblock In Ren{\'{e}} van Bevern and Gregory Kucherov, editors, {\em Computer
  Science - Theory and Applications - 14th International Computer Science
  Symposium in Russia, {CSR} 2019, Novosibirsk, Russia, July 1-5, 2019,
  Proceedings}, volume 11532 of {\em Lecture Notes in Computer Science}, pages
  262--274. Springer, 2019.
\newblock \href {https://doi.org/10.1007/978-3-030-19955-5\_23}
  {\path{doi:10.1007/978-3-030-19955-5\_23}}.

\bibitem{Majumdar1992}
Aniket Majumdar.
\newblock {\em Neighborhood hypergraphs: A framework for covering and packing
  parameters in graphs.}
\newblock Dissertation, Clemson University, 1992.

\bibitem{MegiddoTamir1983}
Nimrod Megiddo and Arie Tamir.
\newblock New results on the complexity of $p$-center problems.
\newblock {\em {SIAM} J. Comput.}, 12(4):751--758, 1983.
\newblock \href {https://doi.org/10.1137/0212051} {\path{doi:10.1137/0212051}}.

\bibitem{Meir1978}
Amram Meir.
\newblock On total covering and matching of graphs.
\newblock {\em J. Comb. Theory {B}}, 24(2):164--168, 1978.
\newblock \href {https://doi.org/10.1016/0095-8956(78)90017-5}
  {\path{doi:10.1016/0095-8956(78)90017-5}}.

\bibitem{MontealegreT16}
Pedro Montealegre and Ioan Todinca.
\newblock On distance-d independent set and other problems in graphs with "few"
  minimal separators.
\newblock In Pinar Heggernes, editor, {\em Graph-Theoretic Concepts in Computer
  Science - 42nd International Workshop, {WG} 2016, Istanbul, Turkey, June
  22-24, 2016, Revised Selected Papers}, volume 9941 of {\em Lecture Notes in
  Computer Science}, pages 183--194, 2016.
\newblock \href {https://doi.org/10.1007/978-3-662-53536-3\_16}
  {\path{doi:10.1007/978-3-662-53536-3\_16}}.

\bibitem{Niedermeier06}
Rolf Niedermeier.
\newblock {\em Invitation to Fixed-Parameter Algorithms}.
\newblock Oxford University Press, 2006.
\newblock \href {https://doi.org/10.1093/ACPROF:OSO/9780198566076.001.0001}
  {\path{doi:10.1093/ACPROF:OSO/9780198566076.001.0001}}.

\bibitem{PeledSun1994}
Uri~N. Peled and Feng Sun.
\newblock Total matchings and total coverings of threshold graphs.
\newblock {\em Discret. Appl. Math.}, 49(1-3):325--330, 1994.
\newblock \href {https://doi.org/10.1016/0166-218X(94)90216-X}
  {\path{doi:10.1016/0166-218X(94)90216-X}}.

\bibitem{PilipczukSiebertz2021}
Michal Pilipczuk and Sebastian Siebertz.
\newblock Kernelization and approximation of distance-r independent sets on
  nowhere dense graphs.
\newblock {\em Eur. J. Comb.}, 94:103309, 2021.
\newblock URL: \url{https://doi.org/10.1016/j.ejc.2021.103309}, \href
  {https://doi.org/10.1016/J.EJC.2021.103309}
  {\path{doi:10.1016/J.EJC.2021.103309}}.

\bibitem{Shier1977}
Douglas~R. Shier.
\newblock A min-max theorem for $p$-center problems on a tree.
\newblock {\em Transportation Science}, 11(3):243--252, 1977.
\newblock URL: \url{http://www.jstor.org/stable/25767877}.

\bibitem{Tamir1987}
Arie Tamir.
\newblock On the solution value of the continuous ${p}$-center location problem
  on a graph.
\newblock {\em Math. Oper. Res.}, 12(2):340--349, 1987.
\newblock \href {https://doi.org/10.1287/moor.12.2.340}
  {\path{doi:10.1287/moor.12.2.340}}.

\bibitem{Tamir1991}
Arie Tamir.
\newblock Obnoxious facility location on graphs.
\newblock {\em {SIAM} J. Discret. Math.}, 4(4):550--567, 1991.
\newblock \href {https://doi.org/10.1137/0404048} {\path{doi:10.1137/0404048}}.

\bibitem{RooijBR09}
Johan M.~M. van Rooij, Hans~L. Bodlaender, and Peter Rossmanith.
\newblock Dynamic programming on tree decompositions using generalised fast
  subset convolution.
\newblock In Amos Fiat and Peter Sanders, editors, {\em Algorithms - {ESA}
  2009, 17th Annual European Symposium, Copenhagen, Denmark, September 7-9,
  2009. Proceedings}, volume 5757 of {\em Lecture Notes in Computer Science},
  pages 566--577. Springer, 2009.
\newblock \href {https://doi.org/10.1007/978-3-642-04128-0\_51}
  {\path{doi:10.1007/978-3-642-04128-0\_51}}.

\bibitem{XiaoS2020}
Mingyu Xiao and Zimo Sheng.
\newblock Improved parameterized algorithms and kernels for mixed domination.
\newblock {\em Theor. Comput. Sci.}, 815:109--120, 2020.
\newblock URL: \url{https://doi.org/10.1016/j.tcs.2020.02.014}, \href
  {https://doi.org/10.1016/J.TCS.2020.02.014}
  {\path{doi:10.1016/J.TCS.2020.02.014}}.

\bibitem{ZhaoKS2011}
Yancai Zhao, Liying Kang, and Moo~Young Sohn.
\newblock The algorithmic complexity of mixed domination in graphs.
\newblock {\em Theor. Comput. Sci.}, 412(22):2387--2392, 2011.
\newblock URL: \url{https://doi.org/10.1016/j.tcs.2011.01.029}, \href
  {https://doi.org/10.1016/J.TCS.2011.01.029}
  {\path{doi:10.1016/J.TCS.2011.01.029}}.

\bibitem{ZiemannZylinski20}
Rados\l{}aw Ziemann and Pawe{\l} {\.Z}yli{\'n}ski.
\newblock Vertex-edge domination in cubic graphs.
\newblock {\em Discret. Math.}, 343(11):112075, 2020.
\newblock URL: \url{https://doi.org/10.1016/j.disc.2020.112075}, \href
  {https://doi.org/10.1016/J.DISC.2020.112075}
  {\path{doi:10.1016/J.DISC.2020.112075}}.

\bibitem{zylinski2019}
Pawe{\l} {\.Z}yli{\'n}ski.
\newblock Vertex-edge domination in graphs.
\newblock {\em Aequationes mathematicae}, 93(4):735--742, 2019.

\end{thebibliography}

\end{document}